\documentclass[a4paper,11pt,reqno]{amsart}

\usepackage[margin=1in,includehead,includefoot]{geometry}
\usepackage[utf8]{inputenc}
\usepackage[T1]{fontenc}
\usepackage{lmodern,microtype}
\usepackage[foot]{amsaddr}
\usepackage{hyperref}
\usepackage{url}
\usepackage{booktabs}
\usepackage{mathtools,amssymb,amsthm}
\usepackage[nice]{nicefrac}
\usepackage[foot]{amsaddr}
\usepackage{graphicx, color}
\usepackage{wrapfig}
\usepackage[margin=5mm]{caption}
\usepackage{enumitem}
\usepackage[textsize=footnotesize]{todonotes}
\usepackage{mdframed}
\usepackage{xpatch}
\usepackage{mycommands}

\graphicspath{{./graphics/}}

\newtheorem{theorem}{Theorem}
\newtheorem{lemma}[theorem]{Lemma}


\hypersetup{
  pdftitle={Combinatorial generation via permutation languages. VI. Binary trees},
  pdfauthor={Petr Gregor, Torsten M\"utze, Namrata}
}

\title{Combinatorial generation via permutation languages. \\  VI. Binary trees}

\author{Petr Gregor}
\address[Petr Gregor]{Department of Theoretical Computer Science and Mathematical Logic, Charles University, Prague, Czech Republic}
\email{gregor@ktiml.mff.cuni.cz}

\author{Torsten M\"utze}
\address[Torsten M\"utze]{Department of Computer Science, University of Warwick, United Kingdom \& Department of Theoretical Computer Science and Mathematical Logic, Charles University, Prague, Czech Republic}
\email{torsten.mutze@warwick.ac.uk}

\author{Namrata}
\address[Namrata]{Department of Computer Science, University of Warwick, United Kingdom}
\email{namrata@warwick.ac.uk}

\thanks{An extended abstract of this paper has been accepted for presentation at ISAAC~2023. This work was supported by Czech Science Foundation grant GA~22-15272S. The three authors participated in the workshop `Combinatorics, Algorithms and Geometry' in March 2024, which was funded by German Science Foundation grant~522790373.}

\begin{document}

\begin{abstract}
In this paper we propose a notion of pattern avoidance in binary trees that generalizes the avoidance of contiguous tree patterns studied by Rowland and non-contiguous tree patterns studied by Dairyko, Pudwell, Tyner, and Wynn.
Specifically, we propose algorithms for generating different classes of binary trees that are characterized by avoiding one or more of these generalized patterns.
This is achieved by applying the recent Hartung--Hoang--M\"utze--Williams generation framework, by encoding binary trees via permutations.
In particular, we establish a one-to-one correspondence between tree patterns and certain mesh permutation patterns.
We also conduct a systematic investigation of all tree patterns on at most 5 vertices, and we establish bijections between pattern-avoiding binary trees and other combinatorial objects, in particular pattern-avoiding lattice paths and set partitions.
\end{abstract}

\maketitle

\section{Introduction}
\label{sec:intro}

Pattern avoidance is a central theme in combinatorics and discrete mathematics.
For example, in Ramsey theory one investigates how order arises in large unordered structures such as graphs, hypergraphs, or subsets of the integers.
The concept also arises naturally in algorithmic applications.
For example, Knuth~\cite{MR3077152} showed that the integer sequences that are sortable by one pass through a stack are precisely 231-avoiding permutations.
Pattern-avoiding permutations are a particularly important and heavily studied strand of research, one that comes with its own associated conference `Permutation Patterns', held annually since~2003.
While it may seem that pattern-avoiding permutations are somewhat limited in scope, via suitable bijections they actually encode many objects studied in other branches of combinatorics.
Pattern avoidance has also been studied directly in these other classes of objects, such as trees~\cite{MR2645188,Dotsenko_2011,MR3118904,MR2967227,MR2872462,PSSS14,MR3548807,MR3872611,MR4118862}, set partitions~\cite{MR309747,MR1370819,MR1750890,MR1769065,MR2419765,MR2398831,MR2824451,MR2863229,MR2599721,MR2954657,MR3003349,MR3245891,MR3548808}, lattice paths~\cite{MR2371066,MR3091032,MR3788052,MR4256410}, heaps~\cite{MR3426216}, matchings~\cite{MR3066344}, and rectangulations~\cite{MR4598046}.
In this work, we focus on binary trees, a class of objects that is fundamental within computer science, and also a classical Catalan family.

So far, two different notions of pattern avoidance in binary trees have been studied in the literature.
We consider a binary tree~$T$, which serves as the host tree, and another binary tree~$P$, which serves as the pattern tree.
Rowland~\cite{MR2645188} considered a \defi{contiguous} notion of pattern containment, where $T$ contains~$P$ if $P$ is present as an induced subtree of~$T$; see Figure~\ref{fig:notions}~(a).
He devised an algorithm to compute the generating function for the number of $n$-vertex binary trees that avoid~$P$, and he showed that this generating function is always algebraic.
Dairyko, Pudwell, Tyner, and Wynn~\cite{MR2967227} considered a \defi{non-contiguous} notion of pattern containment, where $T$ contains~$P$ if $P$ is present as a ``minor'' of~$T$; see Figure~\ref{fig:notions}~(b).
They discovered the remarkable phenomenon that for any two distinct $k$-vertex pattern trees~$P$ and~$P'$, the number of $n$-vertex host trees that avoid~$P$ is the same as the number of trees that avoid~$P'$, i.e., $P$ and~$P'$ are \defi{Wilf-equivalent} patterns.
They also obtain the corresponding generating function (which is independent of~$P$, but only depends on $k$ and~$n$).

In this paper, we consider \defi{mixed} tree patterns, which generalize both of the two aforementioned types of tree patterns, by specifying separately for each edge of~$P$ whether it is considered contiguous or non-contiguous, i.e., whether its end vertices in the occurrence of the pattern must be in a parent-child or ancestor-descendant relationship (in the correct direction left/right), respectively; see Figure~\ref{fig:notions}~(c).

Observe that the notions of tree patterns considered in~\cite{MR2645188} and~\cite{MR2967227} are the tree analogues of consecutive~\cite{MR1979785} and classical permutation patterns, respectively.
Our new notion of mixed patterns is the tree analogue of vincular permutation patterns~\cite{MR1758852}, which generalize classical and consecutive permutation patterns.

\begin{figure}
\includegraphics[page=2]{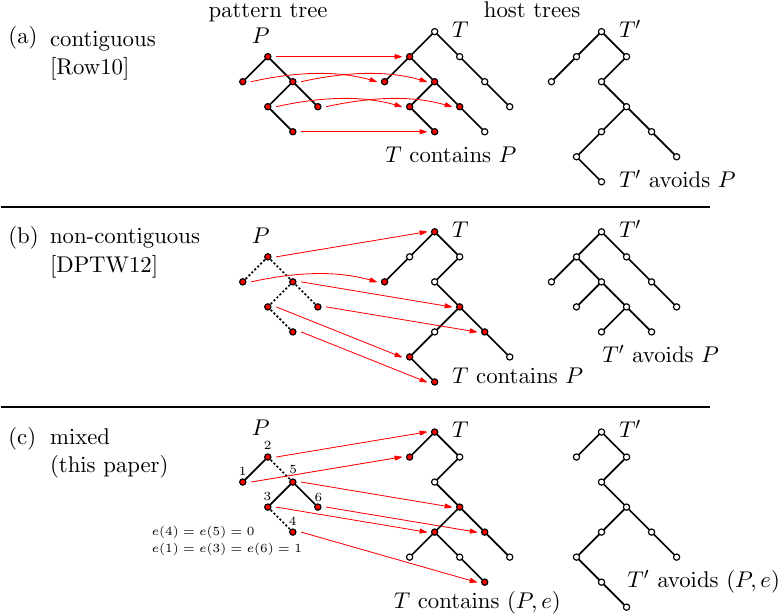}
\caption{Illustration of different notions of pattern containment in binary trees.
Contiguous edges are drawn solid, whereas non-contiguous edges are drawn dotted.}
\label{fig:notions}
\end{figure}

\subsection{The Lucas--Roelants van Baronaigien--Ruskey algorithm}

\begin{wrapfigure}{r}{0.4\textwidth}
\vspace{-6mm}
\includegraphics{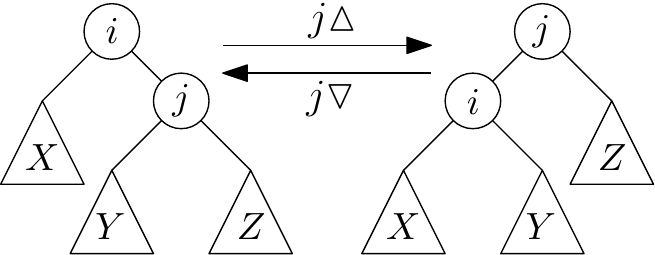}
\caption{Rotation in binary trees.}
\label{fig:rot}
\vspace{-14mm}
\end{wrapfigure}
One of the goals in this paper is to generate different classes of binary trees, i.e., we seek an algorithm that visits every tree from the class exactly once.
Our starting point is a classical result due to Lucas, Roelants van Baronaigien, and Ruskey~\cite{MR1239499}, which asserts that all $n$-vertex binary trees can be generated by tree rotations, i.e., every tree is obtained from its predecessor by a single \defi{tree rotation} operation; see Figures~\ref{fig:rot} and~\ref{fig:LRB}.
The algorithm is an instance of a \defi{combinatorial Gray code}~\cite{MR1491049,MR4649606}, which is a listing of objects such that any two consecutive objects differ in a `small local' change.
The aforementioned Gray code algorithm for binary trees can be implemented in time~$\cO(1)$ per generated tree.

\begin{figure}[h!]
\includegraphics[page=1]{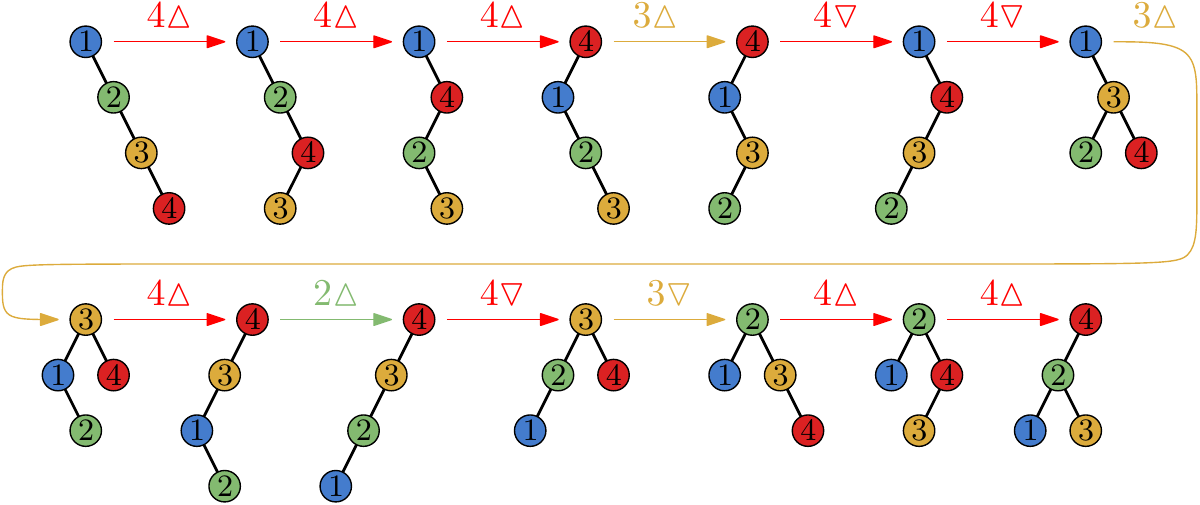}
\caption{The Lucas--Roelants van Baronaigien--Ruskey algorithm to generate all binary trees with $n=4$ vertices by tree rotations.
The vertices are labeled with $1,2,3,4$ according to the search tree property.}
\label{fig:LRB}
\end{figure}

Williams~\cite{MR3126386} discovered a stunningly simple description of the Lucas--Roelants van Baronaigien--Ruskey Gray code for binary trees via the following greedy algorithm, which is based on labeling the vertices with $1,\ldots,n$ according to the search tree property:
Start with the right path, and then repeatedly perform a tree rotation with the largest possible vertex that creates a previously unvisited tree.

\subsection{Our results}

It is well known that binary trees are in bijection with 231-avoiding permutations.
Our first contribution is to generalize this bijection, by establishing a one-to-one correspondence between mixed binary tree patterns and mesh permutation patterns, a generalization of classical permutation patterns introduced by Br{\"a}nd{\'e}n and Claesson~\cite{MR2795782}.
Specifically, we show that $n$-vertex binary trees that avoid a particular (mixed) tree pattern~$P$ are in bijection with 231-avoiding permutations that avoid a corresponding mesh pattern~$\sigma(P)$ (see Theorem~\ref{thm:bijection} below).

This bijection enables us to apply the Hartung--Hoang--M\"utze--Williams generation framework~\cite{MR4391718}, which is based on permutations.
We thus obtain algorithms for efficiently generating different classes of pattern-avoiding binary trees, which work under some mild conditions on the tree pattern(s).
These algorithms are all based on a simple greedy algorithm, which generalizes Williams' algorithm for the Lucas--Roelants van Baronaigien--Ruskey Gray code of binary trees (see Algorithm~S, Algorithm~H, and Theorems~\ref{thm:algoS} and~\ref{thm:algoH}, respectively).
Specifically, instead of tree rotations our algorithms use a more general operation that we refer to as a \defi{slide}.
We implemented our generation algorithm in~C++, and we made it available for download and experimentation on the Combinatorial Object Server~\cite{cos_btree}.

For our new notion of mixed tree patterns, we conduct a systematic investigation of all tree patterns on up to 5 vertices.
This gives rise to many counting sequences, some already present in the OEIS~\cite{oeis} and some new to it, giving rise to several interesting conjectures.
In this work we establish most of these as theorems, by proving bijections between different classes of pattern-avoiding binary trees and other combinatorial objects, in particular pattern-avoiding lattice paths (Section~\ref{sec:motzkin}) and set partitions (Theorem~\ref{thm:staggered}).

This paper is the sixth installment in a series of papers on generating a large variety of combinatorial objects by encoding them in a unified way via permutations.
This algorithmic framework was developed in~\cite{MR4391718} and so far has been applied to generate pattern-avoiding permutations~\cite[Table~1]{MR4391718}, lattice congruences of the weak order on permutations~\cite{MR4344032}, pattern-avoiding rectangulations~\cite{MR4598046}, elimination trees of graphs~\cite{MR4415121}, and acyclic orientations of graphs~\cite{MR4614413}.
The present paper thus further extends the reach of this framework to pattern-avoiding Catalan structures.
For readers familiar with elimination trees, we mention that when the underlying graph is a path with vertices labeled~$1,\ldots,n$, then its elimination trees are precisely all $n$-vertex binary trees.
Very recently, another application of the aforementioned generation framework to derive Gray codes for geometric Catalan structures, specifically staircases and squares, has been presented in~\cite{williams_et_al_23}.

\subsection{Outline of this paper}

In Section~\ref{sec:prelim} we introduce basic notions that will be used throughout the paper.
In Section~\ref{sec:enc-perm} we establish a bijection between binary trees patterns and mesh patterns.
In Section~\ref{sec:generate} we present our algorithms for generating classes of binary trees that are characterized by pattern avoidance.
In Section~\ref{sec:equality} we establish the equality between certain tree patterns that differ in few contiguous or non-contiguous edges.
In Section~\ref{sec:count} we report on our computational results on counting pattern-avoiding binary trees for all tree patterns on at most 5~vertices.
In Section~\ref{sec:bijection} we prove bijections between different classes of pattern-avoiding binary trees and other combinatorial objects, in particular pattern-avoiding lattice paths and set partitions.
In Section~\ref{sec:wilf} we present results for establishing Wilf-equivalence between tree patterns.
We conclude with some open problems in Section~\ref{sec:open}.

\section{Preliminaries}
\label{sec:prelim}

In this section we introduce a few general definitions related to binary trees, and we define our notion of pattern avoidance for those objects.

\subsection{Binary tree notions}
\label{sec:notions}

We consider binary trees whose vertex set is a set of consecutive integers~$\{i,i+1,\ldots,j\}$.
In particular, we write~$\cT_n$ for the set of binary trees with the vertex set~$[n]:=\{1,2,\ldots,n\}$.
The vertex labels of each tree are defined uniquely by the \defi{search tree property}, i.e., for any vertex~$i$, all its left descendants are smaller than~$i$ and all its right descendants are greater than~$i$.
The special empty tree with $n=0$ vertices is denoted by~$\varepsilon$, so $\cT_0=\{\varepsilon\}$.
The following definitions are illustrated in Figure~\ref{fig:preorder}.
For any binary tree~$T$, we denote the root of~$T$ by~$r(T)$.
For any vertex~$i$ of~$T$, its left and right child are denoted by~$c_L(i)$ and~$c_R(i)$, respectively, and its parent is denoted by~$p(i)$.
If~$i$ does not have a left child, a right child or a parent, then we define~$c_L(i):=\varepsilon$, $c_R(i):=\varepsilon$, or $p(i):=\varepsilon$, respectively.
Furthermore, we write~$T(i)$ for the subtree of~$T$ rooted at~$i$.
Also, we define $L(i):=T(c_L(i))$ if $c_L(i)\neq\varepsilon$ and $L(i):=\varepsilon$ otherwise, and $R(i):=T(c_R(i))$ if $c_R(i)\neq \varepsilon$ and $R(i):=\varepsilon$ otherwise.
The subtrees rooted at the left and right child of the root are denoted by~$L(T)$ and~$R(T)$, respectively, i.e., we have $L(T)=L(r(T))$, and similarly $R(T)=R(r(T))$.
A \defi{left path} is a binary tree in which no vertex has a right child.
A \defi{left branch} in a binary tree is a subtree that is isomorphic to a left path.
The notions \defi{right path} and \defi{right branch} are defined analogously, by interchanging left and right.

\begin{figure}
\includegraphics{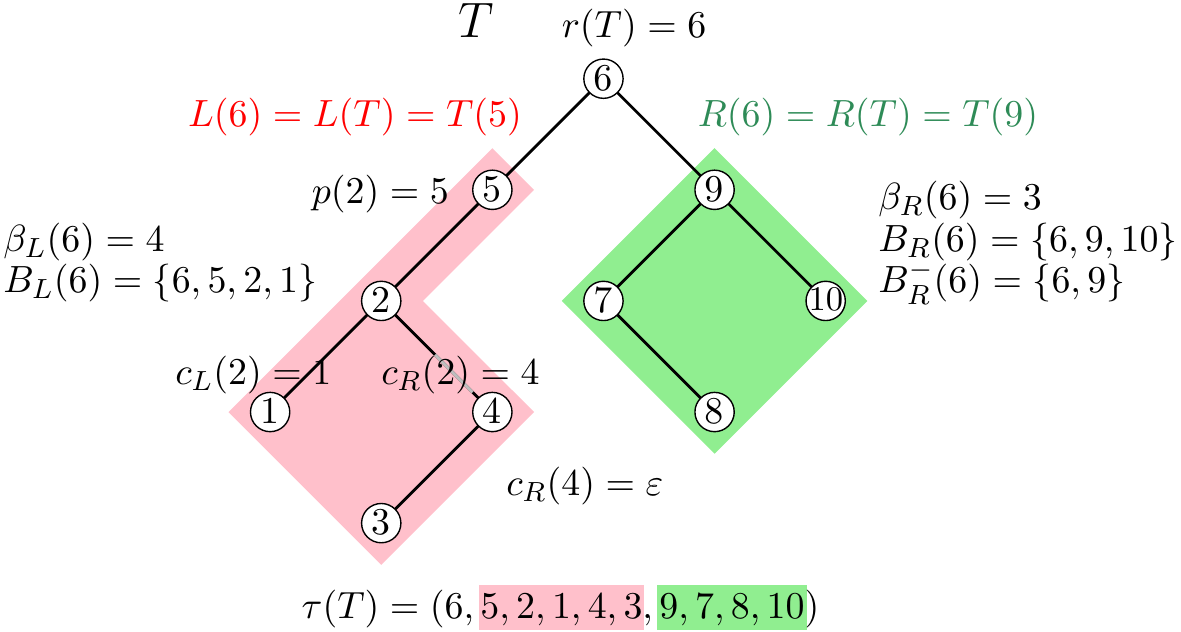}
\caption{Illustration of definitions related to binary trees.}
\label{fig:preorder}
\end{figure}

We associate~$T\in\cT_n$ with a permutation~$\tau(T)$ of~$[n]$ defined by
\begin{equation}
\label{eq:tau}
\tau(T):=\big(r(T),\tau(L(T)),\tau(R(T))\big),
\end{equation}
where the base case of the empty tree~$\varepsilon$ is defined to be the empty permutation~$\tau(\varepsilon):=\varepsilon$.
In words, $\tau(T)$ is the sequence of vertex labels obtained from a preorder traversal of~$T$, i.e., we first record the label of the root and then recursively record labels of its left subtree followed by labels of its right subtree.
Note that the right path~$T\in\cT_n$ satisfies~$\tau(T)=\ide_n$, the identity permutation.

For any vertex~$i$ we let~$\beta_L(i)$ and~$\beta_R(i)$ denote the number of vertices on the left branch or right branch, respectively, starting at~$i$, with the special cases $\beta_L(\varepsilon):=0$ and $\beta_R(\varepsilon):=0$.
We also define~$B_L(i):=\{c_L^{j-1}(i)\mid j=1,\ldots,\beta_L(i)\}$ and~$B_R(i):=\{c_R^{j-1}(i)\mid j=1,\ldots,\beta_R(i)\}$ as the corresponding sets of vertices on this branch.
Lastly, we define~$B_R^-(i):=B_R(i)\setminus c_R^{\beta_R(i)-1}(i)$, i.e., all vertices on the right branch except the last one.

In all the functions defined before that take a vertex~$i$ as an argument, the tree~$T$ containing~$i$ can be inferred from the context, and is omitted as an additional argument to avoid cluttering notation.

\subsection{Pattern-avoiding binary trees}
\label{sec:pat-tree}

\begin{figure}[b!]
\includegraphics{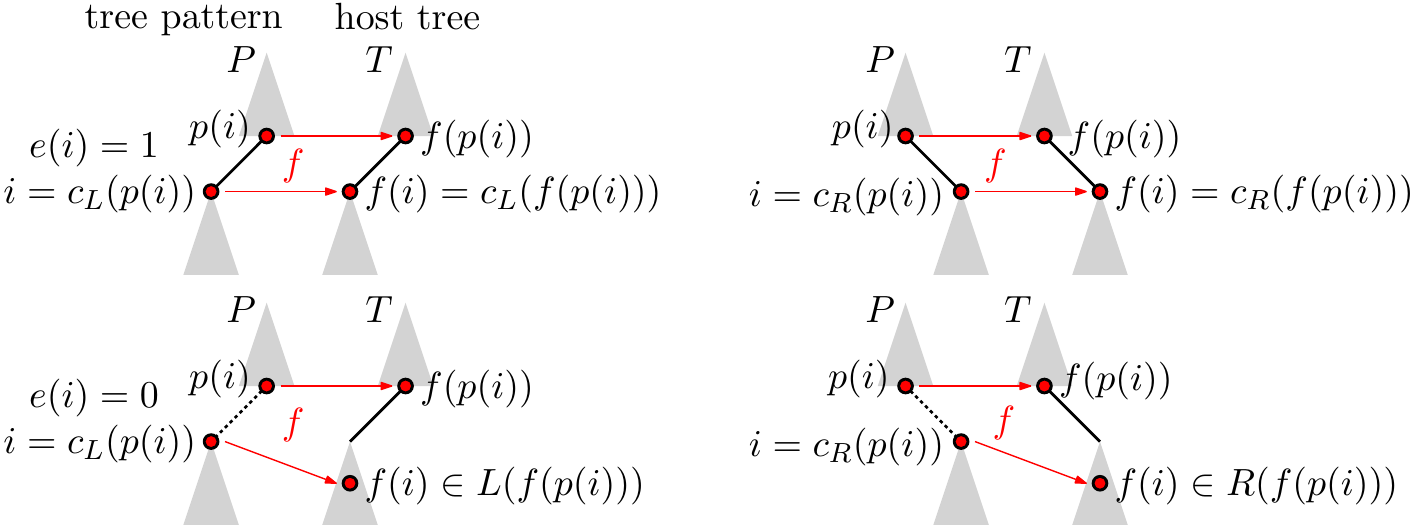}
\caption{Illustration of our notion of pattern containment in binary trees.}
\label{fig:contains}
\end{figure}

Our notion of pattern avoidance in binary trees generalizes the two distinct notions considered in~\cite{MR2645188} and~\cite{MR2967227} (recall Figure~\ref{fig:notions}).
This definition is illustrated in Figure~\ref{fig:contains}.
A \defi{tree pattern} is a pair~$(P,e)$ where~$P\in\cT_k$ and $e\colon [k]\setminus r(P) \rightarrow \{0,1\}$.
For any vertex~$i\in[k]\setminus r(P)$, a value~$e(i)=0$ is interpreted as the edge leading from~$i$ to its parent~$p(i)$ being non-contiguous, whereas a value~$e(i)=1$ is interpreted as this edge being contiguous.
In our figures, edges~$(i, p(i))$ in~$P$ with $e(i)=1$ are drawn solid, and edges with $e(i)=0$ are drawn dotted.
Formally, a tree~$T\in\cT_n$ \defi{contains} the pattern~$(P,e)$ if there is an injective mapping $f\colon [k]\rightarrow [n]$ satisfying the following conditions:
\begin{enumerate}[label=(\roman*),leftmargin=8mm, noitemsep, topsep=1pt plus 1pt]
\item  For every edge~$(i,p(i))$ of~$P$ with~$e(i)=1$, we have that $f(i)$ is a child of~$f(p(i))$ in~$T$.
Specifically, if~$i=c_L(p(i))$ then~$f(i)$ is the left child of~$f(p(i))$, i.e., we have $f(i)=c_L(f(p(i)))$, whereas if~$i=c_R(p(i))$ then~$f(i)$ is the right child of~$f(p(i))$, i.e., we have $f(i)=c_R(f(p(i)))$.
\item For every edge~$(i,p(i))$ of~$P$ with~$e(i)=0$, we have that $f(i)$ is a descendant of~$f(p(i))$ in~$T$.
Specifically, if~$i=c_L(p(i))$, then~$f(i)$ is a left descendant of~$f(p(i))$, i.e., we have $f(i)\in L(f(p(i)))$, whereas if~$i=c_R(p(i))$, then~$f(i)$ is a right descendant of~$f(p(i))$, i.e., we have $f(i)\in R(f(p(i)))$.
\end{enumerate}
We can retrieve the notions of contiguous and non-contiguous pattern containment used in~\cite{MR2645188} and~\cite{MR2967227} as special cases by defining~$e(i):=1$ for all $i\in[k]\setminus r(P)$, or $e(i):=0$ for all $i\in[k]\setminus r(P)$, respectively.

If $T$ does not contain~$(P,e)$, then we say that $T$ \defi{avoids}~$(P,e)$.
Furthermore, we define the set of binary trees with~$n$ vertices that avoid the pattern~$(P,e)$ as
\[ \cT_n(P,e) := \{T \in \cT_n \mid T \text{ avoids } (P, e)\}. \]
Note that $\cT_0(P,e)=\{\varepsilon\}$ for any nonempty tree pattern~$(P,e)$.
For avoiding multiple patterns $(P_1,e_1),\ldots,(P_\ell,e_\ell)$ simultaneously, we define
\[ \cT_n\bigl((P_1,e_1),\ldots,(P_\ell, e_\ell)\bigl) := \bigcap\nolimits_{i=1}^\ell \cT_n(P_i,e_i). \]

Clearly, the set of binary trees that avoids a tree pattern~$(P,e)$, $P\in \cT_k$, is monotonously non-decreasing in~$e$, i.e., if~$e(i) \leq e'(i)$ for every vertex~$i \in [k]\setminus r(P)$, then~$\cT_n(P,e) \seq \cT_n(P,e')$.

Given a tree pattern~$(P,e)$ and a vertex~$i$ in~$P$, we sometimes consider the induced subpattern~$(P(i),e_{P(i)})$, where $e_{P(i)}$ denotes the restriction of~$e$ to the vertex set of~$P(i)\setminus i$.

We often write a tree pattern~$(P,e)$, $P\in\cT_k$, in compact form as a pair~$\big(\tau(P),(e(\tau_2),\ldots,e(\tau_k))\big)$ where $\tau(P)=(\tau_1,\tau_2,\ldots,\tau_k)$; see Figure~\ref{fig:mirror}.
In words, the tree~$P$ is specified by the preorder permutation~$\tau(P)$, and the function~$e$ is specified by the sequence of values for all vertices except the root in the preorder sequence, i.e., this sequence has length~$k-1$.

For any tree pattern~$(P,e)$, we write~$\mu(P,e)$ for the pattern obtained by mirroring the tree, i.e., by changing left and right.
Note that the mirroring operation changes the vertex labels so that the search tree property is maintained, specifically the vertex~$i$ becomes~$n+1-i$.
Trivially, we have~$\cT_n(\mu(P,e))=\mu(\cT_n(P,e))$, in particular $(P,e)$ and~$\mu(P,e)$ are Wilf-equivalent.

\begin{figure}[h!]
\includegraphics[page=1]{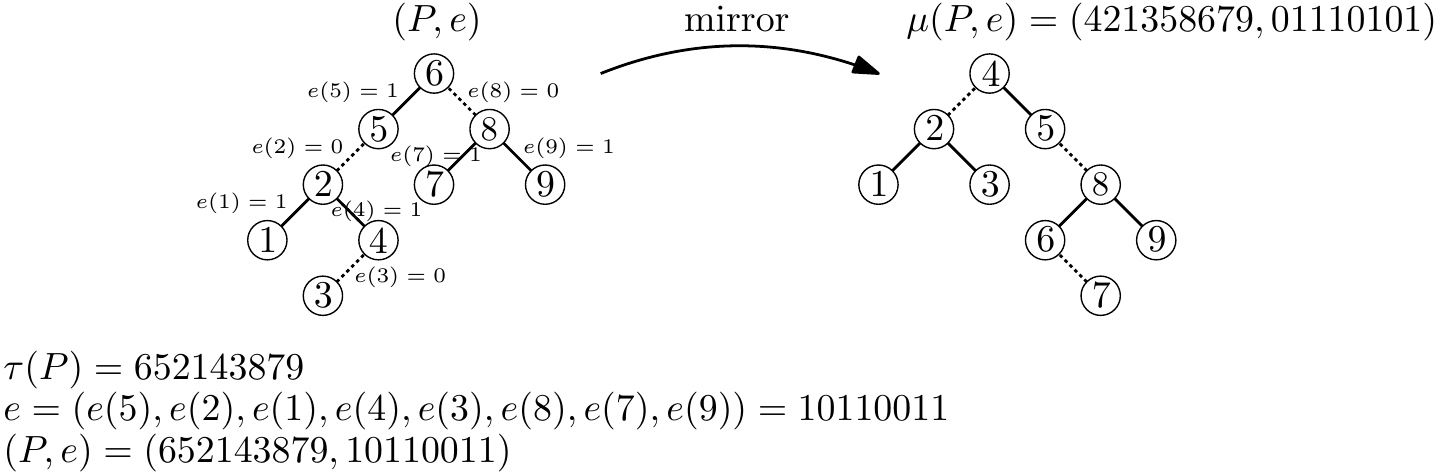}
\caption{Compact encoding of binary tree patterns and mirroring operation.}
\label{fig:mirror}
\end{figure}

\section{Encoding binary trees by permutations}
\label{sec:enc-perm}

In this section we establish that avoiding a tree pattern in binary trees is equivalent to avoiding a corresponding permutation mesh pattern in 231-avoiding permutations (Theorem~\ref{thm:bijection} below).

\subsection{Pattern-avoiding permutations}
\label{sec:pat-perm}

We write~$S_n$ for the set of all permutations of~$[n]$.
Given two permutations~$\pi\in S_n$ and~$\tau\in S_k$, we say that \defi{$\pi$ contains~$\tau$ as a pattern} if there is a sequence of indices~$\nu_1<\cdots<\nu_k$, such that $\pi(\nu_1),\ldots,\pi(\nu_k)$ are in the same relative order as~$\tau=\tau(1),\ldots,\tau(k)$.
If $\pi$ does not contain~$\tau$, then we say that \defi{$\pi$ avoids~$\tau$}.
We write $S_n(\tau)$ for the permutations from~$S_n$ that avoid the pattern~$\tau$.
More generally, for multiple patterns~$\tau_1, \ldots, \tau_\ell$ we define~$S_n(\tau_1, \ldots, \tau_\ell):=\bigcap_{i=1}^\ell S_n(\tau_i)$, i.e., this is the set of permutations of length~$n$ that avoid each of the patterns~$\tau_1, \ldots, \tau_\ell$.

It is well known that preorder traversals of binary trees are in bijection with 231-avoiding permutations (see, e.g.~\cite{DBLP:journals/cacm/Knott77}).

\begin{lemma}
\label{lem:231-tree}
The mapping~$\tau\colon \cT_n \rightarrow S_n(231)$ defined in~\eqref{eq:tau} is a bijection.
\end{lemma}

\subsection{Mesh patterns}
\label{sec:mesh-pat}

Mesh patterns were introduced by Br{\"a}nd{\'e}n and Claesson~\cite{MR2795782}, and they generalize classical permutation patterns discussed in the previous section.
We recap the required definitions; see Figure~\ref{fig:mesh}.
The \defi{grid representation} of a permutation~$\pi \in S_n$ is defined as~$G(\pi) := \{(i,\pi(i)) \mid i \in [n]\}$.
Graphically, this is the permutation matrix corresponding to~$\pi$.

A \defi{mesh pattern} is a pair~$\sigma:=(\tau,C)$, where $\tau\in S_k$ and $C \seq \{0,\ldots,k\}\times \{0,\ldots,k\}$.
In our figures, we depict $\sigma$ by the grid representation of~$\tau$, and we shade all unit squares~$[i,i+1]\times [j,j+1]$ for which~$(i,j)\in C$.
A permutation~$\pi \in S_n$ \defi{contains} the mesh pattern~$\sigma=(\tau,C)$, if there is a sequence of indices~$\nu_1<\cdots<\nu_k$ such that the following two conditions hold:
\begin{enumerate}[label=(\roman*),leftmargin=8mm, noitemsep, topsep=1pt plus 1pt]
\item The entries of~$\pi(\nu_1),\ldots,\pi(\nu_k)$ are in the same relative order as~$\tau=\tau(1),\ldots,\tau(k)$.
\item We let $\lambda_1<\cdots<\lambda_k$ be the values~$\pi(\nu_1),\ldots,\pi(\nu_k)$ sorted in increasing order.
For all pairs~$(i,j) \in C$, we require that~$G(\pi)\cap R_{i,j}=\emptyset$, where $R_{i,j}$ is the rectangular open set defined as $R_{i,j}:=(\nu_i,\nu_{i+1}) \,\times\, (\lambda_j,\lambda_{j+1})$, using the sentinel values $\nu_0:=\lambda_0:= 0$ and~$\nu_{k+1}=\lambda_{k+1}:=n+1$.
\end{enumerate}
The first condition requires a match of the classical pattern~$\tau$ in~$\pi$.
The second condition requires that~$G(\pi)$ has no point in any of the regions~$R_{i,j}$ that correspond to the shaded cells~$C$ of the pattern.
Thus, the classical pattern~$\tau\in S_k$ is the mesh pattern~$(\tau,\emptyset)$.

\begin{figure}[h!]
\includegraphics{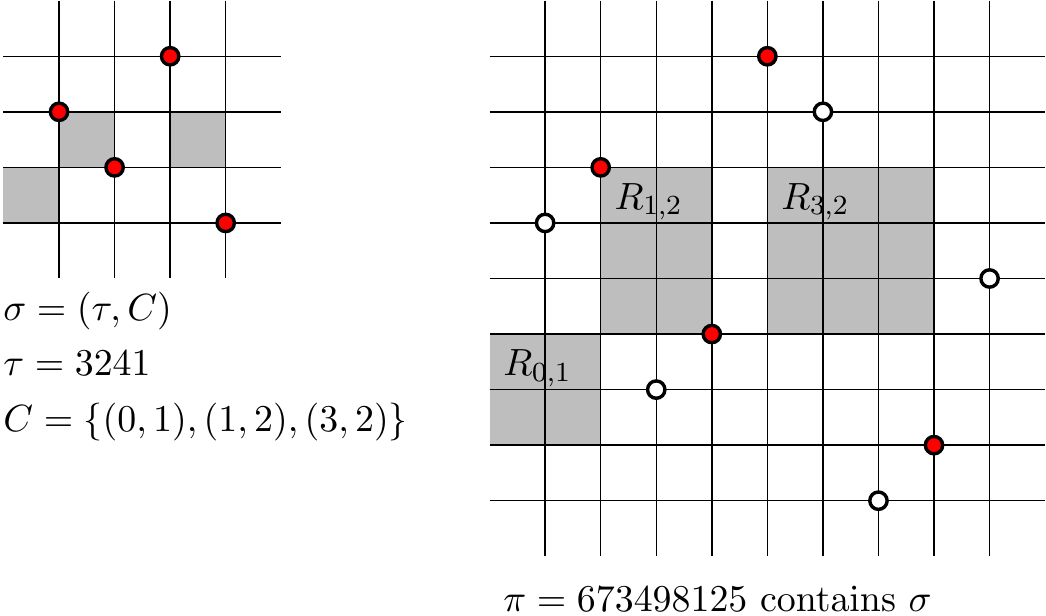}
\caption{Illustration of mesh pattern containment.}
\label{fig:mesh}
\end{figure}

\subsection{From binary tree patterns to mesh patterns}
\label{sec:tree-mesh}

In the following, for a given tree pattern~$(P,e)$, $P\in \cT_k$, we construct a permutation mesh pattern~$\sigma(P,e)=(\tau(P),C)$, consisting of the permutation~$\tau(P)$ obtained by a preorder traversal of the tree~$P$ and a set of shaded cells~$C$.
These definitions are illustrated in Figures~\ref{fig:sigmaPe1} and~\ref{fig:sigmaPe2}.
We consider the inverse permutation of~$\tau(P)\in S_k$, which we abbreviate to~$\rho:=\tau(P)^{-1}\in S_k$.
The permutation~$\rho$ gives the position of each vertex in the preorder traversal~$\tau(P)$ of~$P$.
Recall the definition of the set~$B_R^-(i)$ given in Section~\ref{sec:notions}.
For any vertex~$i\in[k]$ we define
\begin{subequations}
\label{eq:sigmaPe}
\begin{equation}
\label{eq:Ci}
C_i:=\big\{(\rho(i)-1,j)\mid j\in B_R^-(i)\big\},
\end{equation}
and for any $i\in[k]\setminus r(P)$ we define
\begin{equation}
\label{eq:Cip}
C_i':=\begin{cases} \emptyset & \text{if } e(i)=0, \\
\Big\{\big(\rho(i)-1,\min P(i)-1\big),\,\big(\rho(i)-1,\max P(i)\big)\Big\} & \text{if } e(i)=1.
\end{cases}
\end{equation}
Then the mesh pattern~$\sigma(P,e)$ corresponding to the tree pattern~$(P,e)$ is defined as
\begin{equation}
\sigma(P,e):=\Big(\tau(P),\,\bigcup\nolimits_{i\in[k]} C_i\cup \bigcup\nolimits_{i\in[k]\setminus r(P)} C_i'\Big).
\end{equation}
\end{subequations}
In words, for any two vertices (not necessarily distinct and not necessarily forming an edge) on a maximal right branch, such that neither of them is the last vertex on the branch, we shade the cell directly left of the smaller vertex and directly above the larger vertex, and for every edge~$(i,p(i))$ with $e(i)=1$ we shade two additional cells to the left and bottom/top of the submatrix corresponding to the subtree~$P(i)$.

\begin{figure}
\includegraphics[page=1]{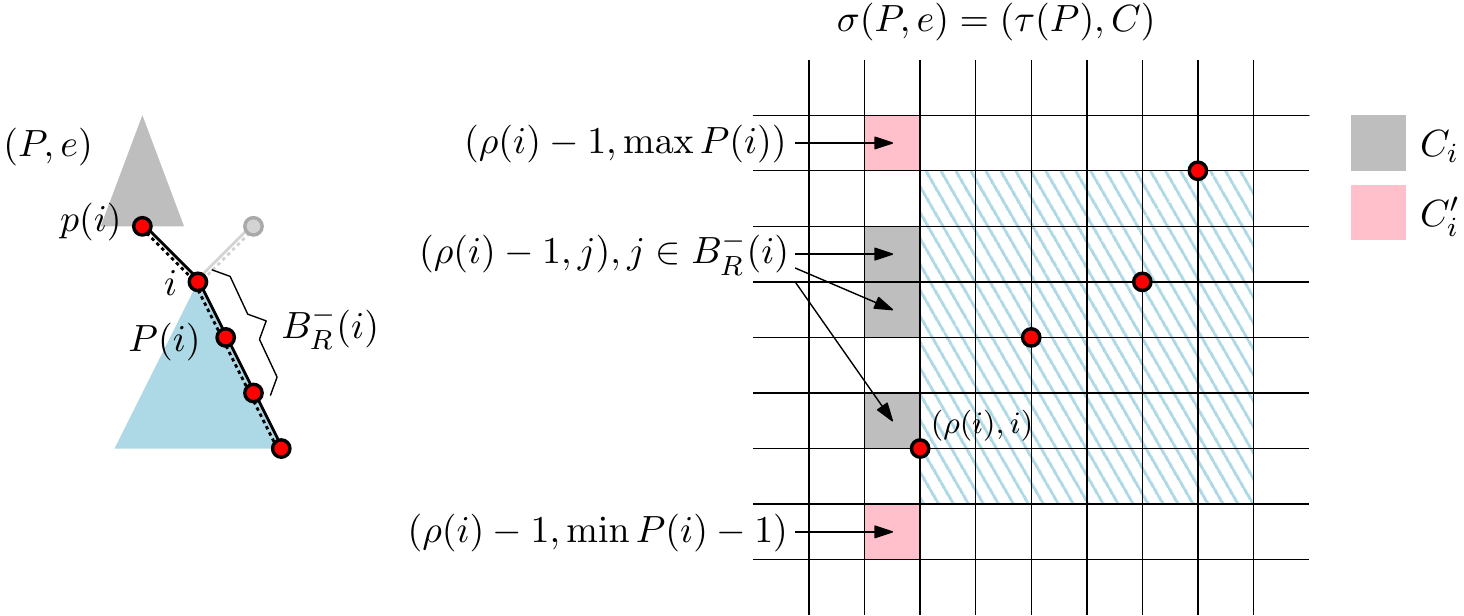}
\caption{Schematic illustration of the definition of the mesh pattern~$\sigma(P,e)$ for a tree pattern~$(P,e)$.
The edges of the tree~$P$ can be contiguous or non-contiguous, and are therefore drawn half solid and half dotted.
In the tree shown in the figure, $i$ is the right child of~$p(i)$, but it might also be the left child of~$p(i)$ (faint lines).
On the right, the shaded cells belong to the mesh pattern, and the hatched region corresponds to the submatrix given by the subtree~$P(i)$.
}
\label{fig:sigmaPe1}
\end{figure}

\begin{figure}
\includegraphics[page=2]{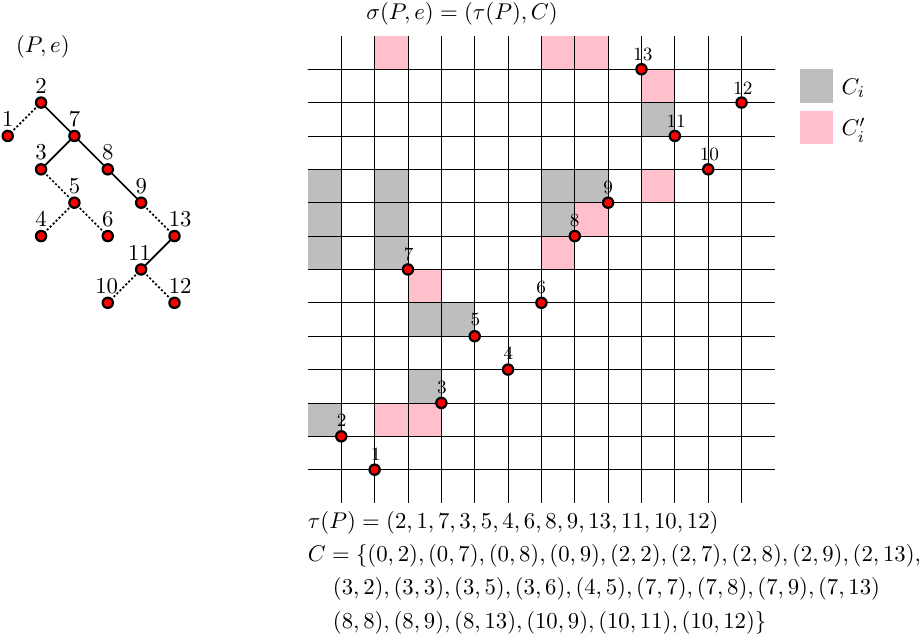}
\caption{Specific example of the mesh pattern~$\sigma(P,e)$ corresponding to a tree pattern~$(P,e)$.}
\label{fig:sigmaPe2}
\end{figure}

The following generalization of Lemma~\ref{lem:231-tree} is the main result of this section.
Our theorem also generalizes Theorem~12 from~\cite{PSSS14}, which is obtained as the special case when all edges of~$P$ are non-contiguous, i.e., $e(i)=0$ for all $i\in[k]\setminus r(P)$.

\begin{theorem}
\label{thm:bijection}
For any tree pattern~$(P,e)$, $P\in \cT_k$, consider the mesh pattern~$\sigma(P,e)=(\tau(P),C)$ defined in~\eqref{eq:sigmaPe}.
Then the mapping~$\tau\colon \cT_n(P,e) \rightarrow S_n(231,\sigma(P,e))$ is a bijection.
\end{theorem}

This theorem extends naturally to avoiding multiple tree patterns~$(P_1,e_1),\ldots,(P_\ell,e_\ell)$, i.e., $\tau\colon \cT_n((P_1,e_1),\ldots,(P_\ell,e_\ell))\rightarrow S_n(231,\sigma(P_1,e_1),\ldots,\sigma(P_\ell,e_\ell))$ is a bijection.

\begin{proof}
By Lemma~\ref{lem:231-tree}, $\tau\colon \cT_n\rightarrow S_n(231)$ is a bijection.
Consequently, it suffices to show that $T\in \cT_n$ contains the tree pattern~$(P,e)$ if and only if~$\tau(T)\in S_n(231)$ contains the mesh pattern~$\sigma(P,e)$.

Using the definitions of tree patterns and mesh patterns from Sections~\ref{sec:pat-tree} and~\ref{sec:mesh-pat}, respectively, and combining them with ~\eqref{eq:sigmaPe}, a straightforward induction shows that if~$T\in\cT_n$ contains the tree pattern~$(P,e)$, then~$\tau(T)$ contains the mesh pattern~$\sigma(P,e)$.
For this argument we also use that in the mesh pattern~$\sigma(P,e)=(\tau(P),C)$, none of the four corner cells is shaded, i.e., we have $(0,0),(0,k),(k,0),(k,k)\notin C$.

It remains to show that if $\tau(T)$ for $T\in\cT_n$ contains the mesh pattern~$\sigma(P,e)=(\tau(P),C)$, then $T$ contains the tree pattern~$(P,e)$.
This means that there are indices $\nu_1<\cdots<\nu_k$ satisfying conditions~(i) and~(ii) stated in Section~\ref{sec:mesh-pat}.
We define the abbreviation~$\pi:=\tau(T)$.
Let $\lambda_1<\cdots<\lambda_k$ be the values~$\pi(\nu_1),\ldots,\pi(\nu_k)$ sorted in increasing order, and let $Q:=\{\lambda_i\mid i \in [k]\}$ be the corresponding set of values of~$\tau(T)$ that correspond to this occurrence of the mesh pattern~$\sigma(P,e)$.
From~\eqref{eq:tau}, we have~$\tau(T)=\big(r(T),\tau(L(T)),\tau(R(T))\big)\in S_n(231)$, and so we consider the following four cases, illustrated in Figure~\ref{fig:proof}.

\begin{figure}
\includegraphics{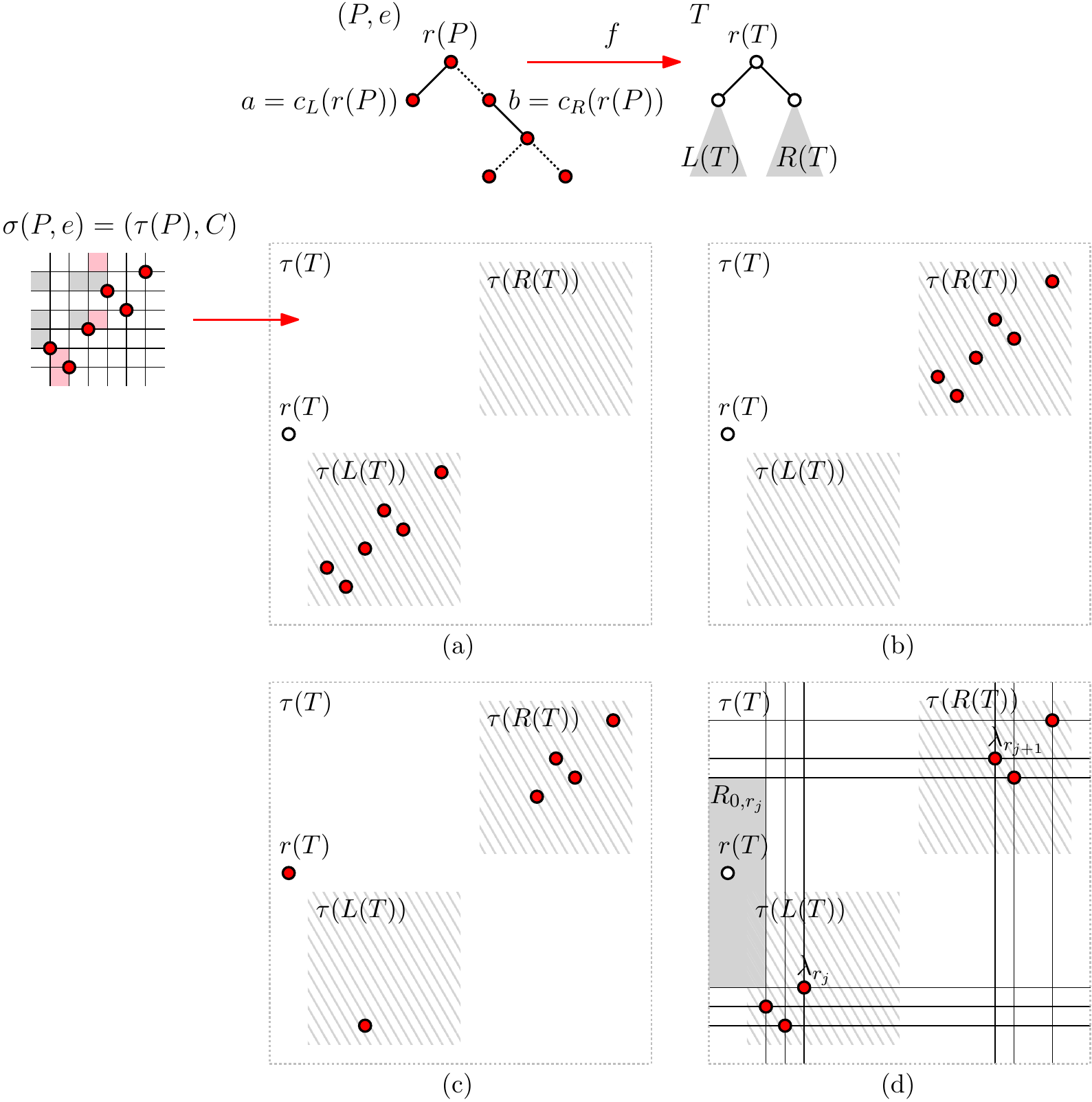}
\caption{Illustration of the four cases in the proof of Theorem~\ref{thm:bijection}.}
\label{fig:proof}
\end{figure}

\textbf{Case~(a):} $Q\seq \tau(L(T))$.
In this case $\tau(L(T))$ contains the mesh pattern~$\sigma(P,e)$.
It follows by induction that $L(T)$ contains the tree pattern~$(P,e)$, and therefore $T$ contains the tree pattern~$(P,e)$.

\textbf{Case~(b):} $Q\seq \tau(R(T))$.
In this case $\tau(R(T))$ contains the mesh pattern~$\sigma(P,e)$.
It follows by induction that $R(T)$ contains the tree pattern~$(P,e)$, and therefore $T$ contains the tree pattern~$(P,e)$.

\textbf{Case~(c):} $r(T)\in Q$.
We define~$a:=c_L(r(P))=r(L(P))$ and~$b:=c_R(r(P))=r(R(P))$.
We assume that $a,b\neq \varepsilon$; the other cases are analogous.
In this case $\tau(L(T))$ contains the mesh pattern~$\sigma(L(P),e_{L(P)})$ and $\tau(R(T))$ contains the mesh pattern~$\sigma(R(P),e_{R(P)})$.

We define $\rho:=\tau(P)^{-1}\in S_k$.
From~\eqref{eq:Ci} we see that if $c_R(a)\neq \varepsilon$, then we have $(\rho(a)-1,j)=(1,j)\in C$ for all $j\in B_R^-(a)$.
Furthermore, if $e(a)=1$, then we have $(\rho(a)-1,\min P(a)-1)=(1,0)\in C$ and $(\rho(a)-1,\max P(a))=(1,r(P)-1)\in C$.
We thus obtain that if $e(a)=1$, then the occurrence of the mesh pattern~$\sigma(L(P),e_{L(P)})$ in~$\tau(L(T))$ must contain the first element of~$\tau(L(T))$.
By induction, we obtain that $L(T)$ contains the tree pattern~$(L(P),e_{L(P)})$, and if $e(a)=1$ then an occurrence of this pattern includes the vertex~$c_L(T)$.

Similarly, from~\eqref{eq:Ci} we see that if $c_R(b)\neq \varepsilon$, then we have $(\rho(b)-1,j)\in C$ for all $j\in B_R^-(b)$.
Furthermore, if $e(b)=1$, then we have $(\rho(b)-1,\min P(b)-1)=(\rho(b)-1,r(P))\in C$ and $(\rho(b)-1,\max P(b))=(\rho(b)-1,k)\in C$.
We thus obtain that if $e(b)=1$, then the occurrence of the mesh pattern~$\sigma(R(P),e_{R(P)})$ in~$\tau(R(T))$ must contain the first element of~$\tau(R(T))$.
By induction, we obtain that $R(T)$ contains the tree pattern~$(R(P),e_{R(P)})$, and if $e(b)=1$ then an occurrence of this pattern includes the vertex~$c_R(T)$.

Combining these observations yields that $T$ contains the tree pattern~$(P,e)$.

\textbf{Case~(d):} $r(T)\notin Q$, $Q\cap \tau(L(T))\neq \emptyset$ and $Q\cap \tau(R(T))\neq \emptyset$.

We define $\rho:=\tau(P)^{-1}\in S_k$, $b:=\beta_R(r(P))$ and $r_i:=c_R^{i-1}(r(P))$ for $i=1,\ldots,b$.
For any vertex~$i\in [k]$ of~$P$, we have that all~$j\in L(i)$ come after~$i$ in~$\tau(P)$ and are smaller than~$i$.
Consequently, there is an integer~$1\leq j<b$ such that for all $i=1,\ldots,j$ we have $\lambda_{r_i}\in \tau(L(T))$ and $\lambda_k\in\tau(L(T))$ for all $k\in L(r_i)$, and moreover for all $i=j+1,\ldots,b$ we have $\lambda_{r_i}\in\tau(R(T))$ and $\lambda_k\in \tau(R(T))$ for all $k\in L(r_i)$.
However, by the definition~\eqref{eq:Ci} we have $(0,r_j)\in C$, which implies that $(1,r(T))\in R_{0,r_j}$, so this case cannot occur.

This completes the proof of the theorem.
\end{proof}

\section{Generating pattern-avoiding binary trees}
\label{sec:generate}

In this section we apply the Hartung--Hoang--M\"utze--Williams generation framework to pattern-avoiding binary trees.
The main results are simple and efficient algorithms (Algorithm~S and Algorithm~H) to generate different classes of pattern-avoiding binary trees, subject to some mild constraints on the tree pattern(s) that are inherited from applying the framework (Theorems~\ref{thm:algoS} and~\ref{thm:algoH}, respectively).

\subsection{Tree rotations and slides}

A natural and well-studied operation on binary trees are tree rotations; see Figure~\ref{fig:rot}.
We consider a tree~$T\in \cT_n$ and one of its edges~$(i,j)$ with $j=c_R(i)$, and we let $Y$ be the left subtree of~$j$, i.e., $Y:=L(j)$.
A \defi{rotation of the edge~$(i,j)$} yields the tree obtained by the following modifications:
The child~$i$ of~$p(i)$ is replaced by~$j$ (unless $p(i)=\varepsilon$ in~$T$), $i$ becomes the left child of~$j$, and $Y$ becomes the right subtree of~$i$.
We denote this operation by~$j\diru$, and we refer to it as \defi{up-rotation of~$j$}, indicating that the vertex~$j$ moves up.
The operation~$j\diru$ is well-defined if and only if $j$ is not the root and $p(j)<j$, or equivalently $j=c_R(p(j))$.
The inverse operation is denoted by~$j\dird$, and we refer to it as \defi{down-rotation of~$j$}, indicating that the vertex~$j$ moves down.
The operation~$j\dird$ is well-defined if and only if $j$ has a left child (which must be smaller), i.e., $c_L(j)\neq \varepsilon$.
An \defi{up-slide} or \defi{down-slide of~$j$ by $d$ steps} is a sequence of $d$ up- or down-rotations of~$j$, respectively, which we write as $(j\diru)^d$ and~$(j\dird)^d$.

\subsection{A simple greedy algorithm}
\label{sec:algo}

We use the following simple greedy algorithm to generate a set of binary trees $\cL_n\seq \cT_n$.
We say that a slide is \defi{minimal} (w.r.t.~$\cL_n$), if every slide of the same vertex in the same direction by fewer steps creates a binary tree that is not in~$\cL_n$.

\begin{algo}{Algorithm~S}{Greedy slides}
This algorithm attempts to greedily generate a set of binary trees $\cL_n\seq \cT_n$ using minimal slides starting from an initial binary tree~$T_0\in \cL_n$.
\begin{enumerate}[label={\bfseries S\arabic*.}, leftmargin=8mm, noitemsep, topsep=3pt plus 3pt]
\item{} [Initialize] Visit the initial tree~$T_0$.
\item{} [Slide] Generate an unvisited binary tree from~$\cL_n$ by performing a minimal slide of the largest possible vertex in the most recently visited binary tree.
If no such slide exists, or the direction of the slide is ambiguous, then terminate.
Otherwise visit this binary tree and repeat~S2.
\end{enumerate}
\end{algo}

\begin{figure}[h!]
\includegraphics[page=2]{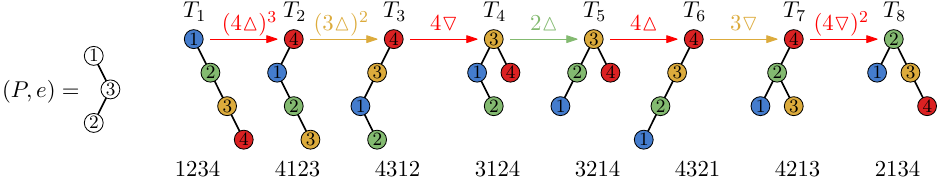}
\caption{Run of Algorithm~S that visits all binary trees in the set~$\cT_4(P,e)$.
Below each tree~$T$ is the corresponding permutation~$\tau(T)$.}
\label{fig:algo-ex1}
\end{figure}

\begin{figure}[h!]
\includegraphics[page=3]{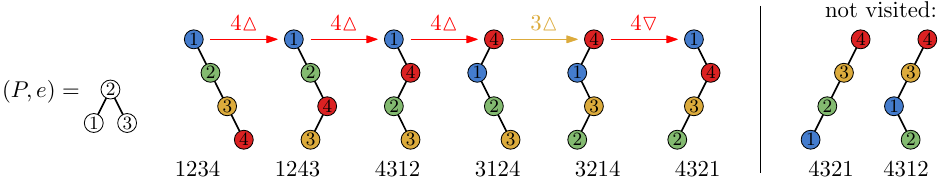}
\caption{Run of Algorithm~S that does not visit all binary trees in the set~$\cT_4(P,e)$.}
\label{fig:algo-ex2}
\end{figure}

To illustrate the algorithm, consider the example in Figure~\ref{fig:algo-ex1}.
Suppose we choose the right path~$T_1$ shown in the figure as initial tree for the algorithm, i.e., $T_0:=T_1$.
In the first iteration, Algorithm~S performs an up-slide of the vertex~4 by three steps to obtain~$T_2$.
This up-slide is minimal, as an up-slide of~4 in~$T_1$ by one or two steps creates the forbidden tree pattern~$(P,e)$.
Note that any tree created from~$T_2$ by a down-slide of~4 either contains the forbidden pattern or has been visited before.
Consequently, the algorithm applies an up-slide of~3 by two steps, yielding~$T_3$.
After five more slides, the algorithm terminates with~$T_8$, and at this point it has visited all eight trees in~$\cT_4(P,e)$.

Now consider the example in Figure~\ref{fig:algo-ex2}, where the algorithm terminates after having visited six different trees from~$\cT_4(P,e)$.
However, the set~$\cT_4(P,e)$ contains two more trees that are not visited by the algorithm.

\begin{wrapfigure}{r}{0.26\textwidth}
\vspace{-5mm}
\includegraphics[page=1]{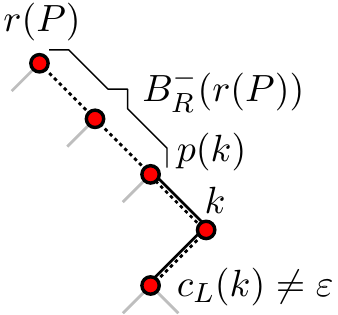}
\caption{Definition of friendly tree patterns.}
\label{fig:friendly}
\end{wrapfigure}
We now formulate simple sufficient conditions on the tree pattern~$(P,e)$ ensuring that Algorithm~S successfully visits all trees in~$\cT_n(P,e)$.
Specifically, we say that a tree pattern~$(P,e)$, $P\in \cT_k$, is \defi{friendly}, if it satisfies the following three conditions; see Figure~\ref{fig:friendly}:
\begin{enumerate}[label=(\roman*),leftmargin=8mm, noitemsep, topsep=1pt plus 1pt]
\item We have $p(k)\neq\varepsilon$ and~$c_L(k)\neq \varepsilon$, i.e., the largest vertex~$k$ is neither the root nor a leaf in~$P$.
\item For every $j\in B_R^-(r(P))\setminus r(P)$ we have $e(j)=0$, i.e., the edges on the right branch starting at the root, except possibly the last one, are all non-contiguous.
\item If~$e(k)=1$, then we have~$e(c_L(k))=0$, i.e., if the edge from~$k$ to its parent is contiguous, then the edge to its left child must be non-contiguous.
\end{enumerate}
Note that for non-contiguous tree patterns, i.e., $e(i)=0$ for all $i\in[k]\setminus r(P)$, conditions~(ii) and~(iii) are always satisfied.

The following is our main result of this section.

\begin{theorem}
\label{thm:algoS}
Let $(P_1,e_1),\ldots,(P_\ell,e_\ell)$ be friendly tree patterns.
Then Algorithm~S initialized with the tree~$\tau^{-1}(\ide_n)$ visits every binary tree from~$\cT_n((P_1,e_1),\ldots,(P_\ell,e_\ell))$ exactly once.
\end{theorem}

Recall that $\tau^{-1}(\ide_n)$ is the right path, i.e., the tree that corresponds to the identity permutation.
Note that by condition~(i) in the definition of friendly tree pattern, we have $\tau^{-1}(\ide_n) \in \cT_n((P_1,e_1),\ldots,(P_\ell,e_\ell))$.
We shall see that our notion of friendly tree patterns is inherited from the notion of tame mesh permutation patterns used in~\cite[Thm.~15]{MR4391718}.

\subsection{Permutation languages}

We prove Theorem~\ref{thm:algoS} by applying the Hartung--Hoang--M\"utze--Williams generation framework~\cite{MR4391718}.
Let us recap the most important concepts.
We interpret a permutation~$\pi\in S_n$ in one-line notation as a string as $\pi=\pi(1),\pi(2),\ldots,\pi(n)=a_1 a_2\cdots a_n$.
Recall that~$\varepsilon\in S_0$ denotes the empty permutation.
For any $\pi\in S_{n-1}$ and any $1\leq i\leq n$, we write $c_i(\pi)\in S_n$ for the permutation obtained from~$\pi$ by inserting the new largest value~$n$ at position~$i$ of~$\pi$, i.e., if $\pi=a_1\cdots a_{n-1}$ then $c_i(\pi)=a_1\cdots a_{i-1} \, n\, a_i \cdots a_{n-1}$.
Moreover, for~$\pi\in S_n$, we write $p(\pi)\in S_{n-1}$ for the permutation obtained from~$\pi$ by removing the largest entry~$n$.
Given a permutation $\pi=a_1\cdots a_n$ with a substring $a_i\cdots a_j$ with $a_i>a_{i+1},\ldots,a_j$, a \defi{right jump of the value~$a_i$ by $j-i$~steps} is a cyclic left rotation of this substring by one position to $a_{i+1}\cdots a_j a_i$.
Similarly, given a substring $a_i\cdots a_j$ with $a_j>a_i,\ldots,a_{j-1}$, a \defi{left jump of the value~$a_j$ by $j-i$~steps} is a cyclic right rotation of this substring to $a_j a_i\cdots a_{j-1}$.

The framework from~\cite{MR4391718} uses the following simple greedy algorithm to generate a set of permutations $L_n\seq S_n$.
We say that a jump is \defi{minimal} (w.r.t.~$L_n$), if every jump of the same value in the same direction by fewer steps creates a permutation that is not in~$L_n$.

\begin{algo}{Algorithm~J}{Greedy minimal jumps}
This algorithm attempts to greedily generate a set of permutations $L_n\seq S_n$ using minimal jumps starting from an initial permutation $\pi_0 \in L_n$.
\begin{enumerate}[label={\bfseries J\arabic*.}, leftmargin=8mm, noitemsep, topsep=3pt plus 3pt]
\item{} [Initialize] Visit the initial permutation~$\pi_0$.
\item{} [Jump] Generate an unvisited permutation from~$L_n$ by performing a minimal jump of the largest possible value in the most recently visited permutation.
If no such jump exists, or the jump direction is ambiguous, then terminate.
Otherwise visit this permutation and repeat~J2.
\end{enumerate}
\end{algo}

One can argue that if a certain jump yields a visited or unvisited permutation, then all permutations in $L_n$ obtained by a jump of the same value in the same direction by more steps are also visited or unvisited, respectively. 
An analogous statement holds for slides in Algorithm~S.

The following main result from~\cite{MR4391718} provides a sufficient condition on the set~$L_n$ to guarantee that Algorithm~J successfully generates all permutations in~$L_n$.
This condition is captured by the following closure property of the set~$L_n$.
A set of permutations~$L_n\seq S_n$ is called a \defi{zigzag language}, if either $n=0$ and $L_0=\{\varepsilon\}$, or if $n\geq 1$ and $L_{n-1}:=\{p(\pi)\mid \pi\in L_n\}$ is a zigzag language such that for every $\pi\in L_{n-1}$ we have~$c_1(\pi)\in L_n$ and~$c_n(\pi)\in L_n$.

\begin{theorem}[\cite{MR4391718}]
\label{thm:jump}
Given any zigzag language of permutations~$L_n$ and initial permutation $\pi_0=\ide_n$, Algorithm~J visits every permutation from~$L_n$ exactly once.
\end{theorem}

\subsection{Tame permutation patterns and friendly tree patterns}

We say that an infinite sequence $L_0,L_1,\ldots$ of sets of permutations is \defi{hereditary}, if $L_{i-1}=p(L_i)$ holds for all $i\geq 1$.
We say that a (classical or mesh) permutation pattern~$\tau$ is \defi{tame}, if $S_n(\tau)$, $n\geq 0$, is a hereditary sequence of zigzag languages.

\begin{lemma}[{\cite[Lem.~6]{MR4391718}}]
\label{lem:intersection}
Let $L_0,L_1,\ldots$ and $M_0,M_1,\ldots$ be two hereditary sequences of zigzag languages.
Then $L_n\cap M_n$ for $n\geq 0$ is also a hereditary sequence of zigzag languages.
\end{lemma}

The following result about classical permutation patterns was proved in~\cite{MR4391718}.

\begin{lemma}[{\cite[Lem.~9]{MR4391718}}]
\label{lem:tame-classical}
If a pattern~$\tau\in S_k$, $k\geq 3$, does not have the largest value~$k$ at the leftmost or rightmost position, then it is tame.
\end{lemma}

Lemma~\ref{lem:tame-classical} applies in particular to the classical pattern~$\tau=231$.
The following more general result was proved for mesh patterns.

\begin{lemma}[{\cite[Thm.~15]{MR4391718}}]
\label{lem:tame-mesh}
Let $\sigma=(\tau,C)$, $\tau\in S_k$, $k\geq 3$, be a mesh pattern, and let $i$ be the position of the largest value~$k$ in~$\tau$.
If the pattern satisfies the following four conditions, then it is tame:
\begin{enumerate}[label=(\roman*),leftmargin=8mm, noitemsep, topsep=3pt plus 3pt]
\item $i$ is different from~1 and~$k$.
\item For all $a\in \{0,\ldots,k\}\setminus \{i-1,i\}$, we have $(a,k)\notin C$.
\item If $(i-1,k)\in C$, then for all $a\in \{0,\ldots,k\}\setminus i-1$ we have $(a,k-1)\notin C$ and for all $b\in \{0,\ldots,k-2\}$ we have that $(i,b)\in C$ implies $(i-1,b)\in C$.
\item If $(i,k)\in C$, then for all $a\in \{0,\ldots,k\}\setminus i$ we have $(a,k-1)\notin C$ and for all $b\in \{0,\ldots,k-2\}$ we have that $(i-1,b)\in C$ implies $(i,b)\in C$.
\end{enumerate}
\end{lemma}

The following crucial lemma connects friendliness of tree patterns to tameness of mesh patterns.

\begin{lemma}
\label{lem:friendly-tame}
Let $(P,e)$, $P\in \cT_k$, be a friendly tree pattern.
Consider the mesh pattern $\sigma(P,e)=(\tau(P),C)$ defined in~\eqref{eq:sigmaPe}, let $i$ be the position of the largest value~$k$ in~$\tau(P)$, and define the mesh pattern~$\sigma^-(P,e):=(\tau(P),C\setminus (i-1,k))$.
Then the mesh pattern~$\sigma^-(P,e)$ is tame.
\end{lemma}

\begin{proof}
By condition~(i) of friendly mesh patterns, we have $p(k)\neq \varepsilon$ and $c_L(k)\neq\varepsilon$ and therefore $i>1$ and $i<k$, respectively, which implies that condition~(i) of Lemma~\ref{lem:tame-mesh} is satisfied for both~$\sigma(P,e)$ and~$\sigma^-(P,e)$.

By condition~(ii) of friendly mesh patterns and the definition~\eqref{eq:Cip}, we have $(a,k)\notin C$ for all $a\in\{0,\ldots,k\}\setminus i-1$.
Furthermore, we have $(i-1,k)\in C$ if and only if~$e(k)=1$.
It follows that conditions~(ii) and~(iv) of Lemma~\ref{lem:tame-mesh} are satisfied for both~$\sigma(P,e)$ and~$\sigma^-(P,e)$.
Furthermore, if $e(k)=0$ then condition~(iii) is also satisfied for both~$\sigma(P,e)$ and~$\sigma^-(P,e)$.
Lastly, if $e(k)=1$ then we have $(i-1,k)\in C$, so $\sigma(P,e)$ does not satisfy condition~(iii), but $\sigma^-(P,e)$ does.
\end{proof}

\begin{lemma}
\label{lem:exchange}
The mesh patterns $\sigma(P,e)$ and~$\sigma^-(P,e)$ defined in Lemma~\ref{lem:friendly-tame} satisfy $S_n(231,\sigma(P,e))=S_n(231,\sigma^-(P,e))$.
\end{lemma}

Lemma~\ref{lem:exchange} is an instance of a so-called \defi{coincidence} among mesh patterns (see~\cite{MR4555705}), that is, situations where two sets of patterns have identical sets of avoiders.
To our knowledge, so far this phenomenon has primarily been studied with respect to single patterns, whereas our lemma involves two pairs of patterns.
In particular, in general we have $S_n(\sigma(P,e))\neq S_n(\sigma^-(P,e))$ and even $|S_n(\sigma(P,e))|\neq |S_n(\sigma^-(P,e))|$, and one such example is $(P,e)=(15234,1011)$.
However, by adding the pattern 231 on both sides, these inequalities become equalities.

\begin{proof}
It suffices to show that a 231-avoiding permutation~$\pi\in S_n$ that contains an occurrence of~$\sigma^-(P,e)$ also contains an occurrence of~$\sigma(P,e)$.
This proof uses an exchange argument; see Figure~\ref{fig:exchange}.
Let $i$ be the position of the largest value~$k$ in~$\tau(P)$.
Furthermore, for any~$j\in[k]$ we let $q(j)$ denote the point in~$G(\pi)$ corresponding to the value~$j$ in the occurrence of~$\sigma^-(P,e)$ in~$\pi$.
If $R_{i-1,k}$ contains no points from~$G(\pi)$, then this is also an occurrence of~$\sigma(P,e)$, and we are done.
Otherwise, we let~$q'$ be the leftmost point in~$G(\pi)\cap R_{i-1,k}$, and we claim that replacing~$q(k)$ by~$q'$ creates an occurrence of~$\sigma(P,e)$ in~$\pi$.
To verify this we first observe that~$(i,k-1)\notin C$ by condition~(iii) of friendly mesh patterns.
Secondly, if~$(a-1,k-1)\in C$ for some $a\in L(k)\setminus c_L(k)$, then by~\eqref{eq:Cip} we have~$(a,k-1)\in G(\tau(P))$.
This implies that~$R_{a-1,k}\cap G(\pi)=\emptyset$, otherwise a point in this region would form an occurrence of~231 together with the points~$q(k)$ and~$q(k-1)$.
Thirdly, if~$(i,a)\in C$ for some $a\in L(k)$, then by~\eqref{eq:Ci} we have~$a\in B_R^-(c_L(k))$.
This implies that~$R_{i-1,a}\cap G(\pi)=\emptyset$, otherwise a point in this region would form an occurrence of~231 together with the points~$q(k)$ and~$q(a)$.
This completes the proof.
\end{proof}

\begin{figure}[h!]
\includegraphics[page=2]{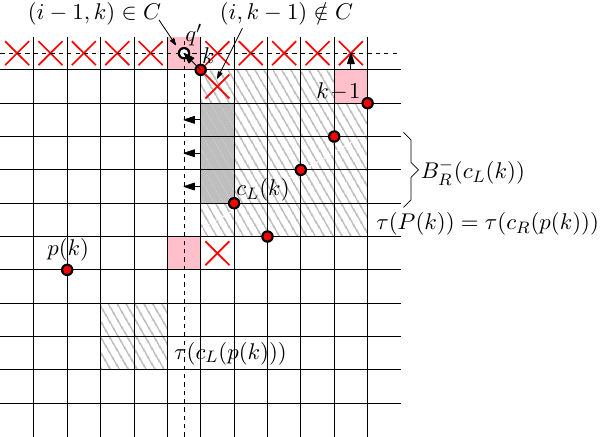}
\caption{Exchange argument in the proof of Lemma~\ref{lem:exchange}.}
\label{fig:exchange}
\end{figure}

\subsection{Proof of Theorem~\ref{thm:algoS}}

We are now in position to prove Theorem~\ref{thm:algoS}.

\begin{proof}[Proof of Theorem~\ref{thm:algoS}]
Lemma~\ref{lem:tame-classical} shows that the classical pattern~231 is tame.
As $(P_i,e_i)$, $i\in [\ell]$, are friendly tree patterns, the mesh patterns~$\sigma^-(P_i,e_i)$, $i\in [\ell]$, are tame by Lemma~\ref{lem:friendly-tame}.
Combining Lemma~\ref{lem:intersection} and Lemma~\ref{lem:exchange} yields that
\[L_n:=S_n\big(231,\sigma(P_1,e_1),\ldots,\sigma(P_\ell,e_\ell)\big)=S_n\big(231,\sigma^-(P_1,e_1),\ldots,\sigma^-(P_\ell,e_\ell)\big)\]
for $n\geq 0$ is a hereditary sequence of zigzag languages.
Theorem~\ref{thm:jump} thus guarantees that Algorithm~J initialized with the identity permutation~$\ide_n$ visits every permutation of~$L_n$ exactly once.
We now show that Algorithm~S is the preimage of Algorithm~J under the mapping~$\tau$, which is a bijection between~$\cL_n:=\cT_n((P_1,e_1),\ldots,(P_\ell,e_\ell))$ and~$L_n$ by Theorem~\ref{thm:bijection}.
It was shown in~\cite[Sec.~3.3]{MR4391718} that minimal jumps in~$S_n(231)$ are in one-to-one correspondence with tree rotations in~$\cT_n$.
Specifically, a minimal left jump of a value~$j$ in the permutation corresponds to an up-rotation of~$j$ in the binary tree, and a minimum right jump of~$j$ corresponds to a down-rotation of~$j$.
Consequently, minimal jumps in~$L_n$ are in one-to-one correspondence with minimal slides in~$\cL_n$.
This completes the proof of the theorem.
\end{proof}

\subsection{Efficient implementation}
\label{sec:algo-eff}

We now describe an efficient implementation of Algorithm~S.
In particular, this implementation is \emph{history-free}, i.e., it does not require to store all previously visited binary trees, but only maintains the current tree in memory.
Algorithm~H below is a straightforward translation of the history-free Algorithm~M for zigzag languages presented in~\cite{MR4598046} from permutations to binary trees.

\begin{algo}{Algorithm~H}{History-free minimal slides}
For friendly tree patterns~$(P_1,e_1),\ldots,(P_\ell,e_\ell)$, this algorithm generates all binary trees from~$\cT_n$ that avoid~$(P_1,e_1),\ldots,(P_\ell,e_\ell)$, i.e., the set $\cL_n:=\cT_n((P_1,e_1),\ldots,(P_\ell,e_\ell))\seq\cT_n$ by minimal slides in the same order as Algorithm~S.
It maintains the current tree in the variable~$T$, and auxiliary arrays $o=(o_1,\ldots,o_n)$ and $s=(s_1,\ldots,s_n)$.
\begin{enumerate}[label={\bfseries H\arabic*.}, leftmargin=8mm, noitemsep, topsep=3pt plus 3pt]
\item{} [Initialize] Set $T\gets \tau^{-1}(\ide_n)$, and $o_j\gets \diru$, $s_j\gets j$ for $j=1,\ldots,n$.
\item{} [Visit] Visit the current binary tree~$T$.
\item{} [Select vertex] Set $j\gets s_n$, and terminate if $j=1$.
\item{} [Slide] In the current binary tree~$T$, perform a slide of the vertex~$j$ that is minimal w.r.t.~$\cL_n$, where the slide direction is up if $o_j=\diru$ and down if~$o_j=\dird$.
\item{} [Update $o$ and~$s$] Set $s_n\gets n$.
If $o_j=\diru$ and $j$ is either the root or its parent is larger than~$j$ set $o_j=\dird$, or if $o_j=\dird$ and $j$ has no left child set $o_j=\diru$, and in both cases set $s_j\gets s_{j-1}$ and $s_{j-1}= j-1$. Go back to H2.
\end{enumerate}
\end{algo}

The two auxiliary arrays used by Algorithm~H store the following information.
The direction in which vertex~$j$ slides in the next step is maintained in the variable~$o_j$.
Furthermore, the array~$s$ is used to determine the vertex that slides in the next step.
Specifically, the vertex~$j$ that slides in the next steps is retrieved from the last entry of the array~$s$ in step~H3, by the instruction~$j\gets s_n$.

The running time per iteration of the algorithm is governed by the time it takes to compute a minimal slide in step~H4.
This boils down to testing containment of the tree patterns~$(P_i,e_i)$, $i\in[\ell]$, in~$T$.

\begin{theorem}
\label{thm:algoH}
Let $(P_1,e_1),\ldots,(P_\ell,e_\ell)$ be friendly tree patterns with $P_i\in \cT_{k_i}$ for $i\in[\ell]$.
Then Algorithm~H visits every binary tree from~$\cT_n((P_1,e_1),\ldots,(P_\ell,e_\ell))$ exactly once, in the same order as Algorithm~S, in time~$\cO(n^2\sum_{i=1}^\ell k_i^2)$ per binary tree.
\end{theorem}

\begin{proof}
The correctness of Algorithm~H follows from~\cite[Thm.~29]{MR4598046}.

For the running time, note that any slide consists of at most~$n$ rotations, and that testing whether~$T\in \cT_n$ contains the tree pattern~$(P_i,e_i)$, $P_i\in \cT_{k_i}$, can be done in time~$\cO(nk_i^2)$ by dynamic programming as follows.
We store for each vertex of~$T$ a table of size~$k_i$ that gives information for each vertex of~$P_i$ whether the corresponding subtree of~$T$ contains the corresponding subtree of~$P_i$ as a pattern.
In fact, we need two such tables, one for tracking `containment' and the other for tracking the stronger property `containment at the root'.
This information can be computed bottom-up in time~$\cO(k_i^2)$ for each of the $n$ vertices of~$T$ (cf.~\cite{MR662611}).
\end{proof}

For details, see our C++ implementation~\cite{cos_btree}.

\section{Equality of tree patterns}
\label{sec:equality}

It turns out that for some edges~$(i,p(i))$ in a tree pattern~$(P,e)$, it is irrelevant whether the edge is considered contiguous ($e(i)=1$) or non-contiguous ($e(i)=0$).
The following theorem captures these conditions formally, and it establishes that~$\cT_n(P,e)=\cT_n(P,e')$ for tree patterns~$(P,e)$ and~$(P,e')$ where $e$ and~$e'$ differ only in a single value.
Theorem~\ref{thm:e-equality} will be used heavily in the tables in the next section, where those `don't care' values of~$e$ are denoted by the hyphen~$\hyph$.
The statement and proof of this theorem is admittedly slightly technical, and we recommend to skip it on first reading.

\begin{figure}[t!]
\includegraphics[page=1]{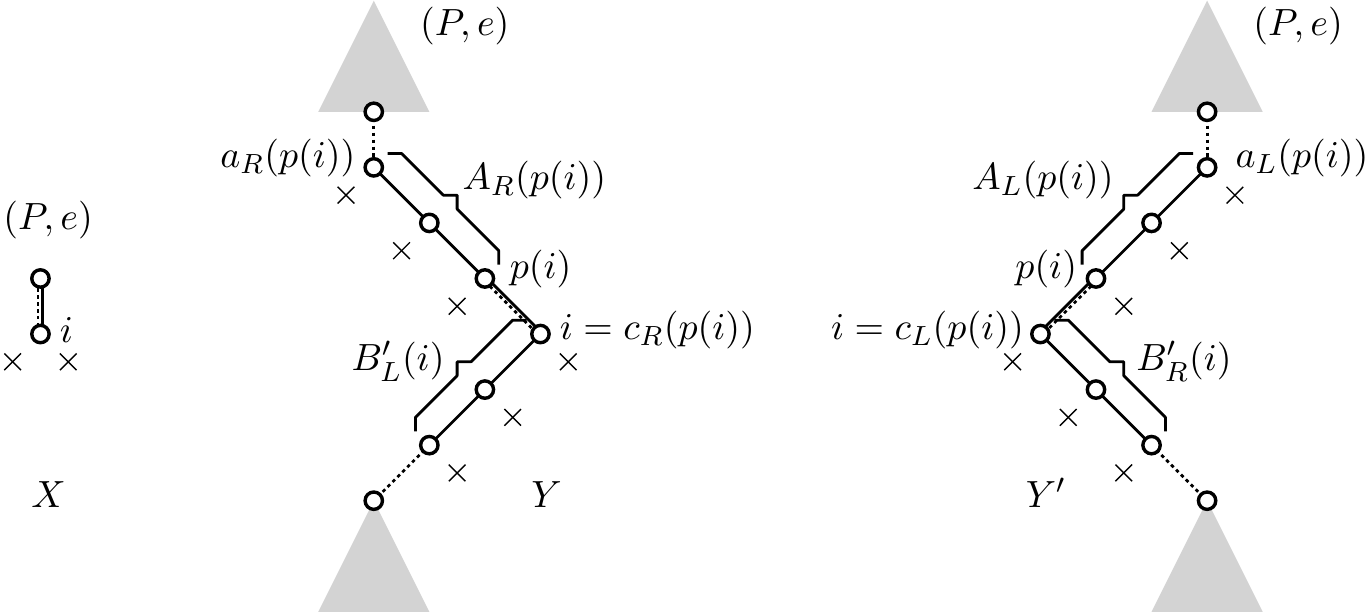}
\caption{Illustration of the definitions in~\eqref{eq:ABXY}.
The gray crosses indicate non-existing subtrees, i.e., subtrees that are empty~$\varepsilon$.
The edge~$(i,p(i))$ that can be both contiguous or non-contiguous by Theorem~\ref{thm:e-equality} is drawn as a double line that is half solid and half dotted.
The same convention is used in later figures.
Vertically drawn edges indicate that this can be a left edge or a right edge.}
\label{fig:XYY'}
\end{figure}

Let $(P,e)$, $P\in\cT_k$, be a tree pattern.
For any vertex~$i\in[k]$, we define
\begin{subequations}
\label{eq:ABXY}
\begin{equation}
\begin{split}
B_L'(i) &:= \big\{c_L^\ell(i)\mid \ell\geq 0\text{ and } e(c_L^j(i))=1 \text{ for all } j=1,\ldots,\ell\big\}, \\
A_R(i) &:= \big\{a\in[k]\mid i=c_R^\ell(a) \text{ for some } \ell\geq 0 \text{ and } e(c_R^j(a))=1 \text{ for all } j=1,\ldots,\ell\big\}.
\end{split}
\end{equation}

In words, $B_L'(i)$ are the descendants of~$i$ in~$P$ that are reachable from~$i$ along a path of contiguous left edges.
Furthermore, $A_R(i)$ are the predecessors of~$i$ in~$P$ that reach~$i$ along a path of contiguous right edges.
Both sets include the vertex~$i$ itself.
We also write $a_R(i)$ for the top vertex from~$A_R(i)$ in the tree~$P$.

We define analogous sets~$B_R'(i)$ and~$A_L(i)$ and vertices~$a_L(i)$ that are obtained by interchanging left and right in the definitions before.
Specifically, these sets are defined as
\begin{equation}
\begin{split}
B_R'(i) &:= \big\{c_R^\ell(i)\mid \ell\geq 0\text{ and } e(c_R^j(i))=1 \text{ for all } j=1,\ldots,\ell\big\}, \\
A_L(i) &:= \big\{a\in[k]\mid i=c_L^\ell(a) \text{ for some } \ell\geq 0 \text{ and } e(c_L^j(a))=1 \text{ for all } j=1,\ldots,\ell\big\},
\end{split}
\end{equation}
and $a_L(i)$ is defined as the top vertex from~$A_L(i)$ in the tree~$P$.
Using these definitions, we consider the following subsets~$X,Y,Y'\seq[k]$ of vertices of~$P$; see Figure~\ref{fig:XYY'}.
\begin{equation}
\label{eq:XY}
\begin{split}
X &:= \big\{i\in[k] \mid c_L(i)=\varepsilon \,\wedge\, c_R(i)=\varepsilon\big\}, \\
Y &:= \Big\{i\in[k] \mid c_L(i)\neq \varepsilon \,\wedge\, i=c_R(p(i)) \,\wedge\, \big(A_R(p(i))=\{p(i)\} \vee B_L'(i)=\{i\}\big) \,\wedge {} \\
  & \hspace{3cm} \big(c_L(j)=\varepsilon \text{ for all } j\in A_R(p(i))\big) \,\wedge\, \big(c_R(j)=\varepsilon \text{ for all } j\in B_L'(i)\big) \,\wedge {} \\
  & \hspace{3cm} \big(a_R(p(i))=r(P)\vee e(a_R(p(i)))=0\big)\Big\}, \\
Y' &:= \Big\{i\in[k] \mid c_R(i)\neq \varepsilon \,\wedge\, i=c_L(p(i)) \,\wedge\, \big(A_L(p(i))=\{p(i)\} \vee B_R'(i)=\{i\}\big) \,\wedge {} \\
  & \hspace{3cm} \big(c_R(j)=\varepsilon \text{ for all } j\in A_L(p(i))\big) \,\wedge\, \big(c_L(j)=\varepsilon \text{ for all } j\in B_R'(i)\big) \,\wedge {} \\
  & \hspace{3cm} \big(a_L(p(i))=r(P)\vee e(a_L(p(i)))=0\big)\Big\}.
\end{split}
\end{equation}
\end{subequations}

Note that $X$ is simply the set of all leaves of~$P$.
The six conditions in the conjunction that defines~$Y$ express the following facts: (1) $i$ has a left child; (2) $i$ is a right child of its parent~$p(i)$; (3) one of the sets~$A_R(p(i))$ or~$B_L'(i)$ is trivial; (4) no vertex in~$A_R(p(i))$ has a left child; (5) no vertex in~$B_L'(i)$ (including~$i$ itself) has a right child; (6) the top vertex~$a_R(p(i))$ in~$A_R(p(i))$ is either the root or the edge to its parent is non-contiguous.
The definition of~$Y'$ is analogous, but with left and right interchanged.

\begin{theorem}
\label{thm:e-equality}
Let $(P,e)$, $P\in\cT_k$, be a tree pattern, and let $X,Y,Y'$ be the sets of vertices in~$P$ defined in~\eqref{eq:ABXY} w.r.t.~$(P,e)$.
Furthermore, let $i\in X\cup Y\cup Y'$ be a vertex of~$P$ with $e(i)=0$, and define $e'\colon [k]\setminus r(P)\rightarrow \{0,1\}$ by $e'(i):=1$ and $e'(j):=e(j)$ for all $j\in[k]\setminus i$.
Then we have $\cT_n(P,e)=\cT_n(P,e')$.
\end{theorem}

Note that $(P,e')$ differs from $(P,e)$ in that the edge~$(i,p(i))$ from $i$ to its parent~$p(i)$ is contiguous instead of non-contiguous.

\begin{proof}
It suffices to show that if~$T\in\cT_n$ contains~$(P,e)$, then $T$ contains~$(P,e')$.
Let $T\in\cT_n$ and consider an occurrence of~$(P,e)$ in~$T$, witnessed by an injection~$f\colon [k]\rightarrow[n]$ that satisfies the conditions stated in Section~\ref{sec:pat-tree}.

We consider the cases~$i\in X$, $i\in Y$, or~$i\in Y'$ separately.
The last two are symmetric, so it suffices to consider whether~$i\in X$ or~$i\in Y$.

\textbf{Case~(a):} $i\in X$.
As $e(i)=0$, $f(i)$ is a descendant of~$f(p(i))$ in~$T$.
Instead of mapping~$i$ to~$f(i)$ in~$T$, we remap it to the corresponding direct child of~$f(p(i))$ in~$T$.
Specifically, if $i=c_L(p(i))$ in~$P$, then we remap~$i$ to~$c_L(f(p(i)))$ in~$T$, and if $i=c_R(p(i))$ in~$P$, then we remap~$i$ to~$c_R(f(p(i)))$ in~$T$.
This is possible as $i$ is a leaf in~$P$.
This shows that $T$ contains~$(P,e')$, as claimed.

\textbf{Case~(b):} $i\in Y$.
We distinguish the subcases~$A_R(p(i))=\{p(i)\}$ and~$B_L'(i)=\{i\}$, at least one of which must hold by the definition of~$Y$ in~\eqref{eq:XY}.
The arguments in these two cases are illustrated in Figure~\ref{fig:e-equality}.

\begin{figure}
\centerline{
\includegraphics[page=2]{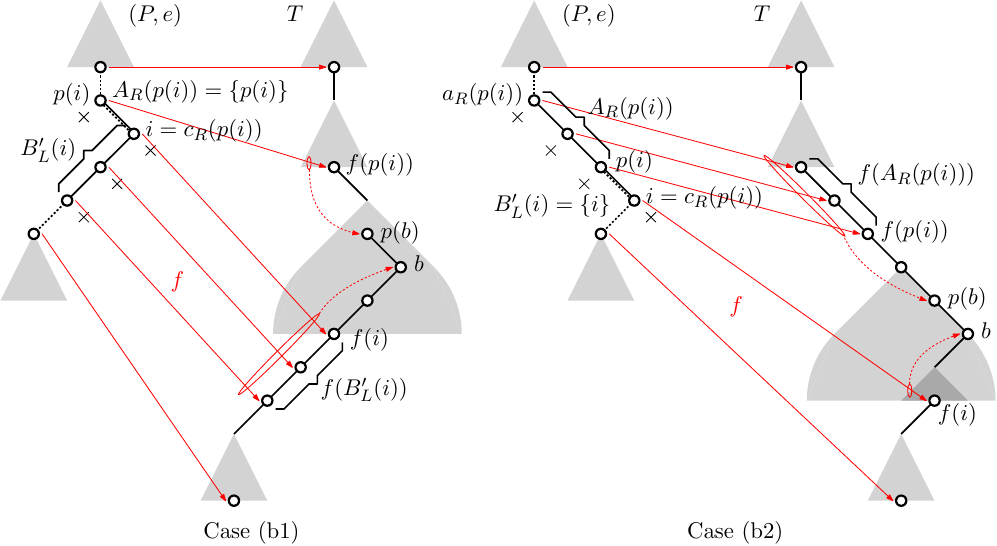}
}
\caption{Illustration of the proof of Theorem~\ref{thm:e-equality}.
The dashed arrows indicate the remapping argument.}
\label{fig:e-equality}
\end{figure}

\textbf{Case~(b1):} $A_R(p(i))=\{p(i)\}$, i.e., $p(i)$ is the root or the edge from~$p(i)$ to its parent is non-contiguous.
We consider the vertex~$f(i)$ in~$T$, and we let~$b$ be the top vertex in~$T$ such that~$f(i)\in B_L(b)$.
We remap~$p(i)$ to~$p(b)$ in~$T$, and we remap $B_L'(i)$ (including $i$ itself) to~$b$ and its direct left descendants.
By the definition~\eqref{eq:XY}, we know that in~$P$ we have $c_L(p(i))=\varepsilon$, $c_R(j)=\varepsilon$ for all $j\in B_L'(i)$, and $p(i)=r(P)$ or $e(p(i))=0$, which ensures that the remapping indeed witnesses an occurrence of~$(P,e')$.

\textbf{Case~(b2):} $B_L'(i)=\{i\}$, i.e., the edge from~$i$ to its left child is non-contiguous.
We consider the vertex~$f(p(i))$ in~$T$, and we let~$b$ be the vertex in~$T$ such that~$b\in B_R(f(p(i)))$ and~$f(i)\in L(b)$.
We remap~$A_R(p(i))$ (including~$p(i)$ itself) to~$p(b)$ and its direct ancestors, and we remap $f(i)$ to~$b$.
By the definition~\eqref{eq:XY}, we know that in~$P$ we have $c_L(j)=\varepsilon$ for all $j\in A_R(p(i))$, $c_R(i)=\varepsilon$, and $a_R(p(i))=r(P)$ or $e(a_R(p(i)))=0$, which ensures that the remapping indeed witnesses an occurrence of~$(P,e')$.
\end{proof}

\section{Tree patterns on at most 5 vertices}
\label{sec:count}

In order to mine interesting conjectures about tree pattern avoidance\footnote{and to pretend that we are doing `big data'}, we conducted systematic computer experiments with all tree patterns~$(P,e)$ on $k=3,4,5$ vertices; see Tables~\ref{tab:size-3}, \ref{tab:size-4} and~\ref{tab:size-5}, respectively.
Specifically, we computed the corresponding counting sequences~$|\cT_n(P,e)|$ for $n=1,\ldots,12$, and searched for matches within the OEIS~\cite{oeis}.
Those counts were computed using Algorithm~H for friendly tree patterns, and via brute-force methods for non-friendly tree patterns.
As mirrored tree patterns are Wilf-equivalent, our tables only contain the lexicographically smaller of any such pair of mirrored trees, using the compact encoding described in Section~\ref{sec:pat-tree}.
Some of the $e$-sequences contain `don't care' entries~$\hyph$, which means that both possible $e$-values~0 or~1 yield the same sets of pattern-avoiding trees by Theorem~\ref{thm:e-equality}.

The last column contains a reference to a proof that the counting sequence is indeed the listed OEIS entry.
The Wilf-equivalence column contains pointers to a Wilf-equivalent tree pattern in Figures~\ref{fig:wilf3}--\ref{fig:wilf5} that has an established OEIS entry.
The question mark at the pattern~$(31245,1\hyph0\hyph)$ in Table~\ref{tab:size-5} means that the match has not been proved formally for all~$n$, but was only verified experimentally for~$n\leq 12$ (so this is an open problem).
For every other pattern we have a reference or a pointer to a Wilf-equivalent pattern that has a reference, unless it is one of the three new counting sequences denoted by~\newa{}, \newb{}, and~\newc{}, which we added to the OEIS using the sequence numbers \href{https://oeis.org/A365508}{A365508}, \href{https://oeis.org/A365509}{A365509}, and~\href{https://oeis.org/A365510}{A365510}, respectively.

\begin{table}[h!]
\caption{Tree patterns with 3 vertices.
See Section~\ref{sec:count} for explanations.}
\label{tab:size-3}
\tiny
\setlength\tabcolsep{1pt}
\centerline{
\begin{tabular}{l|l|l|rrrrrrrrrrrrr|l|l|l}
$P$ & $e$ & Friendly & \multicolumn{13}{c|}{Counts $|\cT_n(P, e)|$ for $n=1,\ldots,12$} & OEIS & Wilf-equivalence & References \\\hline
123 & 0\hyph & & 1&2&4&8&16&32&64&128&256&512&1024&2048& $\dots$ & \textcolor{red} {\href{https://oeis.org/A000079}{A000079}} & & \cite[Thm.~1]{MR2967227}; Sec.~\ref{sec:123-bin} \\
 & 1\hyph & & 1&2&4&9&21&51&127&323&835&2188&5798&15511& $\dots$ & \textcolor{blue}{\href{https://oeis.org/A001006}{A001006}} & & \cite[Class~4.1]{MR2645188}; Sec.~\ref{sec:123-motzkin}; Tab.~\ref{tab:staggered} \\
\hline
132 & \hyph\hyph & 0\hyph, \hyph0 & 1&2&4&8&16&32&64&128&256&512&1024&2048 &$\dots$ & \textcolor{red}{\href{https://oeis.org/A000079}{A000079}} & $(123,0\hyph)$ Lem.~\ref{lem:subtree} & \cite[Class~4.2]{MR2645188}; \cite[Thm.~1]{MR2967227}; Sec.~\ref{sec:132-bin}; Tab.~\ref{tab:staggered} \\
\hline
213 & \hyph\hyph & & 1&2&4&8&16&32&64&128&256&512&1024&2048&$\dots$ & \textcolor{red}{\href{https://oeis.org/A000079}{A000079}} & $(132,\hyph\hyph)$ Lem.~\ref{lem:pathmove} & \cite[Class~4.2]{MR2645188}; \cite[Thm.~1]{MR2967227}; Sec.~\ref{sec:213-bin}; Tab.~\ref{tab:staggered} \\
\hline
\end{tabular}
}
\end{table}

\begin{table}[h!]
\caption{Tree patterns with 4 vertices.}
\label{tab:size-4}
\tiny
\setlength\tabcolsep{1pt}
\centerline{
\begin{tabular}{l|l|l|rrrrrrrrrrrrr|l|l|l}
$P$ & $e$ & Friendly & \multicolumn{13}{c|}{Counts $|\cT_n(P, e)|$ for $n=1,\ldots,12$} & OEIS & Wilf-equivalence & References \\\hline
1234 & 00\hyph & & 1&2&5&13&34&89&233&610&1597&4181&10946&28657&$\dots$ & \textcolor{red}{\href{https://oeis.org/A001519}{A001519}} & & \cite[Thm.~1]{MR2967227} \\
 & 01\hyph & & 1&2&5&13&35&96&267&750&2123&6046&17303&49721&$\dots$ & \textcolor{green}{\href{https://oeis.org/A005773}{A005773}} & $(1432,\hyph1\hyph)$ Lem.~\ref{lem:subtree} & \\
 & 10\hyph & & 1&2&5&13&35&97&275&794&2327&6905&20705&62642&$\dots$ & \textcolor{blue}{\href{https://oeis.org/A025242}{A025242}} & $(1243,1\hyph\hyph)$ Lem.~\ref{lem:subtree} & \\
 & 11\hyph & & 1&2&5&13&36&104&309&939&2905&9118&28964&92940&$\dots$ & \href{https://oeis.org/A036765}{A036765} & & \cite[Class~5.1]{MR2645188}; Tab.~\ref{tab:staggered} \\
\hline
1243 & 0\hyph\hyph & 00\hyph, 0\hyph0 & 1&2&5&13&34&89&233&610&1597&4181&10946&28657&$\dots$ & \textcolor{red}{\href{https://oeis.org/A001519}{A001519}} & $(1234,00\hyph)$ Lem.~\ref{lem:subtree} & \cite[Thm.~1]{MR2967227} \\
 & 1\hyph\hyph & & 1&2&5&13&35&97&275&794&2327&6905&20705&62642&$\dots$ & \textcolor{blue}{\href{https://oeis.org/A025242}{A025242}} & & \cite[Class~5.2]{MR2645188}; Thm.~\ref{thm:motz24}; Tab.~\ref{tab:staggered} \\
\hline
1324 & 0\hyph\hyph & & 1&2&5&13&34&89&233&610&1597&4181&10946&28657&$\dots$ & \textcolor{red}{\href{https://oeis.org/A001519}{A001519}} & $(1234,00\hyph)$ Lem.~\ref{lem:subtree} & \cite[Thm.~1]{MR2967227} \\
 & 1\hyph\hyph & & 1&2&5&13&35&97&275&794&2327&6905&20705&62642&$\dots$ & \textcolor{blue}{\href{https://oeis.org/A025242}{A025242}} & $(1243,1\hyph\hyph)$ Lem.~\ref{lem:pathmove} & \cite[Class~5.2]{MR2645188}; Thm.~\ref{thm:motz24}; Tab.~\ref{tab:staggered} \\
\hline
1423 & 0\hyph\hyph, \hyph0\hyph & 0\hyph\hyph, \hyph0\hyph & 1&2&5&13&34&89&233&610&1597&4181&10946&28657&$\dots$ & \textcolor{red}{\href{https://oeis.org/A001519}{A001519}} & $(1234,00\hyph)$ Lem.~\ref{lem:subtree} & \cite[Thm.~1]{MR2967227}; Tab.~\ref{tab:staggered} \\
 & 11\hyph & & 1&2&5&13&35&97&275&794&2327&6905&20705&62642&$\dots$ & \textcolor{blue}{\href{https://oeis.org/A025242}{A025242}} & $(1324,1\hyph\hyph)$  Lem.~\ref{lem:Lmove3} & \cite[Class~5.2]{MR2645188} \\
\hline
1432 & \hyph0\hyph & \hyph0\hyph & 1&2&5&13&34&89&233&610&1597&4181&10946&28657&$\dots$ & \textcolor{red}{\href{https://oeis.org/A001519}{A001519}} & $(1234,00\hyph)$ Lem.~\ref{lem:subtree} & \cite[Thm.~1]{MR2967227} \\
 & \hyph1\hyph & 01\hyph & 1&2&5&13&35&96&267&750&2123&6046&17303&49721&$\dots$ & \textcolor{green}{\href{https://oeis.org/A005773}{A005773}} & & \cite[Class~5.3]{MR2645188}; Sec.~\ref{sec:1432-motzkin}; Tab.~\ref{tab:staggered} \\
\hline
2134 & \hyph0\hyph & & 1&2&5&13&34&89&233&610&1597&4181&10946&28657&$\dots$ & \textcolor{red}{\href{https://oeis.org/A001519}{A001519}} & $(1432,\hyph0\hyph)$ Lem.~\ref{lem:pathmove} & \cite[Thm.~1]{MR2967227} \\
 & \hyph1\hyph & & 1&2&5&13&35&97&275&794&2327&6905&20705&62642&$\dots$ & \textcolor{blue}{\href{https://oeis.org/A025242}{A025242}} & $(1243,1\hyph\hyph)$ Lem.~\ref{lem:pathmove} & \cite[Class~5.2]{MR2645188}; Thm.~\ref{thm:motz24} \\
\hline
2143 & \hyph0\hyph & \hyph0\hyph & 1&2&5&13&34&89&233&610&1597&4181&10946&28657&$\dots$ & \textcolor{red}{\href{https://oeis.org/A001519}{A001519}} & $(1423,\hyph0\hyph)$ Lem.~\ref{lem:pathmove} & \cite[Thm.~1]{MR2967227}; Tab.~\ref{tab:staggered} \\
 & \hyph1\hyph & \hyph10 & 1&2&5&13&35&97&275&794&2327&6905&20705&62642&$\dots$ & \textcolor{blue}{\href{https://oeis.org/A025242}{A025242}} &  &\cite[Class~5.2]{MR2645188} \\
\end{tabular}
}
\end{table}

\begin{table}
\caption{Tree patterns with 5 vertices.}
\label{tab:size-5}
\tiny
\setlength\tabcolsep{1pt}
\centerline{
\begin{tabular}{l|l|l|rrrrrrrrrrrrr|l|l|l}
$P$ & $e$ & Friendly & \multicolumn{13}{c|}{Counts $|\cT_n(P, e)|$ for $n=1,\ldots,12$} & OEIS & Wilf-equivalence & References \\\hline
12345 & 000\hyph & & 1&2&5&14&41&122&365&1094&3281&9842&29525&88574&$\dots$ & \textcolor{red} {\href{https://oeis.org/A007051}{A007051}} & & \cite[Thm.~1]{MR2967227} \\
 & 001\hyph & & 1&2&5&14&41&123&374&1147&3538&10958&34042&105997&$\dots$ & \textcolor{blue}{\href{https://oeis.org/A054391}{A054391}} & $(12543,0\hyph1\hyph)$ Lem.~\ref{lem:subtree} & \\
 & 010\hyph & & 1&2&5&14&41&123&375&1157&3603&11304&35683&113219&$\dots$ & \newa{}\textrightarrow\href{https://oeis.org/A365508}{A365508} & & \\
 & 011\hyph & & 1&2&5&14&41&124&384&1210&3865&12482&40677&133572&$\dots$ & \textcolor{olive}{\href{https://oeis.org/A159772}{A159772}} & $(15432,\hyph11\hyph)$ Lem.~\ref{lem:subtree} &  \\
 & 100\hyph & & 1&2&5&14&41&123&375&1158&3615&11393&36209&115940&$\dots$ & \textcolor{green}{\href{https://oeis.org/A176677}{A176677}} & $(12354,10\hyph\hyph)$ Lem.~\ref{lem:subtree} & \\
 & 101\hyph & & 1&2&5&14&41&124&383&1202&3819&12255&39651&129190&$\dots$ & \newb{}\textrightarrow\href{https://oeis.org/A365509}{A365509} & & \\
 & 110\hyph & & 1&2&5&14&41&124&385&1221&3939&12886&42648&142544&$\dots$ & \textcolor{teal}{\href{https://oeis.org/A159768}{A159768}} & $(12354,11\hyph\hyph)$ Lem.~\ref{lem:subtree} & \\
 & 111\hyph & & 1&2&5&14&41&125&393&1265&4147&13798&46476&158170&$\dots$ & \href{https://oeis.org/A036766}{A036766} & & \cite[Class~6.1]{MR2645188}; Tab.~\ref{tab:staggered} \\
\hline
12354 & 00\hyph\hyph & 000\hyph, 00\hyph0 & 1&2&5&14&41&122&365&1094&3281&9842&29525&88574&$\dots$ & \textcolor{red}{\href{https://oeis.org/A007051}{A007051}} & & \cite[Thm.~1]{MR2967227} \\
 & 01\hyph\hyph & & 1&2&5&14&41&123&375&1157&3603&11304&35683&113219&$\dots$ & \newa{}\textrightarrow\href{https://oeis.org/A365508}{A365508} & $(12345,010\hyph)$ Lem.~\ref{lem:subtree} & \\
 & 10\hyph\hyph & & 1&2&5&14&41&123&375&1158&3615&11393&36209&115940&$\dots$ & \textcolor{green}{\href{https://oeis.org/A176677}{A176677}} & $(12435,10\hyph\hyph)$ Lem.~\ref{lem:subtree} & \\
 & 11\hyph\hyph & & 1&2&5&14&41&124&385&1221&3939&12886&42648&142544&$\dots$ & \textcolor{teal}{\href{https://oeis.org/A159768}{A159768}} & & \cite[Class~6.2]{MR2645188}; Tab.~\ref{tab:staggered} \\
\hline
12435 & 00\hyph\hyph & & 1&2&5&14&41&122&365&1094&3281&9842&29525&88574&$\dots$ & \textcolor{red}{\href{https://oeis.org/A007051}{A007051}} & & \cite[Thm.~1]{MR2967227} \\
 & 01\hyph\hyph & & 1&2&5&14&41&123&375&1157&3603&11304&35683&113219&$\dots$ & \newa{}\textrightarrow\href{https://oeis.org/A365508}{A365508} & $(12345,010\hyph)$ Lem.~\ref{lem:subtree} & \\
 & 10\hyph\hyph & & 1&2&5&14&41&123&375&1158&3615&11393&36209&115940&$\dots$ & \textcolor{green}{\href{https://oeis.org/A176677}{A176677}} & & Lem.~\ref{lem:12435} \\
 & 11\hyph\hyph & & 1&2&5&14&41&124&385&1221&3939&12886&42648&142544&$\dots$ & \textcolor{teal}{\href{https://oeis.org/A159768}{A159768}} & $(12354,11\hyph\hyph)$ Lem.~\ref{lem:pathmove} & \cite[Class~6.2]{MR2645188}; Tab.~\ref{tab:staggered} \\
\hline
12534 & 00\hyph\hyph, 0\hyph0\hyph & 00\hyph\hyph, 0\hyph0\hyph & 1&2&5&14&41&122&365&1094&3281&9842&29525&88574&$\dots$ & \textcolor{red}{\href{https://oeis.org/A007051}{A007051}} & & \cite[Thm.~1]{MR2967227} \\
 & 011\hyph & & 1&2&5&14&41&123&375&1157&3603&11304&35683&113219&$\dots$ & \newa{}\textrightarrow\href{https://oeis.org/A365508}{A365508} & $(12345,010\hyph)$ Lem.~\ref{lem:subtree} & \\
 & 10\hyph\hyph, 1\hyph0\hyph & & 1&2&5&14&41&123&375&1158&3615&11393&36209&115940&$\dots$ & \textcolor{green}{\href{https://oeis.org/A176677}{A176677}} & $(12435,10\hyph\hyph)$ Lem.~\ref{lem:subtree} & Tab.~\ref{tab:staggered} \\
 & 111\hyph & & 1&2&5&14&41&124&384&1212&3885&12614&41400&137132&$\dots$ & \textcolor{violet}{\href{https://oeis.org/A159769}{A159769}} & & \cite[Class~6.3]{MR2645188} \\
\hline
12543 & 0\hyph0\hyph & 0\hyph0\hyph & 1&2&5&14&41&122&365&1094&3281&9842&29525&88574&$\dots$ & \textcolor{red}{\href{https://oeis.org/A007051}{A007051}} & & \cite[Thm.~1]{MR2967227} \\
 & 0\hyph1\hyph & 001\hyph & 1&2&5&14&41&123&374&1147&3538&10958&34042&105997&$\dots$ & \textcolor{blue}{\href{https://oeis.org/A054391}{A054391}} & $(15234,0\hyph1\hyph)$ Lem.~\ref{lem:subtree} & \\
 & 1\hyph0\hyph & & 1&2&5&14&41&123&375&1158&3615&11393&36209&115940&$\dots$ & \textcolor{green}{\href{https://oeis.org/A176677}{A176677}} & $(12534,1\hyph0\hyph)$ Lem.~\ref{lem:subtree} & \\
 & 101\hyph & & 1&2&5&14&41&124&383&1202&3819&12255&39651&129190&$\dots$ & \newb{}\textrightarrow\href{https://oeis.org/A365509}{A365509} & $(12345,101\hyph)$ Lem.~\ref{lem:subtree} & \\
 & 111\hyph & & 1&2&5&14&41&124&384&1211&3875&12548&41040&135370&$\dots$ & \textcolor{magenta}{\href{https://oeis.org/A159770}{A159770}} & & \cite[Class~6.4]{MR2645188} \\
\hline
13245 & 0\hyph0\hyph & & 1&2&5&14&41&122&365&1094&3281&9842&29525&88574&$\dots$ & \textcolor{red}{\href{https://oeis.org/A007051}{A007051}} & & \cite[Thm.~1]{MR2967227} \\
 & 0\hyph1\hyph & & 1&2&5&14&41&123&375&1157&3603&11304&35683&113219&$\dots$ & \newa{}\textrightarrow\href{https://oeis.org/A365508}{A365508} & $(12345,010\hyph)$ Lem.~\ref{lem:subtree} & \\
 & 1\hyph0\hyph & & 1&2&5&14&41&123&376&1168&3678&11716&37688&122261&$\dots$ & \newc{}\textrightarrow\href{https://oeis.org/A365510}{A365510} & & \\
 & 1\hyph1\hyph & & 1&2&5&14&41&124&385&1221&3939&12886&42648&142544&$\dots$ & \textcolor{teal}{\href{https://oeis.org/A159768}{A159768}} & $(12435,11\hyph\hyph)$ Lem.~\ref{lem:pathmove} & \cite[Class~6.2]{MR2645188}; Tab.~\ref{tab:staggered} \\
\hline
13254 & 0\hyph0\hyph & 0\hyph0\hyph & 1&2&5&14&41&122&365&1094&3281&9842&29525&88574&$\dots$ & \textcolor{red}{\href{https://oeis.org/A007051}{A007051}} & & \cite[Thm.~1]{MR2967227} \\
 & 0\hyph1\hyph & 0\hyph10 & 1&2&5&14&41&123&375&1157&3603&11304&35683&113219&$\dots$ & \newa{}\textrightarrow\href{https://oeis.org/A365508}{A365508} & $(12345,010\hyph)$ Lem.~\ref{lem:subtree} & \\
 & 1\hyph0\hyph & & 1&2&5&14&41&123&376&1168&3678&11716&37688&122261&$\dots$ & \newc{}\textrightarrow\href{https://oeis.org/A365510}{A365510} & $(13245,1\hyph0\hyph)$ Lem.~\ref{lem:subtree} & \\
 & 1\hyph1\hyph & & 1&2&5&14&41&124&385&1220&3929&12822&42309&140922&$\dots$ & \textcolor{brown}{\href{https://oeis.org/A159771}{A159771}} & $(21435,\hyph1\hyph\hyph)$ Lem.~\ref{lem:21435-13254} & \cite[Class~6.5]{MR2645188}; Thm.~\ref{thm:motz25} \\
\hline
14235 & 00\hyph\hyph & & 1&2&5&14&41&122&365&1094&3281&9842&29525&88574&$\dots$ & \textcolor{red}{\href{https://oeis.org/A007051}{A007051}} & & \cite[Thm.~1]{MR2967227} \\
 & 01\hyph\hyph & & 1&2&5&14&41&123&375&1157&3603&11304&35683&113219&$\dots$ & \newa{}\textrightarrow\href{https://oeis.org/A365508}{A365508} & $(12345,010\hyph)$ Lem.~\ref{lem:subtree} & \\
 & 10\hyph\hyph & & 1&2&5&14&41&123&375&1158&3615&11393&36209&115940&$\dots$ & \textcolor{green}{\href{https://oeis.org/A176677}{A176677}} & $(14325,10\hyph\hyph)$ Lem.~\ref{lem:subtree} & Tab.~\ref{tab:staggered} \\
 & 11\hyph\hyph & & 1&2&5&14&41&124&384&1212&3885&12614&41400&137132&$\dots$ & \textcolor{violet}{\href{https://oeis.org/A159769}{A159769}} & & \cite[Class~6.3]{MR2645188} \\
\hline
14325 & 00\hyph\hyph & & 1&2&5&14&41&122&365&1094&3281&9842&29525&88574&$\dots$ & \textcolor{red}{\href{https://oeis.org/A007051}{A007051}} & & \cite[Thm.~1]{MR2967227}\\
 & 01\hyph\hyph & & 1&2&5&14&41&123&375&1157&3603&11304&35683&113219&$\dots$ & \newa{}\textrightarrow\href{https://oeis.org/A365508}{A365508} & $(12345,010\hyph)$ Lem.~\ref{lem:subtree} & \\
 & 10\hyph\hyph & & 1&2&5&14&41&123&375&1158&3615&11393&36209&115940&$\dots$ & \textcolor{green}{\href{https://oeis.org/A176677}{A176677}} & $(12543,1\hyph0\hyph)$ Lem.~\ref{lem:pathmove} & \\
 & 11\hyph\hyph & & 1&2&5&14&41&124&384&1211&3875&12548&41040&135370&$\dots$ & \textcolor{magenta}{\href{https://oeis.org/A159770}{A159770}} & & \cite[Class~6.4]{MR2645188} \\
\hline
15234 & 0\hyph0\hyph, \hyph00\hyph & 0\hyph0\hyph, \hyph00\hyph & 1&2&5&14&41&122&365&1094&3281&9842&29525&88574&$\dots$ & \textcolor{red}{\href{https://oeis.org/A007051}{A007051}} & & \cite[Thm.~1]{MR2967227} \\
 & 0\hyph1\hyph, \hyph01\hyph & 0\hyph1\hyph, \hyph01\hyph & 1&2&5&14&41&123&374&1147&3538&10958&34042&105997&$\dots$ & \textcolor{blue}{\href{https://oeis.org/A054391}{A054391}} & $(15432,\hyph01\hyph)$ Lem.~\ref{lem:subtree} & Tab.~\ref{tab:staggered} \\
 & 110\hyph & & 1&2&5&14&41&123&376&1168&3678&11716&37688&122261&$\dots$ & \newc{}\textrightarrow\href{https://oeis.org/A365510}{A365510} & $(13245,1\hyph0\hyph)$ Lem.~\ref{lem:Lmove3} & \\
 & 111\hyph & & 1&2&5&14&41&124&384&1212&3885&12614&41400&137132&$\dots$ & \textcolor{violet}{\href{https://oeis.org/A159769}{A159769}} & & \cite[Class~6.3]{MR2645188} \\
\hline
15243 & \hyph0\hyph\hyph, 0\hyph0\hyph & \hyph0\hyph\hyph, 0\hyph0\hyph & 1&2&5&14&41&122&365&1094&3281&9842&29525&88574&$\dots$ & \textcolor{red}{\href{https://oeis.org/A007051}{A007051}} & & \cite[Thm.~1]{MR2967227}; Tab.~\ref{tab:staggered} \\
 & 011\hyph & 011\hyph & 1&2&5&14&41&123&375&1157&3603&11304&35683&113219&$\dots$ & \newa{}\textrightarrow\href{https://oeis.org/A365508}{A365508} & $(12345,010\hyph)$ Lem.~\ref{lem:subtree} & \\
 & 110\hyph & & 1&2&5&14&41&123&376&1168&3678&11716&37688&122261&$\dots$ & \newc{}\textrightarrow\href{https://oeis.org/A365510}{A365510} & $(15234,110\hyph)$ Lem.~\ref{lem:subtree} & \\
 & 111\hyph & & 1&2&5&14&41&124&385&1220&3929&12822&42309&140922&$\dots$ & \textcolor{brown}{\href{https://oeis.org/A159771}{A159771}} & & \cite[Class~6.5]{MR2645188} \\
\hline
 15324 & \hyph0\hyph\hyph & \hyph0\hyph\hyph & 1&2&5&14&41&122&365&1094&3281&9842&29525&88574&$\dots$ & \textcolor{red}{\href{https://oeis.org/A007051}{A007051}} & & \cite[Thm.~1]{MR2967227} \\
 & 01\hyph\hyph & 01\hyph\hyph & 1&2&5&14&41&123&375&1157&3603&11304&35683&113219&$\dots$ & \newa{}\textrightarrow\href{https://oeis.org/A365508}{A365508} & $(12345,010\hyph)$ Lem.~\ref{lem:subtree} & Tab.~\ref{tab:staggered} \\
 & 11\hyph\hyph & & 1&2&5&14&41&124&384&1211&3875&12548&41040&135370&$\dots$ & \textcolor{magenta}{\href{https://oeis.org/A159770}{A159770}} & & \cite[Class~6.4]{MR2645188} \\
\hline
15423 & \hyph0\hyph\hyph & \hyph0\hyph\hyph & 1&2&5&14&41&122&365&1094&3281&9842&29525&88574&$\dots$ & \textcolor{red}{\href{https://oeis.org/A007051}{A007051}} & & \cite[Thm.~1]{MR2967227} \\
 & 01\hyph\hyph & 01\hyph\hyph & 1&2&5&14&41&123&375&1157&3603&11304&35683&113219&$\dots$ & \newa{}\textrightarrow\href{https://oeis.org/A365508}{A365508} & $(12345,010\hyph)$ Lem.~\ref{lem:subtree} & Tab.~\ref{tab:staggered} \\
 & \hyph10\hyph & 010\hyph & 1&2&5&14&41&123&375&1157&3603&11304&35683&113219&$\dots$ & \newa{}\textrightarrow\href{https://oeis.org/A365508}{A365508} & $(12345,010\hyph)$ Lem.~\ref{lem:subtree} & \\
 & 111\hyph & & 1&2&5&14&41&124&384&1211&3875&12548&41040&135370&$\dots$ & \textcolor{magenta}{\href{https://oeis.org/A159770}{A159770}} & & \cite[Class~6.4]{MR2645188} \\
\hline
15432 & \hyph00\hyph & \hyph00\hyph & 1&2&5&14&41&122&365&1094&3281&9842&29525&88574&$\dots$ & \textcolor{red}{\href{https://oeis.org/A007051}{A007051}} & & \cite[Thm.~1]{MR2967227} \\
 & \hyph01\hyph & \hyph01\hyph & 1&2&5&14&41&123&374&1147&3538&10958&34042&105997&$\dots$ & \textcolor{blue}{\href{https://oeis.org/A054391}{A054391}} & $(21345,\hyph01\hyph)$ Lem.~\ref{lem:pathmove} & \\
 & \hyph10\hyph & 010\hyph & 1&2&5&14&41&123&375&1157&3603&11304&35683&113219&$\dots$ & \newa{}\textrightarrow\href{https://oeis.org/A365508}{A365508} & $(12345,010\hyph)$ Lem.~\ref{lem:subtree} & \\
 & \hyph11\hyph & 011\hyph & 1&2&5&14&41&124&384&1210&3865&12482&40677&133572&$\dots$ & \textcolor{olive}{\href{https://oeis.org/A159772}{A159772}} & & \cite[Class~6.6]{MR2645188}; Tab.~\ref{tab:staggered} \\
\hline
21345 & \hyph00\hyph & & 1&2&5&14&41&122&365&1094&3281&9842&29525&88574&$\dots$ & \textcolor{red}{\href{https://oeis.org/A007051}{A007051}} & & \cite[Thm.~1]{MR2967227} \\
 & \hyph01\hyph & & 1&2&5&14&41&123&374&1147&3538&10958&34042&105997&$\dots$ & \textcolor{blue}{\href{https://oeis.org/A054391}{A054391}} & $(21543,\hyph01\hyph)$ Lem.~\ref{lem:subtree} & \\
 & \hyph10\hyph & & 1&2&5&14&41&123&376&1168&3678&11716&37688&122261&$\dots$ & \newc{}\textrightarrow\href{https://oeis.org/A365510}{A365510} & $(15243,110\hyph) $ Lem.~\ref{lem:Lmove2} &  \\
 & \hyph11\hyph & & 1&2&5&14&41&124&385&1221&3939&12886&42648&142544&$\dots$ & \textcolor{teal}{\href{https://oeis.org/A159768}{A159768}} & $(13245,1\hyph1\hyph)$ Lem.~\ref{lem:pathmove} & \cite[Class~6.2]{MR2645188} \\
\hline
21354 & \hyph0\hyph\hyph & \hyph00\hyph, \hyph0\hyph0 & 1&2&5&14&41&122&365&1094&3281&9842&29525&88574&$\dots$ & \textcolor{red}{\href{https://oeis.org/A007051}{A007051}} & & \cite[Thm.~1]{MR2967227} \\
 & \hyph10\hyph & & 1&2&5&14&41&123&376&1168&3678&11716&37688&122261&$\dots$ & \newc{}\textrightarrow\href{https://oeis.org/A365510}{A365510} & $(21345,\hyph10\hyph)$ Lem.~\ref{lem:subtree} & \\
 & \hyph11\hyph & & 1&2&5&14&41&124&384&1212&3885&12613&41389&137055&$\dots$ & \href{https://oeis.org/A159773}{A159773} & & \cite[Class~6.7]{MR2645188} \\
\hline
21435 & \hyph0\hyph\hyph & & 1&2&5&14&41&122&365&1094&3281&9842&29525&88574&$\dots$ & \textcolor{red}{\href{https://oeis.org/A007051}{A007051}} & & \cite[Thm.~1]{MR2967227} \\
 & \hyph1\hyph\hyph & & 1&2&5&14&41&124&385&1220&3929&12822&42309&140922&$\dots$ & \textcolor{brown}{\href{https://oeis.org/A159771}{A159771}} &  $(13254,1\hyph1\hyph)$ Lem.~\ref{lem:21435-13254} & \cite[Class~6.5]{MR2645188}; Thm.~\ref{thm:motz25} \\
\hline
\end{tabular}
}
\end{table}

\begin{table}
\tiny
\setlength\tabcolsep{1pt}
\centerline{
\begin{tabular}{l|l|l|rrrrrrrrrrrrr|l|l|l}
$P$ & $e$ & Friendly & \multicolumn{13}{c|}{Counts $|\cT_n(P, e)|$ for $n=1,\ldots,12$} &  OEIS & Wilf-equivalence & References \\\hline
21534 & \hyph0\hyph\hyph & \hyph0\hyph\hyph & 1&2&5&14&41&122&365&1094&3281&9842&29525&88574&$\dots$ & \textcolor{red}{\href{https://oeis.org/A007051}{A007051}} & & \cite[Thm.~1]{MR2967227}; Tab.~\ref{tab:staggered} \\
 & \hyph10\hyph & \hyph10\hyph & 1&2&5&14&41&123&375&1158&3615&11393&36209&115940&$\dots$ & \textcolor{green}{\href{https://oeis.org/A176677}{A176677}} & $(31254,0\hyph1\hyph)$ Lem.~\ref{lem:Lmove1} & \\
 & \hyph11\hyph & & 1&2&5&14&41&124&384&1212&3885&12614&41400&137132&$\dots$ & \textcolor{violet}{\href{https://oeis.org/A159769}{A159769}} & & \cite[Class~6.3]{MR2645188} \\
\hline
21543 & \hyph00\hyph & \hyph00\hyph & 1&2&5&14&41&122&365&1094&3281&9842&29525&88574&$\dots$ & \textcolor{red}{\href{https://oeis.org/A007051}{A007051}} & & \cite[Thm.~1]{MR2967227} \\
 & \hyph01\hyph & \hyph01\hyph & 1&2&5&14&41&123&374&1147&3538&10958&34042&105997&$\dots$ & \textcolor{blue}{\href{https://oeis.org/A054391}{A054391}} & & Sec.~\ref{sec:21543-motzkin}; Tab.~\ref{tab:staggered} \\
 & \hyph10\hyph & \hyph10\hyph & 1&2&5&14&41&123&375&1158&3615&11393&36209&115940&$\dots$ & \textcolor{green}{\href{https://oeis.org/A176677}{A176677}} & $(21534,\hyph10\hyph)$ Lem.~\ref{lem:subtree} & \\
 & \hyph11\hyph & & 1&2&5&14&41&124&384&1212&3885&12614&41400&137132&$\dots$ & \textcolor{violet}{\href{https://oeis.org/A159769}{A159769}} & & \cite[Class~6.3]{MR2645188} \\
\hline
31245 & 0\hyph0\hyph & & 1&2&5&14&41&122&365&1094&3281&9842&29525&88574&$\dots$ & \textcolor{red}{\href{https://oeis.org/A007051}{A007051}} & & \cite[Thm.~1]{MR2967227} \\
 & 0\hyph1\hyph & & 1&2&5&14&41&123&375&1158&3615&11393&36209&115940&$\dots$ & \textcolor{green}{\href{https://oeis.org/A176677}{A176677}} & $(12534,1\hyph0\hyph)$ Lem.~\ref{lem:pathmove} &  \\
 & 1\hyph0\hyph & & 1&2&5&14&41&123&375&1158&3615&11393&36209&115940&$\dots$ & \textcolor{green}{\href{https://oeis.org/A176677}{A176677}} & & ? \\
 & 1\hyph1\hyph & & 1&2&5&14&41&124&384&1212&3885&12614&41400&137132&$\dots$ & \textcolor{violet}{\href{https://oeis.org/A159769}{A159769}} & & \cite[Class~6.3]{MR2645188} \\
\hline
31254 & 0\hyph0\hyph & 0\hyph0\hyph & 1&2&5&14&41&122&365&1094&3281&9842&29525&88574&$\dots$ & \textcolor{red}{\href{https://oeis.org/A007051}{A007051}} & & \cite[Thm.~1]{MR2967227} \\
 & 0\hyph1\hyph, 1\hyph0\hyph & 0\hyph10, 1\hyph0\hyph & 1&2&5&14&41&123&375&1158&3615&11393&36209&115940&$\dots$ & \textcolor{green}{\href{https://oeis.org/A176677}{A176677}} & $(31245,1\hyph0\hyph)$ Lem.~\ref{lem:subtree} & \\
 & 1\hyph1\hyph & 1\hyph10 & 1&2&5&14&41&124&384&1212&3885&12614&41400&137132&$\dots$ &\textcolor{violet}{\href{https://oeis.org/A159769}{A159769}} & & \cite[Class~6.3]{MR2645188} \\
 \hline
32145 & 0\hyph0\hyph & & 1&2&5&14&41&122&365&1094&3281&9842&29525&88574&$\dots$ & \textcolor{red}{\href{https://oeis.org/A007051}{A007051}} & & \cite[Thm.~1]{MR2967227} \\
 & 0\hyph1\hyph, 1\hyph0\hyph & & 1&2&5&14&41&123&375&1158&3615&11393&36209&115940&$\dots$ & \textcolor{green}{\href{https://oeis.org/A176677}{A176677}} & $(31245,0\hyph1\hyph)$ Lem.~\ref{lem:subtree} & \\
 & 1\hyph1\hyph & & 1&2&5&14&41&124&384&1212&3885&12614&41400&137132&$\dots$ & \textcolor{violet}{\href{https://oeis.org/A159769}{A159769}} & & \cite[Class~6.3]{MR2645188} \\
\end{tabular}
}
\end{table}

\section{Bijections with other combinatorial objects}
\label{sec:bijection}

In this section we establish bijections between pattern-avoiding binary trees and other combinatorial objects, specifically binary strings, pattern-avoiding Motzkin paths, and pattern-avoiding set partitions.

\subsection{Binary trees and bitstrings}

\subsubsection{Bijection between $\cT_n(132,\hyph\hyph)$ and bitstrings~$\{0,1\}^{n-1}$}
\label{sec:132-bin}

This bijection is illustrated in Figure~\ref{fig:bij-bin}~(a).
Consider a tree~$T \in \cT_n(P,e)$ where $(P,e):=(132, \hyph\hyph)$.
We define $b:=\beta_L(r(T))$ and~$\ell_i:=c_L^{i-1}(r(T))$ for~$i=1,\ldots,b$, i.e., we consider the left branch~$(\ell_1,\ldots,\ell_b)$ starting from the root of~$T$.
Due to the forbidden tree pattern~$(P,e)$, the tree~$T$ has exactly one right branch with~$\beta_R(\ell_i)$ many vertices starting at~$\ell_i$, for all~$i=1,\ldots,b$.
We map~$T$ to a bitstring of length~$n-1$ by concatenating sequences of~1s and~0s alternatingly, of lengths $\beta_R(\ell_1)-1,\beta_R(\ell_2),\beta_R(\ell_3),\ldots,\beta_R(\ell_b)$.
This is clearly a bijection between~$\cT_n(P,e)$ and~$\{0,1\}^{n-1}$.

\subsubsection{Bijection between $\cT_n(123, 0\hyph)$ and bitstrings~$\{0,1\}^{n-1}$}
\label{sec:123-bin}

This bijection is illustrated in Figure~\ref{fig:bij-bin}~(b).
Consider a tree~$T \in \cT_n(P,e)$ where $(P,e):=(123, 0 \hyph)$.
We define $b:=\beta_L(r(T))$ and~$\ell_i:=c_L^{i-1}(r(T))$ for~$i=1,\ldots,b$, i.e., we consider the left branch~$(\ell_1,\ldots,\ell_b)$ starting from the root of~$T$.
Due to the forbidden tree pattern~$(P,e)$, the right subtree of~$\ell_i$ is an all left-branch with~$\beta_L(c_R(\ell_i))$ many vertices, for all~$i=1,\ldots,b$.
We map~$T$ to a bitstring of length~$n-1$ by concatenating sequences of~1s and~0s alternatingly, of lengths~$\beta_L(c_R(\ell_1)),\beta_L(c_R(\ell_2))+1,\beta_L(c_R(\ell_3))+1,\ldots,\beta_L(c_R(\ell_b))+1$.
This is clearly a bijection between~$\cT_n(P,e)$ and~$\{0,1\}^{n-1}$.

\subsubsection{Bijection between $\cT_n(213, \hyph\hyph)$ and bitstrings~$\{0,1\}^{n-1}$}
\label{sec:213-bin}

This bijection is illustrated in Figure~\ref{fig:bij-bin}~(c).
Consider a tree~$T\in \cT_n(P,e)$ where $(P,e):=(213,\hyph\hyph)$.
Due to the forbidden tree pattern~$(P,e)$, no vertex of~$T$ has two children, i.e., $T$ is a path.
We map~$T$ to a bitstring of length~$n-1$ by going down the path starting at the root and recording a 1-bit for every edge going to the right, and a 0-bit for every edge going to the left.
This is clearly a bijection between~$\cT_n(P,e)$ and~$\{0,1\}^{n-1}$.

\begin{figure}
\includegraphics{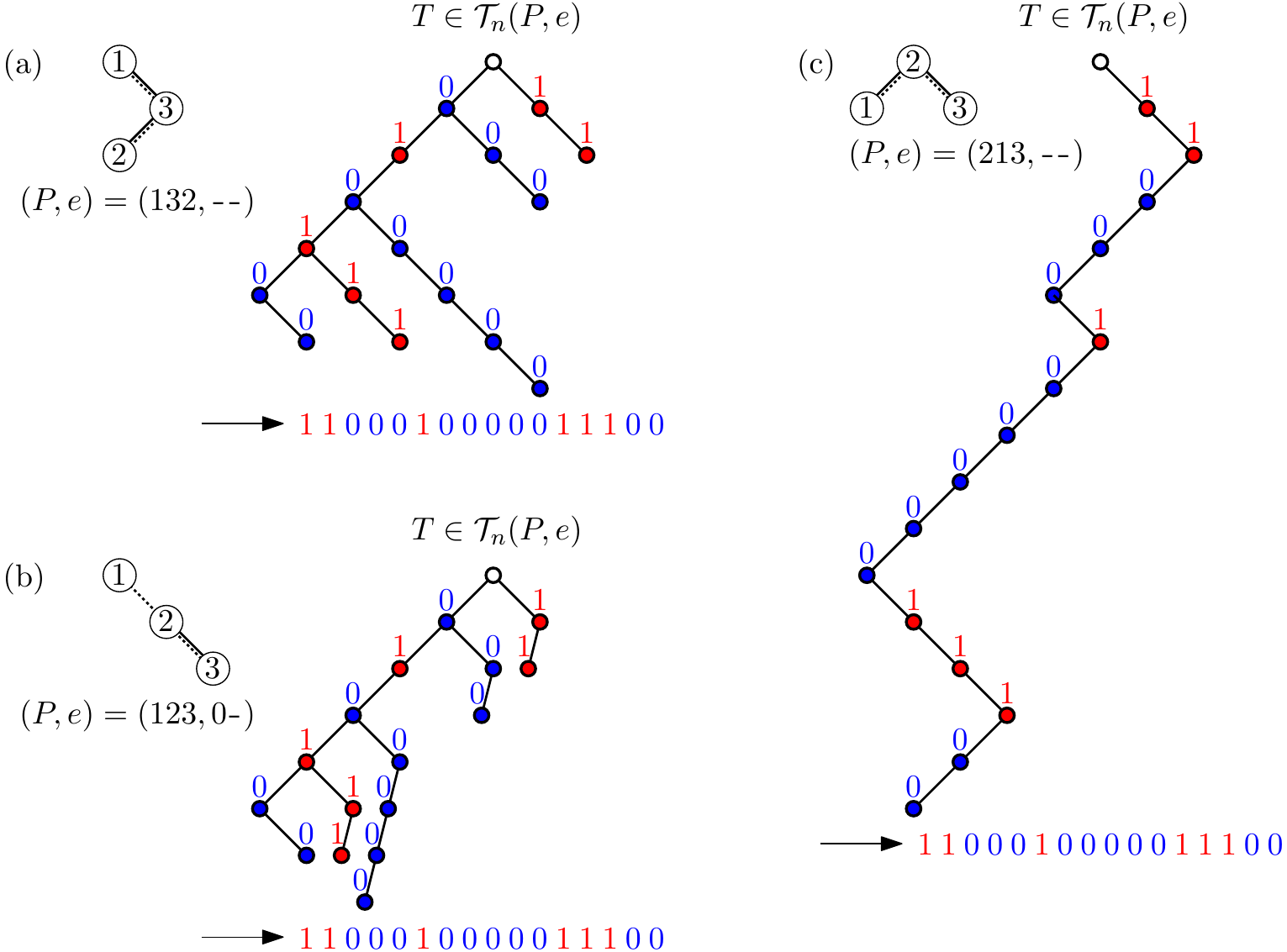}
\caption{Bijections between binary trees and bitstrings.}
\label{fig:bij-bin}
\end{figure}

\subsection{Binary trees and Motzkin paths}
\label{sec:motzkin}

\begin{figure}
\includegraphics[page=1]{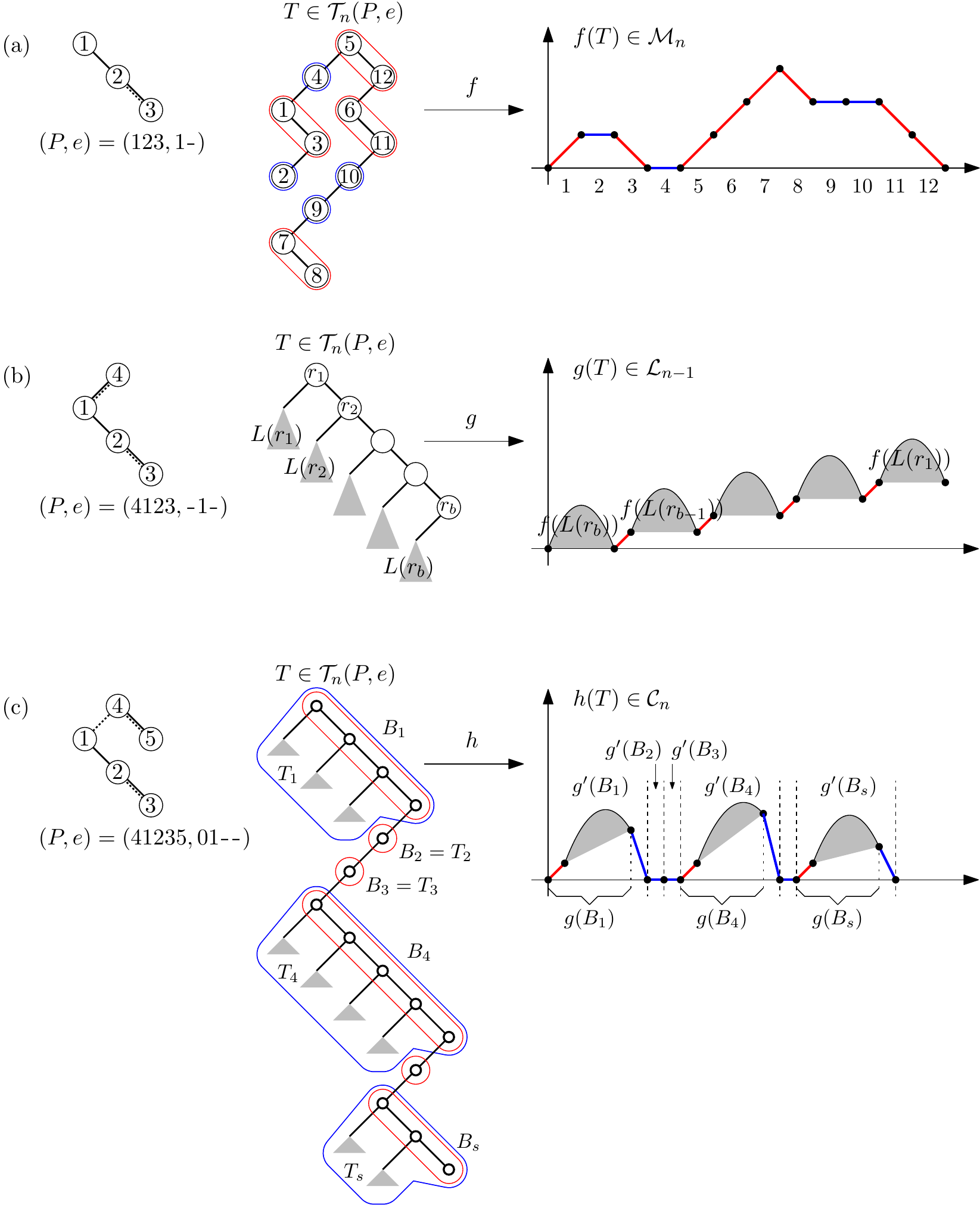}
\caption{Bijections between pattern-avoiding binary trees and different types of Motzkin paths.}
\label{fig:bij1}
\end{figure}

In this section, we present bijections between pattern-avoiding binary trees and different types of Motzkin paths.

Specifically, we consider lattice paths with steps~$\tU:=(1,1)$, $\tD:=(1,-1)$, $\tF:=(1,0)$, and~$\tD_h:=(1,-h)$ for~$h\geq 2$.
An \defi{$n$-step Motzkin path} starts at~$(0,0)$, ends at~$(n,0)$, uses only steps~$\tU$, $\tD$ or~$\tF$, and it never goes below the $x$-axis.
We write~$\cM_n$ for the set of all $n$-step Motzkin paths (OEIS~A001006).
An \defi{$n$-step Motzkin left factor} starts at~$(0,0)$, uses $n$ many steps~$\tU$, $\tD$ or~$\tF$, and it never goes below the $x$-axis.
We write~$\cL_n$ for the set of all $n$-step Motzkin left factors (OEIS~A005773).
An \defi{$n$-step Motzkin path with catastrophes}~\cite{MR4256410} starts at~$(0,0)$, ends at~$(n,0)$, uses only steps~$\tU$, $\tD$, $\tF$, or $\tD_h$ for $h\geq 2$, such that all $\tD_h$-steps end on the $x$-axis, and it never goes below the $x$-axis (OEIS~A054391).
We write~$\cC_n$ for the set of all $n$-step Motzkin paths with catastrophes.

\subsubsection{Bijection between~$\cT_n(123,1\hyph)$ and Motzkin paths~$\cM_n$}
\label{sec:123-motzkin}

This bijection is illustrated in Figure~\ref{fig:bij1}~(a).
Consider a tree~$T\in \cT_n(P,e)$ where $(P,e):=(123,1\hyph)$.
Due to the forbidden pattern~$(P,e)$, every maximal right branch in~$T$ consists of one or two vertices, but not more.
We map~$T$ to an $n$-step Motzkin path~$f(T)$ as follows.
Every maximal right branch in~$T$ consisting of one vertex~$i$ creates an $\tF$-step at position~$i$ in~$f(T)$.
Every maximal right branch in~$T$ consisting of two vertices~$i$ and~$j$, where $j=c_R(i)$, creates a pair of $\tU$-step and $\tD$-step at the same height at positions~$i$ and~$j$ in~$f(T)$, respectively.
It is easy to verify that~$f$ is indeed a bijection between~$\cT_n(P,e)$ and~$\cM_n$.

We remark that Rowland~\cite{MR2645188} described a bijection between~$\cT_n(123,1\hyph)$ and~$\cM_n$ that is different from~$f$.

\subsubsection{Bijection between~$\cT_n(1432,\hyph 1\hyph)$ and Motzkin left factors~$\cL_{n-1}$}
\label{sec:1432-motzkin}

This bijection is illustrated in Figure~\ref{fig:bij1}~(b), and it uses as a building block the bijection~$f$ defined in the previous section.
Instead of~$(1432,\hyph 1\hyph)$, we consider the mirrored tree pattern~$(P,e):=\mu(1432,\hyph 1\hyph)=(4123,\hyph 1\hyph)$ for convenience.
Consider a tree~$T \in \cT_n(P,e)$.
We define $b:=\beta_R(r(T))$ and~$r_i:=c_R^{i-1}(r(T))$ for~$i=1,\ldots,b$, i.e., we consider the right branch~$(r_1,\ldots,r_b)$ starting from the root of~$T$.
Due to the forbidden tree pattern~$(P,e)$, each subtree~$L(r_i)$ for $i=1,\ldots,b$ is $(123,1\hyph)$-avoiding.
Using the bijection~$f$ described in the previous section, we can thus map each subtree~$L(r_i)$ to a Motzkin path~$f(L(r_i))$.
Therefore, we map~$T$ to an $(n-1)$-step Motzkin left factor~$g(T)$ by combining the subpaths~$f(L(r_i))$, separating them by in total $b-1$ many $\tU$-steps, one between every two consecutive subpaths~$f(L(r_i))$ and~$f(L(r_{i+1}))$.
To make the proof work, the subpaths~$f(L(r_i))$ can be combined in increasing order from left to right on~$g(T)$, i.e., for $i=1,\ldots,b$, or in decreasing order, i.e., for $i=b,b-1,\ldots,1$, and for reasons that will become clear in the next section we combine them in decreasing order, i.e.,
\begin{equation}
\label{eq:gT}
g(T):=f(L(r_b)),\tU,f(L(r_{b-1})),\ldots,\tU,f(L(r_1)).
\end{equation}
The mapping~$g$ is clearly a bijection between~$\cT_n(P,e)$ and~$\cL_{n-1}$.

\subsubsection{Bijection between~$\cT_n(21543,\hyph 01\hyph)$ and Motzkin paths with catastrophes~$\cC_n$}
\label{sec:21543-motzkin}

This bijection is illustrated in Figure~\ref{fig:bij1}~(c), and it uses as a building block the bijection~$g$ defined in the previous section.
Instead of~$(21543,\hyph 01\hyph)$, we consider the mirrored tree pattern~$(P,e):=\mu(21543,\hyph 01\hyph)=(41235,01\hyph\hyph)$ for convenience.
Consider a tree~$T \in \cT_n(P,e)$ and the rightmost leaf in~$T$, and partition the path from the root of~$T$ to that leaf into a sequence of maximal right branches~$B_1,\ldots,B_\ell$.
For $i=1,\ldots,\ell$, we let~$T_i$ be the subtree of~$T$ that consists of~$B_i$ plus the left subtrees of all vertices on~$B_i$ except the last one.
Note that~$T_1,\ldots,T_\ell$ form a partition of~$T$.
Furthermore, $T$ avoiding~$(P,e)$ is equivalent to each of the~$T_i$, $i=1,\ldots,\ell$, avoiding~$(4123,01\hyph)$.
Using the bijection~$g$ described in the previous section, we can thus map each subtree~$T_i$ to a Motzkin left factor~$g(T_i)$, and by appending one additional appropriate step~$\tF$, $\tD$ or~$\tD_h$ for~$h\geq 2$ we obtain a Motzkin path~$g'(T_i)$.
Note that the rightmost leaf of~$T_i$ has no left child, and thus the definition~\eqref{eq:gT} yields that~$g'(T_i)$ touches the $x$-axis only at the first point and last point, but at no intermediate (integer) points.
Therefore, we map~$T$ to an $n$-step Motzkin path with catastrophes~$h(T)$ by concatenating the Motzkin subpaths~$g'(T_i)$ for $i=1,\ldots,\ell$, i.e., $h(T):=g'(T_1),g'(T_2),\ldots,g'(T_\ell)$.
It can be readily checked that $h$ is a bijection between~$\cT_n(P,e)$ and~$\cC_n$.

\subsubsection{Binary trees and Motzkin paths with 2-colored $\tF$-steps}

We now consider Motzkin paths whose $\tF$-steps come in two possible colors, which we denote by $\tFleft$ and~$\tFright$, respectively.
We write $\cM_n'$ for the set of $n$-step Motzkin paths with 2-colored $\tF$-steps.
For strings~$\tau_1,\ldots,\tau_\ell$, each using symbols from~$\{\tU,\tD,\tFleft,\tFright\}$, we write~$\cM_n'(\tau_1,\ldots,\tau_\ell)$ for the Motzkin paths from~$\cM_n'$ that avoid each of~$\tau_1,\ldots,\tau_\ell$ as (consecutive) substrings.

\begin{figure}[b!]
\includegraphics[page=3]{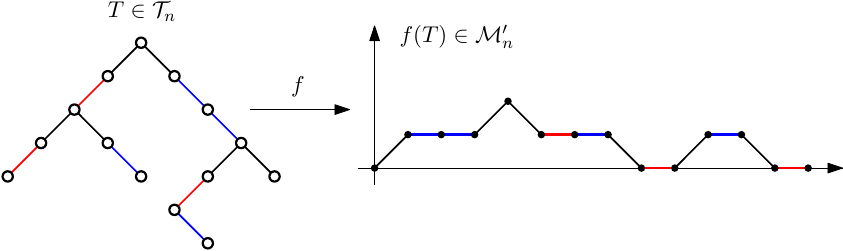}
\caption{Bijection between binary trees and Motzkin paths with 2-colored $\tF$-steps.}
\label{fig:motz2col}
\end{figure}

There is a natural bijection $f\colon \cT_n\rightarrow \cM_n'$, illustrated in Figure~\ref{fig:motz2col}.
In particular, Motzkin paths with 2-colored $\tF$-steps are a Catalan family.
Given a tree~$T\in\cT_n$, we define $f(T)$ by considering four cases:
If $r(T)$ has no children, then $f(T):=\varepsilon$.
If $r(T)$ has only a left child, then $f(T):=\tFleft,f(L(T))$.
If $r(T)$ has only a right child, then $f(T):=\tFright,f(R(T))$.
If $r(T)$ has two children, then $f(T):=\tU,f(R(T)),\tD,f(L(T))$.

In the following, we consider the restriction of~$f$ to various set of binary trees given by pattern avoidance.

\begin{theorem}
\label{thm:motz24}
The following four mappings $f\colon \cT_n(P,e)\rightarrow \cM_n'(X)$ are bijections:
\begin{enumerate}[label=(\roman*),leftmargin=8mm, noitemsep, topsep=1pt plus 1pt]
\item $(P,e)=(2134,111)=(2134,\hyph1\hyph)$ with $X=\{\tU\}\times\{\tU,\tFright\}=\{\tU\tU,\tU\tFright\}$;
\item $(P,e)=(3214,111)=(3214,1\hyph\hyph)$ with $X=\{\tD\}\times\{\tU,\tFleft\}=\{\tD\tU,\tD\tFleft\}$;
\item $(P,e)=(1324,111)=(1324,1\hyph\hyph)$ with $X=\{\tU,\tFright\}\times\{\tU\}=\{\tU\tU,\tFright\tU\}$;
\item $(P,e)=(1243,111)=(1243,1\hyph\hyph)$ with $X=\{\tU,\tFright\}\times\{\tU,\tFright\}\times\{\tU,\tFleft\}=\{\tU\tU\tU,\tU\tU\tFleft,\tU\tFright\tU,\tU\tFright\tFleft,\tFright\tU\tU,\allowbreak\tFright\tU\tFleft,\tFright\tFright\tU,\tFright\tFright\tFleft\}$.
\end{enumerate}
Furthermore, all four sets of Motzkin paths are Wilf-equivalent and counted by OEIS~A025242.
\end{theorem}

\begin{proof}
The first part of the lemma follows directly from the definition of~$f$.
Furthermore, the tree patterns in~(i) and~(ii) are Wilf-equivalent as they are mirror images of each other.
Lastly, the tree patterns in~(i), (iii) and (iv) are Wilf-equivalent by Lemma~\ref{lem:pathmove}.
\end{proof}

\begin{theorem}
\label{thm:motz25}
The following three mappings $f\colon \cT_n(P,e)\rightarrow \cM_n'(X)$ are bijections:
\begin{enumerate}[label=(\roman*),leftmargin=8mm, noitemsep, topsep=1pt plus 1pt]
\item $(P,e)=(21435,1111)=(21435,\hyph1\hyph\hyph)$ with $X=\{\tU\tU\}$;
\item $(P,e)=(42135,1111)=(42135,1\hyph\hyph\hyph)$ with $X=\{\tD\tU\}$;
\item $(P,e)=(13254,1111)=(13254,1\hyph1\hyph)$ with $X=\{\tU,\tFright\}\times\{\tU\}\times \{\tU,\tFleft\}=\{\tU\tU\tU,\tU\tU\tFleft,\tFright\tU\tU,\tFright\tU\tFleft\}$.
\end{enumerate}
Furthermore, all three sets of Motzkin paths are Wilf-equivalent and counted by OEIS~A159771.
\end{theorem}

\begin{proof}
The first part of the lemma follows directly from the definition of~$f$.
Furthermore, the tree patterns in~(i) and~(ii) are Wilf-equivalent as they are mirror images of each other.
Lastly, the tree patterns in~(i) and~(iii) are Wilf-equivalent by Lemma~\ref{lem:21435-13254}.
\end{proof}

\subsection{Binary trees and set partitions}

In this section, we present bijections between pattern-avoiding binary trees and different pattern-avoiding set partitions.

A \defi{set partition of~$[n]$} is a collection of non-empty disjoint subsets~$B_1,\ldots,B_\ell$, called \defi{blocks}, whose union is~$[n]$.
A \defi{crossing} in a set partition is a quadruple of elements~$a<b<c<d$ such that $a,c\in B_i$ and~$b,d\in B_j$ with $i\neq j$.
A partition is called \defi{non-crossing} if it has no crossings.
Set partitions are counted by the Bell numbers, and non-crossing set partitions are a well-known Catalan family.

A set partition can be identified uniquely by its \defi{restricted growth string (RGS)}, which is the string~$x_1\cdots x_n$ given by sorting the blocks by their smallest element, and such that $x_i=j$ if the element~$i$ is contained in the $j$th block in this ordering.
For example, the RGS for the partition~$\{\{1,3\}, \{2,4,7\}, \{5\}, \{6,8\}\}$ of~$[9]$ is~$12123424$.
Restricted growth strings are characterized by the conditions $x_1=1$, and $x_{i+1}\leq \max\{x_1,x_2,\ldots,x_i\}+1$ for all $i=1,\ldots,n-1$.
We write~$\cP_n$ for the set of all restricted growth strings of set partitions of~$[n]$.
The notion of pattern containment in permutations extends straightforwardly to pattern containment in strings, in particular in restricted growth strings.
Such a pattern string~$\tau$ may contain repeated entries, which means that the corresponding entries in the string in an occurrence of the pattern have the same value.
We write~$\cP_n(\tau_1,\ldots,\tau_\ell)$ for restricted growth strings from~$\cP_n$ that avoid each the patterns~$\tau_1,\ldots,\tau_\ell$.
Observe that~$\cP_n(1212)$ are precisely non-crossing set partitions.
The study of pattern avoidance in set partitions was initiated by Kreweras~\cite{MR309747} and Klazar~\cite{MR1370819,MR1750890,MR1769065}; see also~\cite{MR2419765,MR2398831,MR2599721,MR2824451,MR2863229,MR2924748,MR3090560,MR2954657,MR3003349,MR3245891,MR3548808}.

\subsubsection{Bijection between~$\cT_n$ and non-crossing set partitions~$\cP(1212)$}

We define two bijections~$\varphi_L,\varphi_R\colon \cT_n\rightarrow \cP_n(1212)$ that will be used in the remainder of this section; see Figure~\ref{fig:bij-TnPn}.
For a given tree~$T\in\cT_n$, the blocks of the set partition~$\varphi_L(T)$ are defined by the sets of vertices in the maximal left branches of~$T$.
Formally, we write $i\sim_L j$ if $i$ and~$j$ are in the same left branch of~$T$, which is an equivalence relation.
Then the set partition~$\varphi_L(T)$ is given by the equivalence classes of~$\sim_L$, i.e.,
\begin{subequations}
\label{eq:phiLR}
\begin{equation}
\label{eq:phiL}
\varphi_L(T):=[n]/{\sim_L}.
\end{equation}
We also define
\begin{equation}
\label{eq:phiR}
\varphi_R(T):=[n]/{\sim_R},
\end{equation}
\end{subequations}
where $i\sim_R j$ if $i$ and~$j$ are in the same right branch of~$T$.

\begin{figure}
\includegraphics[page=1]{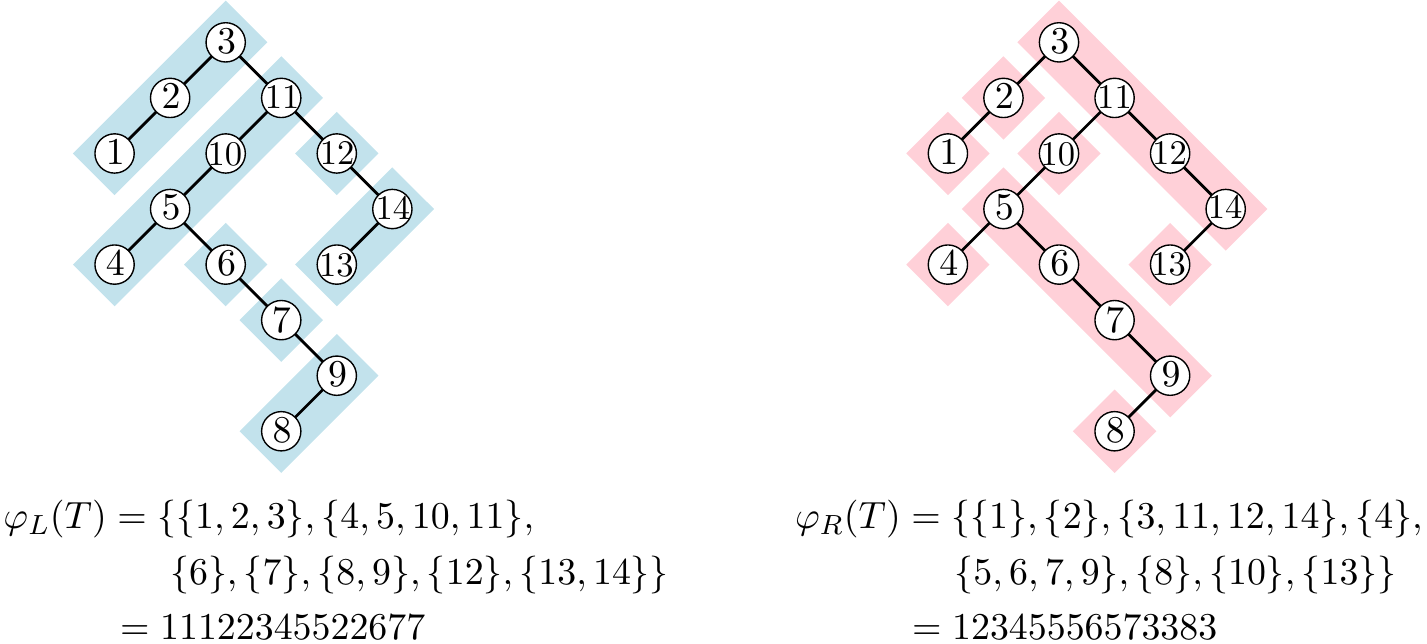}
\caption{Two bijections between binary trees and non-crossing set partitions.}
\label{fig:bij-TnPn}
\end{figure}

\begin{lemma}
\label{lem:1212-tree}
The mappings~$\varphi_L,\varphi_R\colon \cT_n \rightarrow \cP_n(1212)$ defined in~\eqref{eq:phiLR} are bijections.
\end{lemma}

\begin{proof}
It suffices to prove the statement for~$\varphi_L$, as $\varphi_R$ is defined symmetrically.

Given $T\in\cT_n$, we first show that the RGS~$\varphi_L(T)$ avoids~$1212$.
Suppose for the sake of contradiction that~$\varphi_L(T)$ contains the pattern~$1212$, and let $a<b<c<d$ be the positions of the occurrence of this pattern.
This means that in~$T$ the vertices~$a$ and~$c$ are in the same left branch, and the vertices~$b$ and~$d$ are in the same left branch, different from the first one.
As $a$ and~$c$ are in the same branch and $a<c$ we conclude that $a\in L(c)$.
Furthermore, as $a<b<c$ we have~$b\in L(c)$.
As $b$ and~$d$ are in the same branch, it follows that $d\in L(c)$.
However, this is a contradiction to $c<d$.

It is easy to see that the mapping~$\varphi_L$ is injective.
To see that it is surjective, consider an RGS~$x\in\cP_n(1212)$.
The first entry of~$x$ is~1.
Let $j$ be the position of the last~1 in~$x$, and let $x_L$ and~$x_R$ be the substrings of~$x$ strictly to the left and right of position~$j$, respectively.
As $x$ avoids~$1212$, no symbol in~$x$ appears both to the left and right of position~$j$.
Furthermore, the condition $x_{i+1}\leq \max\{x_1,x_2,\ldots,x_i\}+1$ for all $i=1,\ldots,n-1$ implies that all entries in~$x$ to the left of position~$j$ are strictly smaller than all entries to the right of position~$j$.
Consequently, the binary tree~$\varphi_L^{-1}(x)$ has the root~$j$ with the left subtree $\varphi^{-1}(x_L)$ and the right subtree $\varphi_L^{-1}(x_R)$.
\end{proof}

\subsubsection{Staggered tree patterns}

In the following, we consider the restriction of~$\varphi_L$ and~$\varphi_R$ to various sets of binary trees given by pattern avoidance, and we derive conditions to ensure that avoiding a tree pattern~$P$ corresponds to avoiding the RGS pattern~$\varphi_L(P)$ or~$\varphi_R(P)$, respectively (in addition to~$1212$).

A tree pattern~$(P,e)$ is called \defi{staggered}, if one of the following two recursive conditions is satisfied; see Figure~\ref{fig:staggered}:
\begin{itemize}[leftmargin=8mm, noitemsep, topsep=1pt plus 1pt]
\item $(P,e)$ is a contiguous left path, i.e., $e(i)=1$ for all edges $(i,p(i))$ on this path.
\item The root of~$(P,e)$ is contained in a contiguous left branch, exactly one vertex~$i$ on the branch has a right child~$j:=c_R(i)$ with $e(j)=0$, and $(P(j),e_{P(j)})$ is staggered.
\end{itemize}

\begin{figure}[h!]
\includegraphics[page=2]{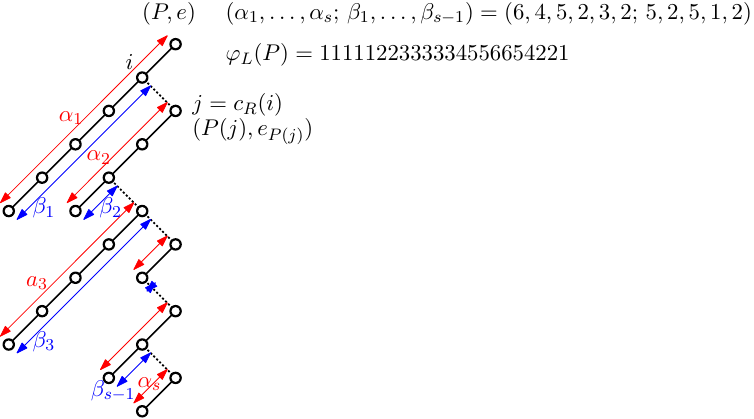}
\caption{Definition of staggered tree patterns.}
\label{fig:staggered}
\end{figure}

For the following discussion, for any string~$x$ and any integer~$k\geq 0$, we write $x^k$ for the concatenation of $k$ copies of~$x$.
A staggered tree pattern~$(P,e)$ is described uniquely by an integer sequence
\begin{equation*}
(\alpha_1,\ldots,\alpha_{s-1},\alpha_\ell;\, \beta_1,\ldots,\beta_{\ell-1}),
\end{equation*}
referred to as its \defi{signature}, where $\alpha_i$ is the number of vertices on the $i$th contiguous left branch, counted from the root, and the $\beta_i$th vertex counted from bottom to top on the $i$th contiguous branch is the unique vertex on that branch having a right child.
Clearly, we have $1\leq \beta_i\leq \alpha_i$, and furthermore
\begin{equation}
\label{eq:phiLP}
\varphi_L(P)=1^{\beta_1}2^{\beta_2}3^{\beta_3}\cdots (\ell-1)^{\beta_{\ell-1}} \ell^{\,\alpha_\ell} (\ell-1)^{\alpha_{\ell-1}-\beta_{\ell-1}}\cdots 3^{\alpha_3-\beta_3}2^{\alpha_2-\beta_2}1^{\alpha_1-\beta_1};
\end{equation}
see Figure~\ref{fig:staggered}.

\begin{theorem}
\label{thm:staggered}
Let $(P,e)$, $P\in\cT_k$, be a staggered tree pattern with signature~$(\alpha_1,\ldots,\alpha_\ell;\,\beta_1,\ldots,\beta_{\ell-1})$, satisfying the following two conditions: (i) if $\ell>1$ and $\beta_1=\alpha_1$, then $\alpha_1\in\{1,2\}$; (ii) $\beta_i<\alpha_i$ for all $i=2,\ldots,\ell-1$.
Then the mapping~$\varphi_L\colon \cT_n(P,e) \rightarrow \cP_n(1212, \varphi_L(P))$ is a bijection.
\end{theorem}

\begin{proof}
By Lemma~\ref{lem:1212-tree}, $\varphi_L\colon \cT_n \rightarrow \cP_n(1212)$ is a bijection.
Consequently, if suffices to show that $T\in\cT_n$ contains the tree pattern~$(P,e)$ if and only if $\varphi_L(T)\in\cP_n(1212)$ contains the (RGS) pattern~$\varphi_L(P)$.

Using that all left edges~$(i,p(i))$ of~$(P,e)$ are contiguous, i.e., they satisfy $e(i)=1$, and the definition of the mapping~$\varphi_L$, we see easily that if~$T\in\cT_n$ contains the tree pattern~$(P,e)$, then $\varphi_L(T)$ contains the pattern~$\varphi_L(P)$.

It remains to show that if $x:=\varphi_L(T)$ for $T\in\cT_n$ contains the pattern~$\varphi_L(P)$, then $T$ contains the tree pattern~$(P,e)$.
For this we argue by induction on $\ell$, i.e., on the number of contiguous left branches of the staggered pattern.
For this part of the argument, conditions~(i) and~(ii) in the theorem will become relevant.
To settle the induction basis, suppose that $\ell=1$, i.e., $(P,e)$ is a contiguous left path.
In this case we have $\varphi_L(P)=1^k=1\cdots 1$, i.e., $x$ contains $k$ occurrences of the same symbol.
The definition of~$x=\varphi_L(T)$ shows that consequently, $T$ contains a left branch on at least $k$ vertices, i.e., $T$ contains~$(P,e)$, as claimed.
For the induction step suppose that $\ell>1$, i.e., the staggered tree pattern~$(P,e)$ has at least two contiguous left branches.
Consider the contiguous left branch in~$P$ starting at the root, which has $\alpha_1$ vertices, and let $i$ be the $\beta_1$th vertex on this branch counted from bottom to top, which has a right child~$j:=c_R(i)$.
As $(P,e)$ satisfies conditions~(i) and~(ii), these conditions are also satisfied for the smaller staggered pattern~$(P(j),e_{P(j)})$ (in fact, condition~(i) is satisfied trivially).
By~\eqref{eq:phiLP}, the pattern~$\varphi_L(P)$ has $\alpha_1$ many 1s, split into two groups of size~$\beta_1$ and~$\alpha_1-\beta_1$ that surround all larger symbols.
Let $i_1<\cdots<i_{\alpha_1}$ be the positions of the symbols to which those 1s are matched in the occurrence of the pattern~$\varphi_L(P)$ in~$x$.
By the definition of~$\varphi_L(T)$, the tree~$T$ contains a left branch that includes the vertices~$i_1,\ldots,i_{\alpha_1}$ successively from bottom to top.
Let $x'$ be the substring of~$x$ strictly between positions~$i_{\beta_1}$ and~$i_{\beta_1+1}$ if $\beta_1<\alpha_1$ and strictly after position~$i_{\alpha_1}$ if $\beta_1=\alpha_1$.
Using that $x$ is non-crossing and $\beta_i<\alpha_i$ for $i=1,\ldots,\ell-1$, we may assume w.l.o.g.\ that $x'$ does not contain any occurrences of the symbol that is present at positions~$i_1,\ldots,i_{\alpha_1}$ in~$x$.
By induction, we know that $T':=\varphi_L^{-1}(x')$ contains the staggered tree pattern~$(P(j),e_{P(j)})$.

We distinguish two cases, namely $\beta_1<\alpha_1$ and $\beta_1=\alpha_1$.

\begin{table}
\caption{All tree patterns~$(P,e)$ on at most 5 vertices for which $(P,e)$ or $\mu(P,e)$ is staggered and the corresponding RGS patterns given by Theorems~\ref{thm:staggered} or~\ref{thm:staggered-mu}, respectively.}
\label{tab:staggered}
\begin{center}
\begin{tabular}{|c|c|c|c|}
\hline
Tree Patterns & RGS patterns & Bijection & OEIS \\ \hline\hline
$(123,11)=(123,1\hyph)$ & 1212, 111 & $\varphi_R$ & \textcolor{blue}{\href{https://oeis.org/A001006}{A001006}} \\
\hline
$(132,10)=(132,\hyph\hyph)$ & 1212, 121 & $\varphi_R$ & \textcolor{red}{\href{https://oeis.org/A000079}{A000079}} \\
\hline
$(213,10)=(213,\hyph\hyph)$ & 1212, 112 & $\varphi_L$ & \textcolor{red}{\href{https://oeis.org/A000079}{A000079}} \\
$(213,01)=(213,\hyph\hyph)$ & 1212, 122 & $\varphi_R$ & \textcolor{red}{\href{https://oeis.org/A000079}{A000079}} \\
\hline\hline
$(1234,111)=(1234,11\hyph)$ & 1212, 1111 & $\varphi_R$ & \href{https://oeis.org/A036765}{A036765} \\
\hline
$(1243,110)=(1243,1\hyph\hyph)$ & 1212, 1121 & $\varphi_R$ & \textcolor{blue}{\href{https://oeis.org/A025242}{A025242}} \\
$(4312,110)=(4312,1\hyph\hyph)$ & 1212, 1211 & $\varphi_L$ & \textcolor{blue}{\href{https://oeis.org/A025242}{A025242}} \\
\hline
$(1324,101)=(1324,1\hyph\hyph)$ & 1212, 1211 & $\varphi_R$ & \textcolor{blue}{\href{https://oeis.org/A025242}{A025242}} \\
$(4213,110)=(4213,1\hyph\hyph)$ & 1212, 1121 & $\varphi_L$ & \textcolor{blue}{\href{https://oeis.org/A025242}{A025242}} \\
\hline
$(1423,010)=(1423,0\hyph\hyph)$ & 1212, 1232 & $\varphi_L$ & \textcolor{red}{\href{https://oeis.org/A001519}{A001519}} \\
$(1423,101)=(1423,\hyph0\hyph)$ & 1212, 1221 & $\varphi_R$ & \textcolor{red}{\href{https://oeis.org/A001519}{A001519}} \\
$(4132,010)=(4132,0\hyph\hyph)$ & 1212, 1213 & $\varphi_R$ & \textcolor{red}{\href{https://oeis.org/A001519}{A001519}} \\
\hline
$(1432,011)=(1432,\hyph1\hyph)$ & 1212, 1222 & $\varphi_L$ & \textcolor{green}{\href{https://oeis.org/A005773}{A005773}} \\
$(4123,011)=(4123,\hyph1\hyph)$ & 1212, 1112 & $\varphi_R$ & \textcolor{green}{\href{https://oeis.org/A005773}{A005773}} \\
\hline
$(2143,101)=(2143,\hyph0\hyph)$ & 1212, 1122 & $\varphi_L$ & \textcolor{red}{\href{https://oeis.org/A001519}{A001519}} \\
\hline\hline
$(12345,1111)=(12345,111\hyph)$ & 1212, 11111 & $\varphi_R$ & \href{https://oeis.org/A036766}{A036766} \\
\hline
$(12354,1110)=(12354,11\hyph\hyph)$ & 1212, 11121 & $\varphi_R$ & \textcolor{teal}{\href{https://oeis.org/A159768}{A159768}} \\
$(54312,1110)=(54312,11\hyph\hyph)$ & 1212, 12111 & $\varphi_L$ & \textcolor{teal}{\href{https://oeis.org/A159768}{A159768}} \\
\hline
$(12435,1101)=(12435,11\hyph\hyph)$ & 1212, 11211 & $\varphi_R$ & \textcolor{teal}{\href{https://oeis.org/A159768}{A159768}} \\
\hline
$(12534,1101)=(12534,1\hyph0\hyph)$ & 1212, 11221 & $\varphi_R$ & \textcolor{green}{\href{https://oeis.org/A176677}{A176677}} \\
$(54132,1101)=(54132,1\hyph0\hyph)$ & 1212, 12211 & $\varphi_L$ & \textcolor{green}{\href{https://oeis.org/A176677}{A176677}} \\
\hline
$(13245,1011)=(13245,1\hyph1\hyph)$ & 1212, 12111 & $\varphi_R$ & \textcolor{teal}{\href{https://oeis.org/A159768}{A159768}} \\
$(53214,1110)=(53214,11\hyph\hyph)$ & 1212, 11121 & $\varphi_L$ & \textcolor{teal}{\href{https://oeis.org/A159768}{A159768}} \\
\hline
$(14235,1011)=(14235,10\hyph\hyph)$ & 1212, 12211 & $\varphi_R$ & \textcolor{green}{\href{https://oeis.org/A176677}{A176677}} \\
$(52143,1101)=(52143,1\hyph0\hyph)$ & 1212, 11221 & $\varphi_L$ & \textcolor{green}{\href{https://oeis.org/A176677}{A176677}} \\
\hline
$(15234,1011)=(15234,\hyph01\hyph)$ & 1212, 12221 & $\varphi_R$ & \textcolor{blue}{\href{https://oeis.org/A054391}{A054391}} \\
\hline
$(15243,0101)=(15243,0\hyph0\hyph)$ & 1212, 12332 & $\varphi_L$ & \textcolor{red}{\href{https://oeis.org/A007051}{A007051}} \\
$(51423,0101)=(15243,0\hyph0\hyph)$ & 1212, 12213 & $\varphi_R$ & \textcolor{red}{\href{https://oeis.org/A007051}{A007051}} \\
$(15243,1010)=(15243,\hyph0\hyph\hyph)$ & 1212, 12321 & $\varphi_R$ & \textcolor{red}{\href{https://oeis.org/A007051}{A007051}} \\
\hline
$(15324,0110)=(15324,01\hyph\hyph)$ & 1212, 12232 & $\varphi_L$ & \newa{}\textrightarrow\href{https://oeis.org/A365508}{A365508} \\
$(51324,0101)=(51324,01\hyph\hyph)$ & 1212, 12113 & $\varphi_R$ & \newa{}\textrightarrow\href{https://oeis.org/A365508}{A365508} \\
\hline
$(15423,0110)=(15423,01\hyph\hyph)$ & 1212, 12322 & $\varphi_L$ & \newa{}\textrightarrow\href{https://oeis.org/A365508}{A365508} \\
$(51243,0110)=(51243,01\hyph\hyph)$ & 1212, 11213 & $\varphi_R$ & \newa{}\textrightarrow\href{https://oeis.org/A365508}{A365508} \\
\hline
$(15432,0111)=(15432,\hyph11\hyph)$ & 1212, 12222 & $\varphi_L$ & \textcolor{olive}{\href{https://oeis.org/A159772}{A159772}} \\
$(51234,0111)=(15432,\hyph11\hyph)$ & 1212, 11112 & $\varphi_R$ & \textcolor{olive}{\href{https://oeis.org/A159772}{A159772}} \\
\hline
$(21534,1010)=(21534,\hyph0\hyph\hyph)$ & 1212, 11232 & $\varphi_L$ & \textcolor{red}{\href{https://oeis.org/A007051}{A007051}} \\
$(41325,0101)=(41325,0\hyph\hyph\hyph)$ & 1212, 12133 & $\varphi_R$ & \textcolor{red}{\href{https://oeis.org/A007051}{A007051}} \\
\hline
$(21543,1011)=(21543,\hyph01\hyph)$ & 1212, 11222 & $\varphi_L$ & \textcolor{blue}{\href{https://oeis.org/A054391}{A054391}} \\
$(41235,0111)=(41235,01\hyph\hyph)$ & 1212, 11122 & $\varphi_R$ & \textcolor{blue}{\href{https://oeis.org/A054391}{A054391}} \\
\hline
\end{tabular}
\end{center}
\end{table}

\textbf{Case~(a):} $\beta_1<\alpha_1$.
By the definition of~$x'$, all vertices in~$T'$ are sandwiched between $i_{\beta_1}$ and~$i_{\beta_1+1}$, implying that $T'=R(i_{\beta_1})$ in~$T$.
Using that $T'$ contains the tree pattern~$(P(j),e_{P(j)})$, and the property~$e(j)=0$ from the definition of staggered tree patterns, we conclude that $T$ contains the tree pattern~$(P,e)$.

\textbf{Case~(b):} $\beta_1=\alpha_1$.
By the definition of~$x'$, all vertices in~$T'$ are larger than~$i_{\alpha_1}$.
If $T'=R(i_{\alpha_1})$ in~$T$, then $T$ contains the tree pattern~$(P,e)$, as argued in the previous case.
Otherwise, consider the lowest common ancestor~$\ihat$ of $i_{\alpha_1}$ and~$T'$ in~$T$.
Specifically, we have $i_1,\ldots,i_{\alpha_1}\in L(\ihat)$ and $T'\seq R(\ihat)$ in~$T$.
Condition~(i) stated in the theorem asserts that $\alpha_1\in\{1,2\}$, and therefore $T$ contains the tree pattern~$(P,e)$.
Specifically, $\ihat$ and its left child in~$T$ make the occurrence of~$(P(j),e_{P(j)})$ to an occurrence of~$(P,e)$.

This completes the proof of the theorem.
\end{proof}

By applying the mirroring operation~$\mu$ to a tree pattern, we obtain the following immediate consequence of Theorem~\ref{thm:staggered}.

\begin{theorem}
\label{thm:staggered-mu}
Let $(P,e)$ be such that $\mu(P,e)$ is staggered and satisfies the conditions of Theorem~\ref{thm:staggered}.
Then the mapping~$\varphi_R\colon \cT_n(P,e) \rightarrow \cP_n(1212, \varphi_R(P))$ is a bijection.
\end{theorem}

Theorems~\ref{thm:staggered} and~\ref{thm:staggered-mu} are quite versatile.
Applying them to all tree patterns on at most 5 vertices, we obtain the correspondences with pattern-avoiding set partitions listed in Table~\ref{tab:staggered}.
This also establishes some interesting Wilf-equivalences between various pattern-avoiding non-crossing set partitions, for example between the three sets~$\cP_n(1212,1232)$, $\cP_n(1212,1221)$, and~$\cP_n(1212,1213)$, or between the three sets~$\cP_n(1212,12332)$, $\cP_n(1212,12213)$, and~$\cP_n(1212,12321)$ (cf.~\cite{MR2863229}).

\section{Wilf-equivalence of tree patterns}
\label{sec:wilf}

In this section we provide five general lemmas for establishing Wilf-equivalence of certain tree patterns that are obtained by replacing some subpattern~$(P,e)$ with a Wilf-equivalent subpattern~$(P',e')$, or by moving it to a different vertex in the surrounding tree pattern.
We also give results for two specific patterns on 5~vertices: a Wilf-equivalence that is not covered by the general lemmas and a counting argument based on a Catalan-like recurrence.
We apply these result to systematically study Wilf-equivalences between all tree patterns on at most 5 vertices; see Tables~\ref{tab:size-3}--\ref{tab:size-5}.

\subsection{Subpattern replacement and shifting lemmas}

The first lemma considers replacing a subpattern with a Wilf-equivalent subpattern attached by a non-contiguous edge to a contiguous tree; see Figure~\ref{fig:subtree}.

\begin{lemma}
\label{lem:subtree}
Let $(S,1\cdots1)$ be a contiguous tree pattern, and let $x$ be a vertex in $S$ that does not have a right child and for any edge~$(u,v)$ of~$S$ such that $v$ is the right child of~$u$, we have that $v$ is a leaf or $x$ is in the subtree~$S(v)$.
Let $Q$ and~$Q'$ denote the tree patterns obtained from $(S,1\cdots1)$ by attaching tree patterns~$(P,e)$ and~$(P',e')$ with $P,P'\in\cT_k$, respectively, with a non-contiguous edge to~$x$ as a right subtree.
If $(P,e)$ and $(P',e')$ are Wilf-equivalent, then $Q$ and $Q'$ are also Wilf-equivalent.
\end{lemma}

Note that while all edges of~$(S,1\cdots1)$ are contiguous, no assumption is made about the edges of~$P$ or~$P'$, i.e., the functions~$e$ and~$e'$ are arbitrary.
However, the edge from~$x$ to the root of~$(P,e)$ or~$(P',e')$ must be non-contiguous in both~$Q$ and~$Q'$.

\begin{figure}[b!]
\includegraphics[page=1]{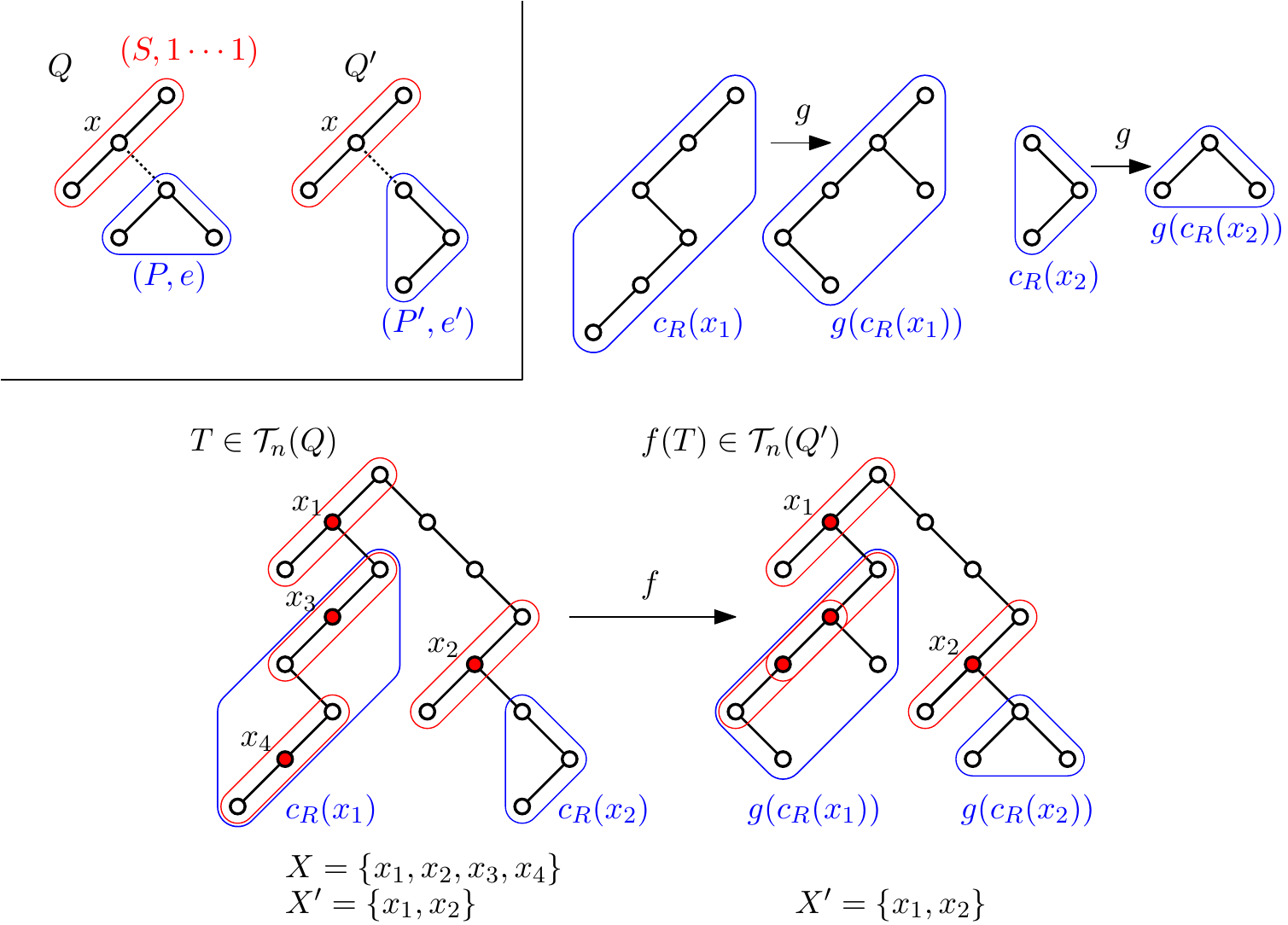}
\caption{Illustration of Lemma~\ref{lem:subtree}.
The bijection $g\colon \cT_n(213,\hyph\hyph)\rightarrow\cT_n(132,\hyph\hyph)$ is the composition of the bijections described in Sections~\ref{sec:213-bin} and \ref{sec:132-bin}.}
\label{fig:subtree}
\end{figure}

\begin{proof}
The proof is illustrated in Figure~\ref{fig:subtree}.
As $(P,e)$ and~$(P',e')$ are Wilf-equivalent, there is a bijection $g_n\colon \cT_n(P,e) \to \cT_n(P',e')$ for all~$n\geq 0$, and we let $g$ be the union of those functions over all~$n\geq 0$.
Let $T \in \cT_n(Q)$ and consider all occurrences of $(S,1\cdots1)$ in~$T$.
Let $X=\{x_1,\dots, x_\ell\}$ be the corresponding occurrences of the vertex~$x$ of~$Q$ in the host tree~$T$, i.e., $x_i\in[n]$ denotes the vertex to which $x$ is mapped in the $i$th occurrence of~$(S,1\cdots1)$, for all $i=1,\ldots,\ell$.
Furthermore, let $X'\seq X$ be minimal such that every $x_i\in X\setminus X'$ is in the right subtree of some vertex in~$X'$.
As $T$ avoids~$Q$, the subtree~$R(x_i)$ avoids~$(P,e)$ for every~$x_i\in X$.
We define $f(T)\in\cT_n$ as the tree obtained from~$T$ by replacing~$R(x_i)$ with~$g(R(x_i))$ for every~$x_i \in X'$.
Although this may introduce new occurrences of~$S$ in~$f(T)$, any new occurrence of~$S$ must include at least one vertex inside a right subtree~$g(R(x_i))$ of a vertex~$x_i \in X'$ other than the right child of~$x_i$ (since no other parts of~$T$ were modified by~$f$).
Thus, the new occurrence is either entirely confined to the right subtree~$g(R(x_i))$ of some~$x_i$, or it must contain the edge from~$x_i$ to its right child, in which case the assumption on~$S$ guarantees that the vertex~$x$ is mapped to the right subtree~$g(R(x_i))$ of~$x_i$ in this occurrence.
Thus, since $g(R(x_i))$ for each $x_i \in X'$ avoids~$(P',e')$, the tree~$f(T)$ avoids~$Q'$.
Furthermore, since the set $X'$ is invariant under~$f$, the mapping $f$ is reversible, so it is a bijection.
\end{proof}

Our second lemma considers moving a subpattern attached by a non-contiguous edge along a left branch; see Figure~\ref{fig:pathmove}.

\begin{lemma}
\label{lem:pathmove}
Let $P_d$ denote the left path on $d$ vertices and let $x,x'$ be two distinct vertices of~$P_d$.
Let $Q$ and~$Q'$ denote the tree patterns obtained from the contiguous pattern~$(P_d,1\cdots1)$ by attaching a tree pattern~$(P,e)$ with a non-contiguous edge to~$x$ and~$x'$ as a right subtree, respectively.
Then $Q$ and~$Q'$ are Wilf-equivalent.
\end{lemma}

\begin{figure}[h!]
\centerline{
\includegraphics[page=2]{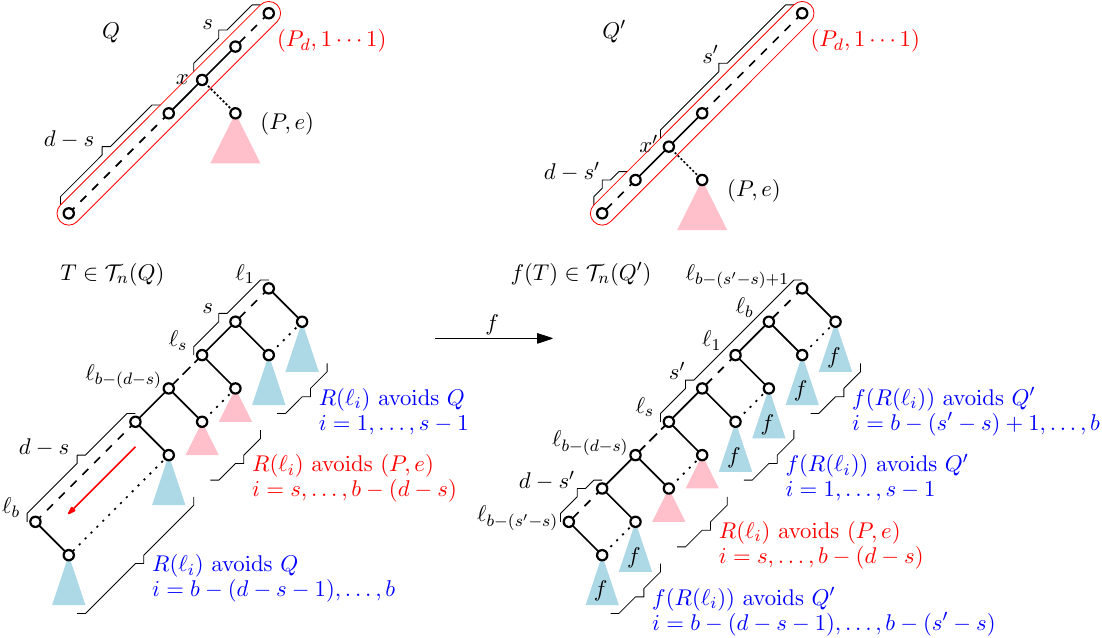}
}
\caption{Illustration of Lemma~\ref{lem:pathmove}.}
\label{fig:pathmove}
\end{figure}

\begin{proof}
The proof is illustrated in Figure~\ref{fig:pathmove}.
Let $s$ and $s'$ be the number of vertices above and including~$x$ or~$x'$ in~$Q$ and~$Q'$, respectively.
We assume w.l.o.g.\ that $s<s'$.
We inductively describe a bijection $f \colon \cT_n(Q) \to \cT_n(Q')$.
Consider a tree~$T \in \cT_n(Q)$.
If $n$ is strictly smaller than the number of vertices of~$Q$, then we define $f(T):=T$, i.e., $f$ is defined to be the identity mapping.
Otherwise we define $b:=\beta_L(r(T))$ and~$\ell_i:=c_L^{i-1}(r(T))$ for~$i=1,\ldots,b$, i.e., we consider the left branch~$(\ell_1,\ldots,\ell_b)$ starting at the root of~$T$.
If $b<d$ then $f(T)$ is obtained by recursively applying~$f$ to each subtree~$R(\ell_i)$ for $i=1,\dots,b$.
It remains to consider the case that~$b\ge d$.
Since $T$ avoids~$Q$, each subtree $R(\ell_i)$ for $i=s,\ldots,b-(d-s)$ avoids~$(P,e)$ and each subtree~$R(\ell_i)$ for $i=1,\dots,s-1$ and $i=b-(d-s-1),\dots,b$ avoids~$Q$.
The tree~$f(T)$ is obtained from~$T$ by cyclically down-shifting the left branch~$(\ell_1,\ldots,\ell_b)$ and its right subtrees by $s'-s$ positions.
Thus the vertex~$\ell_s$ becomes the $s'$th vertex from the root in~$f(T)$.
Furthermore, to the subtrees~$R(\ell_i)$ for $i=1,\dots,s-1$ and $i=b-(d-s-1),\ldots,b$ (avoiding $Q$) we recursively apply the mapping~$f$.
By construction, the tree~$f(T)$ avoids the tree pattern~$Q'$.
Furthermore, the mapping $f:\cT_n(Q) \to \cT_n(Q')$ is reversible, so it is a bijection.
\end{proof}

The next four lemmas consider different ways of moving a subpattern attached by a non-contiguous edge to a contiguous L-shaped path of length~2; see Figures~\ref{fig:Lmove1}--\ref{fig:Lmove3}.

\begin{figure}[h!]
\includegraphics[page=3]{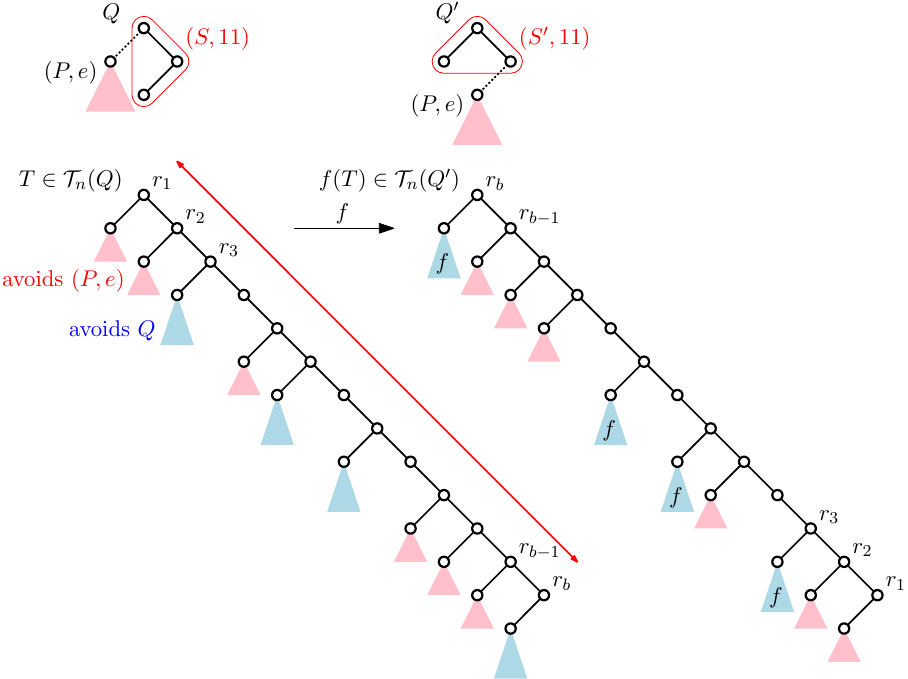}
\caption{Illustration of Lemma~\ref{lem:Lmove1}.}
\label{fig:Lmove1}
\end{figure}

\begin{lemma}
\label{lem:Lmove1}
Let $S$ and $S'$ be the paths of length~2 with $\tau(S)=132$ and $\tau(S')=213$, respectively.
Let $Q$ and~$Q'$ denote the tree patterns obtained from the contiguous patterns~$(S,11)$ and~$(S',11)$ by attaching a tree pattern~$(P,e)$ with a non-contiguous edge to the root of~$S$ as a left subtree and to the right leaf of~$S'$ as a left subtree, respectively.
Then $Q$ and~$Q'$ are Wilf-equivalent.
\end{lemma}

\begin{proof}
The proof is illustrated in Figure~\ref{fig:Lmove1}.
We inductively describe a bijection $f \colon \cT_n(Q) \to \cT_n(Q')$.
Consider a tree~$T \in \cT_n(Q)$.
If $n$ is strictly smaller than the number of vertices of~$Q$, then we define $f(T):=T$, i.e., $f$ is defined to be the identity mapping.
Otherwise we define $b:=\beta_R(r(T))$ and~$r_i:=c_R^{i-1}(r(T))$ for~$i=1,\ldots,b$, i.e., we consider the right branch~$(r_1,\ldots,r_b)$ starting at the root of~$T$.
Since $T$ avoids~$Q$, in every maximal sequence of consecutive non-empty left subtrees~$L(r_i)\neq\varepsilon$, the last one avoids~$Q$ and all earlier ones avoid~$(P,e)$.
The tree~$f(T)$ is obtained from~$T$ by reversing the order of vertices and their left subtrees on the branch~$(r_1,\ldots,r_b)$, and by applying~$f$ recursively to the subtrees avoiding~$Q$ (i.e., the last subtree in each maximal sequence of consecutive non-empty left subtrees~$L(r_i)\neq\varepsilon$ of~$T$).
By construction, the tree~$f(T)$ avoids the tree pattern~$Q'$.
Furthermore, the mapping $f:\cT_n(Q) \to \cT_n(Q')$ is reversible, so it is a bijection.
\end{proof}

\begin{lemma}
\label{lem:Lmove2}
Let $S$ and $S'$ be the paths of length~2 with $\tau(S)=132$ and $\tau(S')=213$, respectively.
Let $Q$ and~$Q'$ denote the tree patterns obtained from the contiguous patterns~$(S,11)$ and~$(S',11)$ by attaching a tree pattern~$(P,e)$ with a non-contiguous edge to the leaf of~$S$ as a right subtree and to the left leaf of~$S'$ as a left subtree, respectively.
Then $Q$ and~$Q'$ are Wilf-equivalent.
\end{lemma}

\begin{figure}[h!]
\includegraphics[page=4]{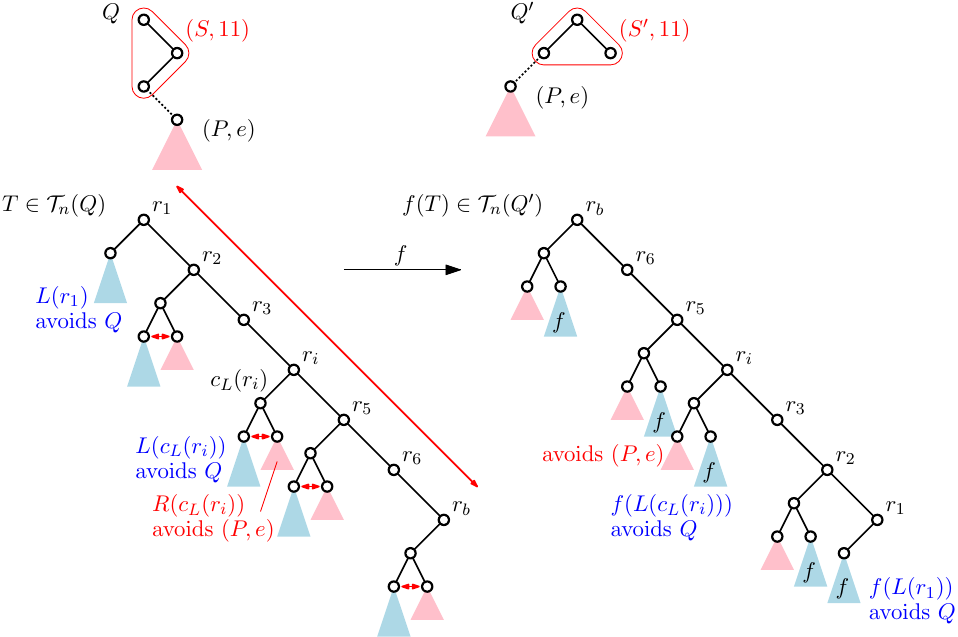}
\caption{Illustration of Lemma~\ref{lem:Lmove2}.}
\label{fig:Lmove2}
\end{figure}

\begin{proof}
The proof is illustrated in Figure~\ref{fig:Lmove2}.
We inductively describe a bijection $f \colon \cT_n(Q) \to \cT_n(Q')$.
Consider a tree~$T \in \cT_n(Q)$.
If $n$ is strictly smaller than the number of vertices of~$Q$, then we define $f(T):=T$, i.e., $f$ is defined to be the identity mapping.
Otherwise we define $b:=\beta_R(r(T))$ and~$r_i:=c_R^{i-1}(r(T))$ for~$i=1,\ldots,b$, i.e., we consider the right branch~$(r_1,\ldots,r_b)$ starting at the root of~$T$.
Since $T$ avoids~$Q$, the subtree $L(r_1)$ avoids~$Q$ and for each $i=2,\dots,b$, if $c_L(r_i)\ne \varepsilon$ then $L(c_L(r_i))$ avoids~$Q$ and $R(c_L(r_i))$ avoids~$(P,e)$.
The tree~$f(T)$ is obtained from~$T$ by reversing the order of vertices and their left subtrees on the branch~$(r_1,\ldots,r_b)$, by swapping the subtrees~$L(c_L(r_i))$ and~$R(c_L(r_i))$ if $c_L(r_i)\ne \varepsilon$ for $i=2,\ldots,b$, and by applying~$f$ recursively to the subtrees avoiding~$Q$ (i.e., the subtree $L(r_1)$ of~$T$ and the subtrees $L(c_L(r_i))$ if $c_L(r_i)\neq \varepsilon$ for $i=2,\ldots,b$).
By construction, the tree~$f(T)$ avoids the tree pattern~$Q'$.
Furthermore, the mapping $f:\cT_n(Q) \to \cT_n(Q')$ is reversible, so it is a bijection.
\end{proof}

\begin{lemma}
\label{lem:Lmove3}
Let $S$ be the path of length~2 with $\tau(S)=132$.
Let $Q$ and~$Q'$ denote the tree patterns obtained from the contiguous pattern~$(S,11)$ by attaching a tree pattern~$(P,e)$ with a non-contiguous edge to the leaf and the middle vertex of~$S$ as a right subtree, respectively.
Then $Q$ and~$Q'$ are Wilf-equivalent.
\end{lemma}

\begin{figure}[h!]
\includegraphics[page=5]{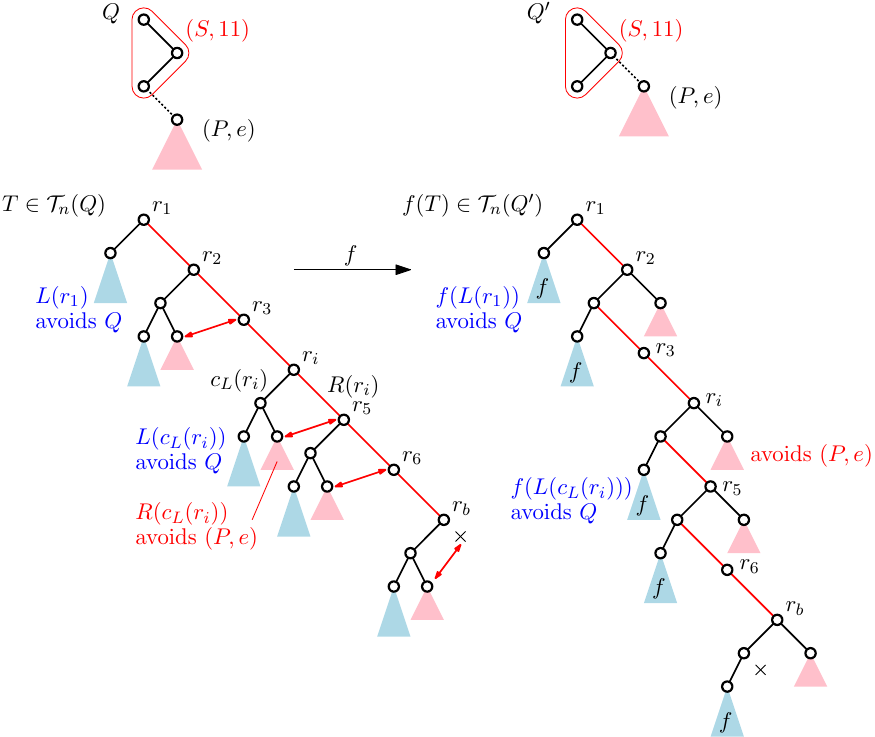}
\caption{Illustration of Lemma~\ref{lem:Lmove3}.}
\label{fig:Lmove3}
\end{figure}

\begin{proof}
The proof is illustrated in Figure~\ref{fig:Lmove3}.
We inductively describe a bijection $f \colon \cT_n(Q) \to \cT_n(Q')$.
Consider a tree~$T \in \cT_n(Q)$.
If $n$ is strictly smaller than the number of vertices of~$Q$, then we define $f(T):=T$, i.e., $f$ is defined to be the identity mapping.
Otherwise we define $b:=\beta_R(r(T))$ and~$r_i:=c_R^{i-1}(r(T))$ for~$i=1,\ldots,b$, i.e., we consider the right branch~$(r_1,\ldots,r_b)$ starting at the root of~$T$.
Since $T$ avoids~$Q$, the subtree $L(r_1)$ avoids~$Q$ and for each $i=2,\dots,b$, if $c_L(r_i)\ne \varepsilon$ then $L(c_L(r_i))$ avoids~$Q$ and $R(c_L(r_i))$ avoids~$(P,e)$.
The tree~$f(T)$ is obtained from~$T$ as follows.
First, we replace $L(r_1)$ recursively by~$f(L(r_1))$.
Next, we consider every $i\in\{2,\dots,b\}$ for which $c_L(r_i)\ne \varepsilon$, we replace $L(c_L(r_i))$ recursively by~$f(L(c_L(r_i)))$, and we swap the subtrees~$R(c_L(r_i))$ and~$R(r_i)$, i.e., $R(c_L(r_i))$ is attached as the right subtree of~$r_i$, and $R(r_i)$ is attached as the right subtree of~$c_L(r_i)$.
Some of these subtrees may be empty~$\varepsilon$, then by attaching an empty tree we mean attaching no tree.
In particular, $R(r_b)$ is empty.
By construction, the tree~$f(T)$ avoids the tree pattern~$Q'$.
Furthermore, the mapping $f:\cT_n(Q) \to \cT_n(Q')$ is reversible, as the path on the vertices~$r_i$ and~$c_L(r_i)$ if they exist, for $i=2,\dots,b$, is uniquely determined in~$f(T)$.
Consequently, $f$ is a bijection, as claimed.
\end{proof}

Clearly, by applying the mirroring operation~$\mu$, we obtain variants of the preceding lemmas where the direction of attachment is interchanged.

\subsection{Specific patterns with 5 vertices}

\begin{lemma}
\label{lem:21435-13254}
The tree patterns~$\cT_n(21435, \hyph1\hyph\hyph)$ and~$\cT_n(13254, 1\hyph1\hyph)$ are Wilf-equivalent.
\end{lemma}

The proof uses the path reversal bijection technique employed in the proofs of Lemmas~\ref{lem:Lmove1} and~\ref{lem:Lmove2}; see Figure~\ref{fig:21435-13254}.
We omit the details.

\begin{figure}[h!]
\includegraphics{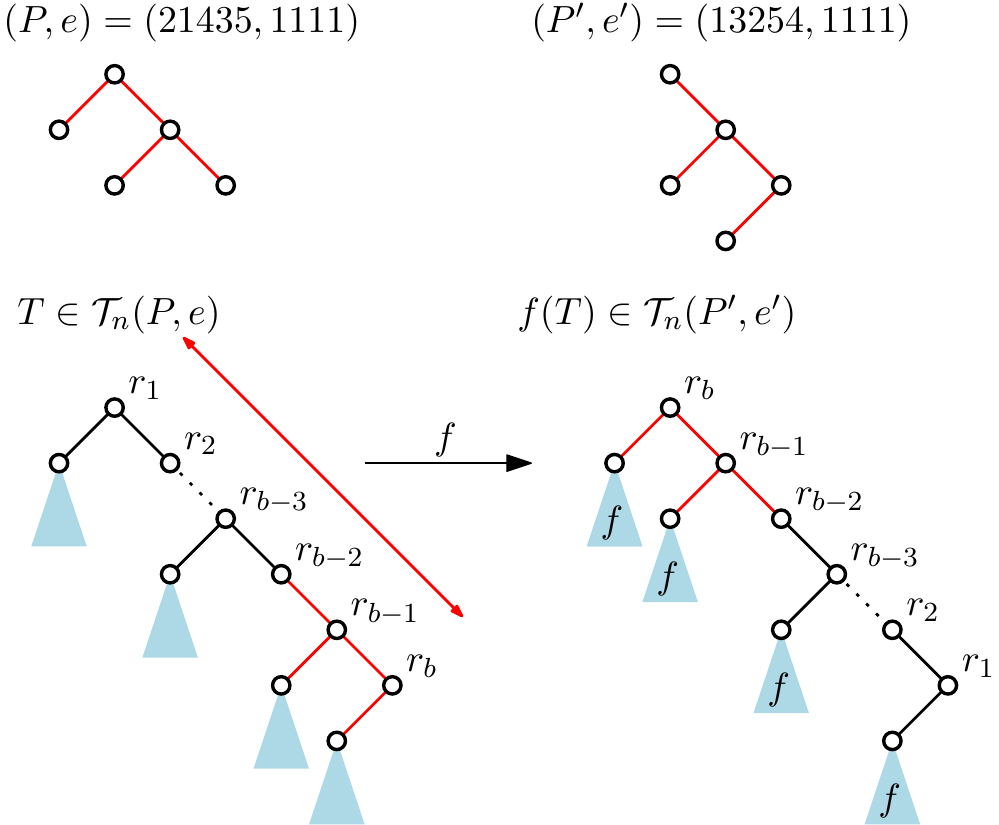}
\caption{Illustration of Lemma~\ref{lem:21435-13254}.}
\label{fig:21435-13254}
\end{figure}

\begin{figure}
\includegraphics{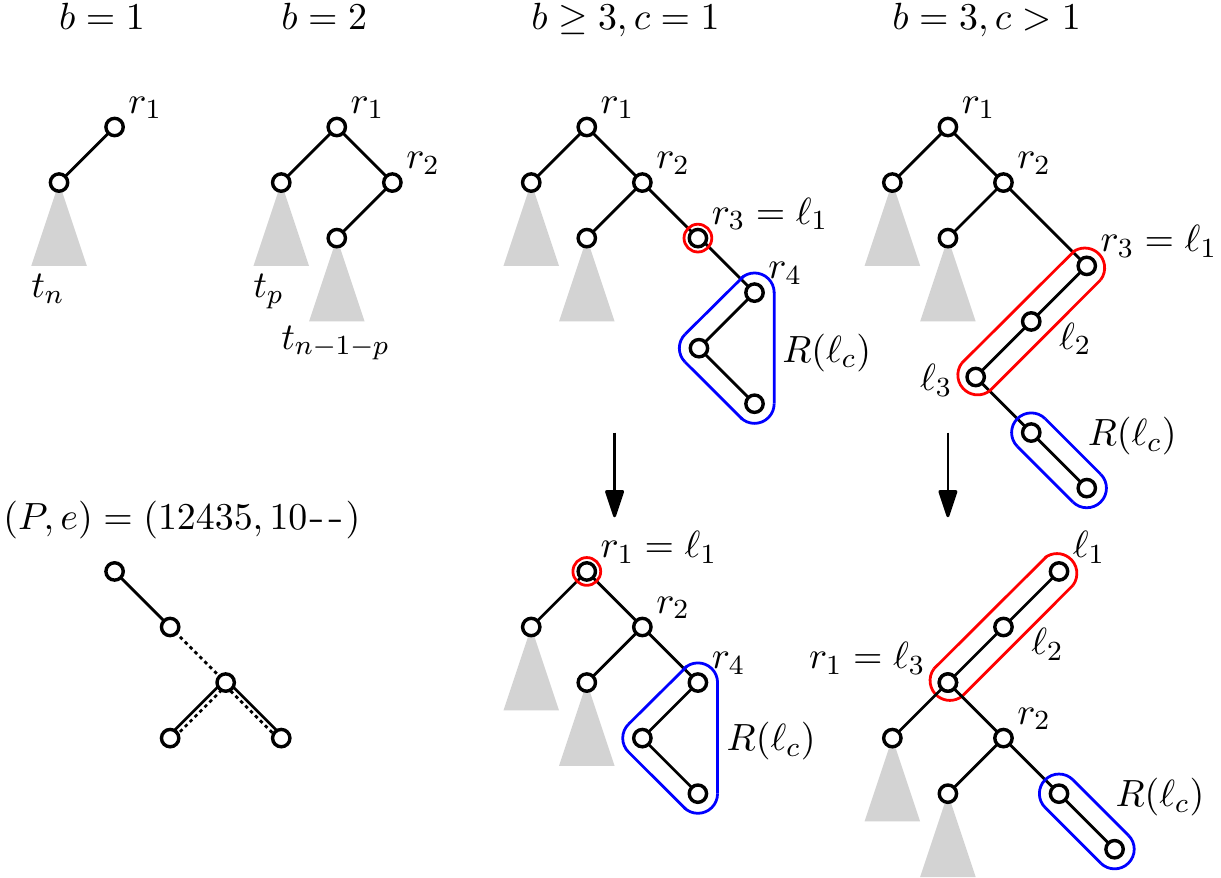}
\caption{Illustration of the proof of Lemma~\ref{lem:12435}.
The rightmost two figures with $b\geq 3$ illustrate the mapping~$f$ in two cases, the first trivial~($c=1$) and the other non-trivial~($c>1$).}
\label{fig:12435}
\end{figure}

\begin{lemma}
\label{lem:12435}
The class of trees $\cT_n(12435,10\hyph\hyph)$ is counted by the sequence OEIS~\href{https://oeis.org/A176677}{A176677}.
\end{lemma}

\begin{proof}
The sequence $(a_i)_{i\ge 1}$ in OEIS~\href{https://oeis.org/A176677}{A176677} is defined by the recurrence $a_0=a_1=1$ and $a_{n+1}=-1+\sum_{p=0}^n a_p a_{n-p}$ for $n\ge 1$.
Using this definition, a straightforward computation shows that $t_n:=a_{n+1}$ satisfies the recursion $t_0=t_1=1$ and
\begin{equation}
\label{eq:tn}
t_{n+1}=2t_n-1+\sum_{p=0}^{n-1}t_p t_{n-1-p}\quad\text{for }n\ge 1.
\end{equation}
To prove the lemma, we show that $t_n=|\cT_n(P,e)|$ for $(P,e):=(12435,10\hyph\hyph)$ by induction.

Clearly, we have $t_0=1$ as $\cT_0(P,e)=\{\varepsilon\}$, which settles the induction basis.
For the induction step, let $n\ge 1$ and consider all trees $T$ from $\cT_{n+1}(P,e)$, distinguished by the number $b:=\beta_R(r(T))$ of vertices on the right branch $(r_1,\dots,r_b)$ starting at the root.
Clearly, there are exactly $t_n$ trees $T$ with $b=1$ as $L(T)\in \cT_n(P,e)$ and $R(T)=\varepsilon$ in such trees.
Furthermore, the number of trees $T$ with $b=2$ is $\sum_{p=0}^{n-1}t_{p}t_{n-1-p}$ as $L(T)=L(r_1)\in \cT_{p}(P,e)$ and $L(r_2)\in \cT_{n-1-p}(P,e)$ in such trees where $p=|L(T)|$ ranges from~$0$ to~$n-1$.

We claim that the number of trees with $b\ge 3$ is $t_n-1$ by mapping them bijectively to all trees in~$\cT_n(P,e)$ except the left path.
Let $T\in \cT_{n+1}(P,e)$ with $b\ge 3$.
Observe that $T(r_3)$ is a path as $T$ avoids~$(P,e)$.
Let $c:=\beta_L(r_3)$ and let $(r_3=\ell_1,\dots,\ell_c)$ denote the left branch in~$T$ starting at~$r_3$.
We map $T\in\cT_{n+1}(P,e)$ to a tree $f(T)\in\cT_n(P,e)$ as follows; see Figure~\ref{fig:12435}.
We start $f(T)$ with the branch $(\ell_1,\dots,\ell_c)$, making $\ell_1$ the root of~$f(T)$, we glue $T(r_1)\setminus T(\ell_1)$ with the vertex~$r_1$ to the vertex~$\ell_c$ (so these two vertices merge into one), and the subtree $R(\ell_c)$ (possibly empty) is attached to~$r_2$ as a right child instead of~$r_3=\ell_1$.
Observe that $f(T) \in \cT_n(P,e)$, the tree $f(T)$ is not the left path, and the mapping~$f$ is reversible.
This completes the inductive proof of~\eqref{eq:tn}.
\end{proof}

\begin{figure}
\includegraphics[page=1]{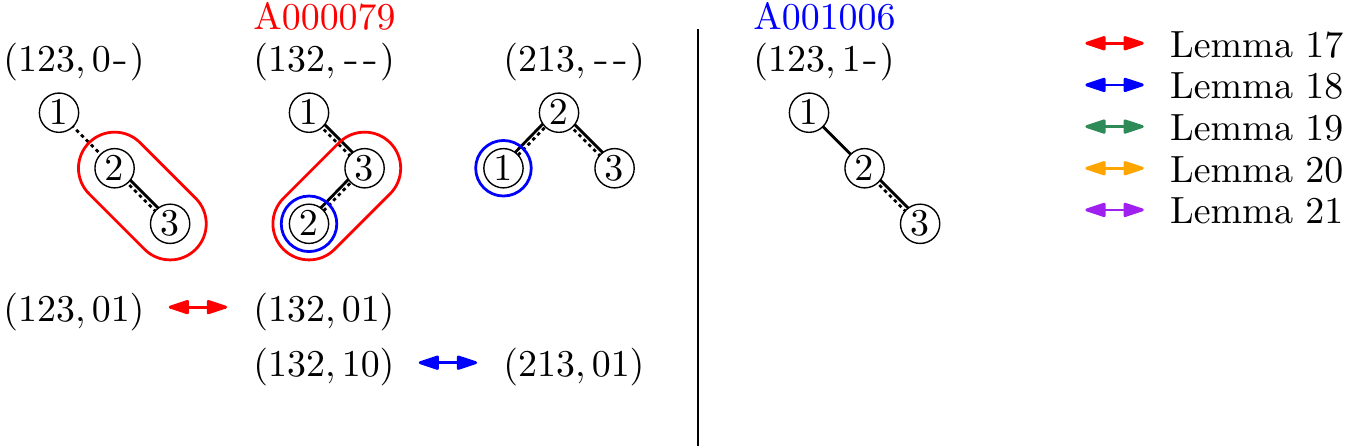}
\caption{Wilf classes of tree patterns on 3~vertices.
Modified subtrees are highlighted in the colors of the corresponding lemmas (see legend).}
\label{fig:wilf3}
\end{figure}

\begin{figure}
\includegraphics[page=2]{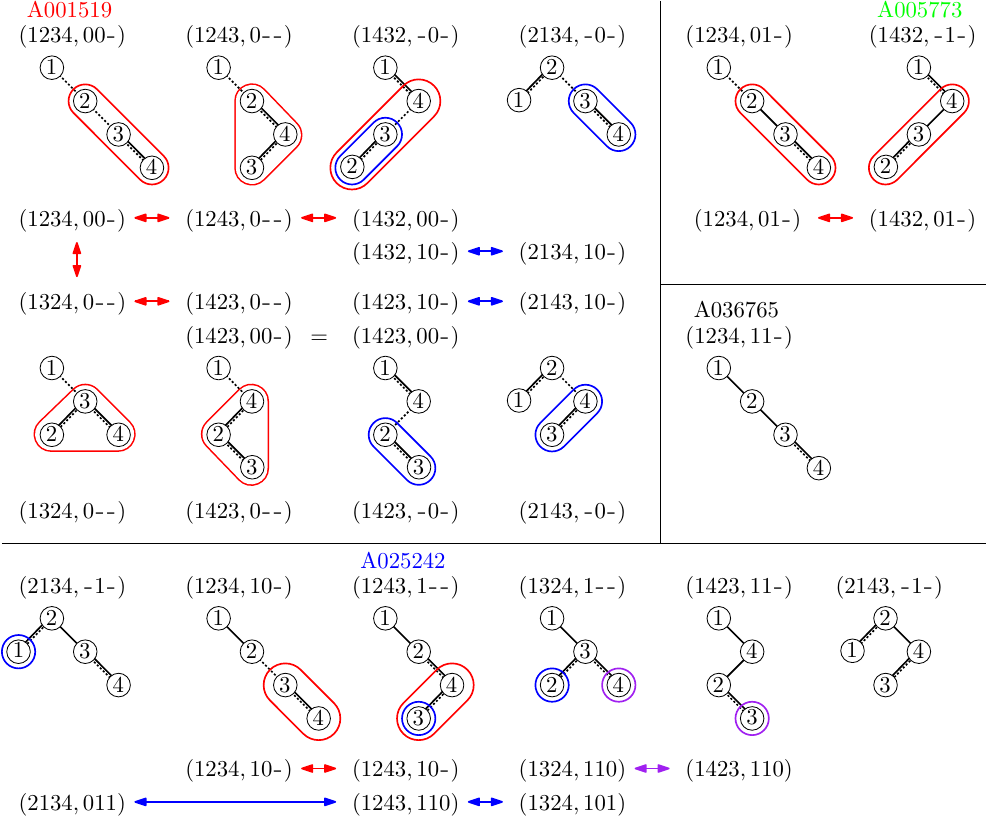}
\caption{Wilf classes of tree patterns on 4 vertices (see legend in Figure~\ref{fig:wilf3}).}
\label{fig:wilf4}
\end{figure}

\begin{figure}
\includegraphics[page=3,scale=0.8]{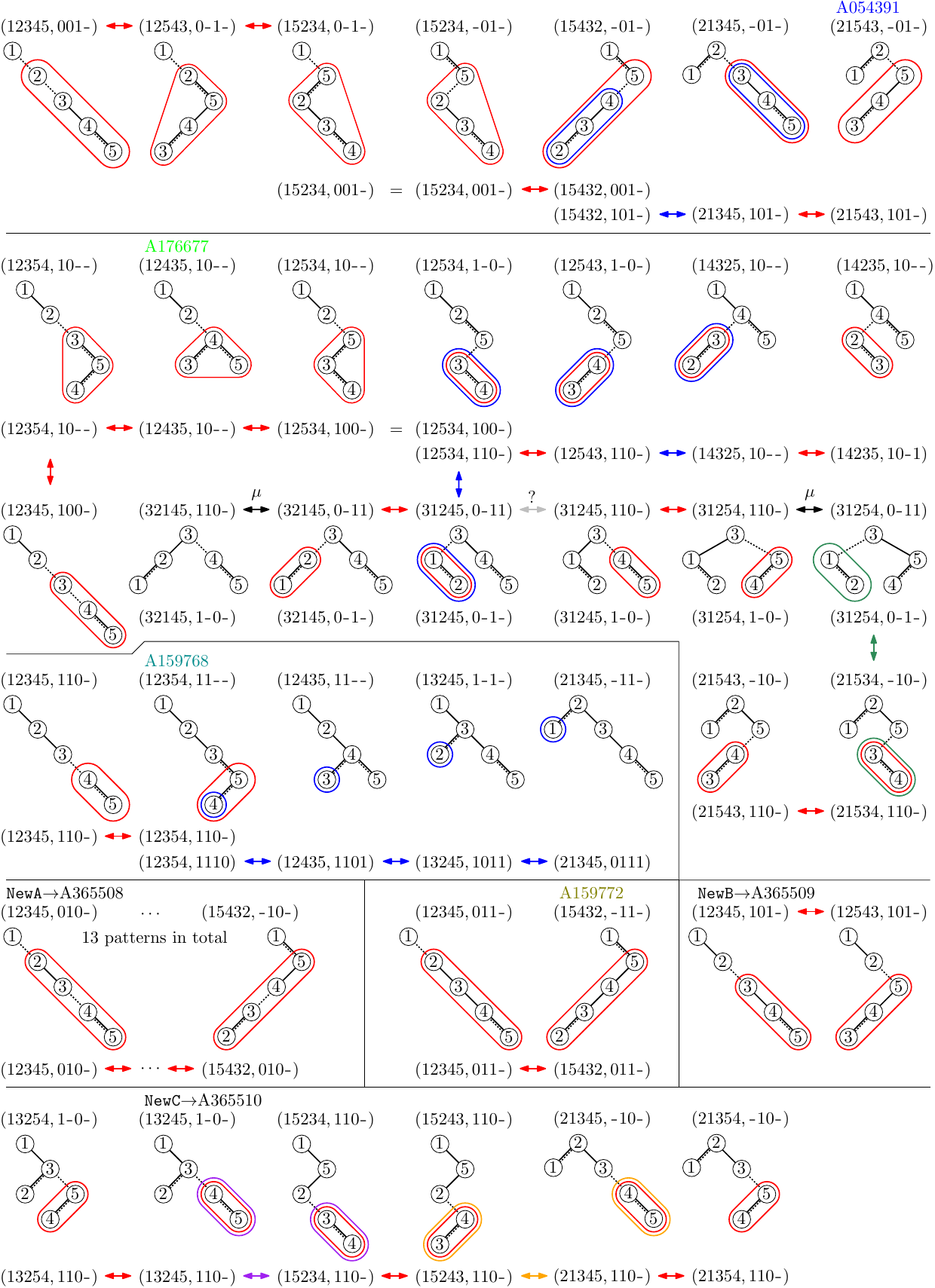}
\caption{Wilf classes of tree patterns on 5 vertices (see legend in Figure~\ref{fig:wilf3}; the question mark denotes an open problem).
Non-contiguous patterns are omitted as they are all counted by OEIS~\href{https://oeis.org/A007051}{A007051} (\cite[Thm.~1]{MR2967227}).
Also contiguous patterns where no lemma applies are omitted.
Only two out of 13 patterns from the class \newa{} are shown.}
\label{fig:wilf5}
\end{figure}

\subsection{Wilf-equivalent patterns with up to 5~vertices}

In this section we apply the lemmas derived in the preceding sections to establish Wilf-equivalences between tree patterns on at most 5 vertices; see Figures~\ref{fig:wilf3}--\ref{fig:wilf5}.
If needed, we use mirrored variants of the lemmas, which is not distinguished in the figures.

It remains an open problem to find a bijection between the tree patterns $(31245,0\hyph1\hyph)$ and $(31245,1\hyph0\hyph)$, or between any of their Wilf-equivalent patterns.
The first class of trees is counted by OEIS~\href{https://oeis.org/A176677}{A176677}, as it is Wilf-equivalent to~$(12435,10\hyph\hyph)$, and then we can use Lemma~\ref{lem:12435}.
For the second class of trees we are missing an argument connecting it to the first class.

\section{Open Problems}
\label{sec:open}

\begin{itemize}[leftmargin=8mm, noitemsep, topsep=1pt plus 1pt]
\item
Are there elegant bijections between pattern-avoiding binary trees and other interesting combinatorial objects such as Motzkin paths with 2-colored $\tF$-steps at odd heights (OEIS~A176677), or so-called skew Motzkin paths (OEIS~A025242)?
For the first family of objects, such a bijection might help to prove Wilf-equivalence between the tree patterns $(31245,0\hyph1\hyph)$ and $(31245,1\hyph0\hyph)$, or between any of their Wilf-equivalent patterns.

\item
For purely contiguous or non-contiguous tree patterns~$(P,e)$, there are recursions to derive the generating function for $|\cT_n(P,e)|$; see~\cite{MR2645188} and \cite{MR2967227}.
For our more general patterns with some contiguous and some non-contiguous edges, these methods seem to fail.
Therefore, it is an interesting open question whether there is an algorithm to compute those more general generating functions, and to understand some of their
properties.
Furthermore, can the set of pattern-avoiding trees for such pure (non-friendly) patterns be generated efficiently?

\item
In addition to contiguous and non-contiguous edges~$(i,p(i))$ of a binary tree pattern, which we encode by $e(i)=1$ and $e(i)=0$, there is another very natural notion of pattern containment that is intermediate between those two, which we may encode by setting~$e(i):=\nicefrac{1}{2}$.
Specifically, for such an edge with $e(i)=\nicefrac{1}{2}$ in the pattern tree~$P$, we require from the injection~$f$ described in Section~\ref{sec:pat-tree} that~$f(i)$ is a descendant of~$f(p(i))$ along a left or right branch in the host tree~$T$.
Specifically, if $i=c_L(p(i))$, then $f(i)=c_L^j(f(p(i)))$ for some $j>0$, whereas if $i=c_R(p(i))$, then $f(i)=c_R^j(f(p(i)))$ for some $j>0$.
Theorem~\ref{thm:bijection} can be generalized to also capture this new notion, by modifying the definition~\eqref{eq:Cip} in the natural way to
\[
C_i':=\begin{cases} \emptyset & \text{if } e(i)=0, \\
\big\{(\rho(i)-1,\min P(i)-1)\} & \text{if } e(i)=\nicefrac{1}{2} \text{ and } i=c_L(p(i)), \\
\big\{(\rho(i)-1,\max P(i))\big\} & \text{if } e(i)=\nicefrac{1}{2} \text{ and } i=c_R(p(i)), \\
\Big\{\big(\rho(i)-1,\min P(i)-1\big),\,\big(\rho(i)-1,\max P(i)\big)\Big\} & \text{if } e(i)=1.
\end{cases}
\]
The notion of friendly tree pattern can be generalized by modifying condition~(iii) in Section~\ref{sec:algo} as follows:
(iii') If~$e(k)\in\{1,\nicefrac{1}{2}\}$, then we have~$e(c_L(k))\in\{0,\nicefrac{1}{2}\}$.
It is worthwhile to investigate this new notion of pattern containment/avoidance and its interplay with the other two notions.
Our computer experiments show that there are patterns with edges~$e(i)=\nicefrac{1}{2}$ that give rise to counting sequences that are distinct from the ones obtained from patterns with edges~$e(i)=1$ (contiguous) and~$e(i)=0$ (non-contiguous).
The corresponding functionality has already been built into our generation tool~\cite{cos_btree}.

This intermediate notion of pattern-avoidance in binary trees has interesting applications in the context of pattern-avoidance in rectangulations, a  line of inquiry that was initiated in~\cite{MR4598046}.
\end{itemize}

\bibliographystyle{alpha}
\bibliography{refs}

\end{document}